\documentclass[a4paper,UKenglish,cleveref, thm-restate]{lipics-v2021}

\pdfoutput=1 
\hideLIPIcs  
\usepackage[dvipsnames,svgnames,table]{xcolor}
\usepackage{url}

\usepackage{listings}
\usepackage{amsmath} 
\DeclareFontFamily{U}{skulls}{}
\DeclareFontShape{U}{skulls}{m}{n}{ <-> skull }{}

\usepackage{float}

\usepackage{multicol} 

\usepackage{amsmath}

\usepackage{graphicx}
\usepackage{xspace}
\usepackage{color} 
\usepackage{tikz-cd}
\usepackage{mathtools}
\usepackage{hyperref}
\usepackage[capitalise]{cleveref}
\usepackage{complexity}
\usepackage[shortcuts]{extdash}
\usepackage{subcaption}
\usepackage{thm-restate}
\usepackage{wrapstuff}
\usepackage{tikz}
\usetikzlibrary{automata, positioning, calc, shapes.geometric}
\usetikzlibrary{decorations.pathreplacing}
\usepackage[textsize=small]{todonotes}

\newcommand{\LTL}{\mathsf{LTL}}

\newcommand{\Aa}{{\mathcal A}}
\newcommand{\Bb}{{\mathcal B}}
\newcommand{\Ac}{{\mathcal A}}
\newcommand{\Bc}{{\mathcal B}}
\newcommand{\Cc}{{\mathcal C}}
\newcommand{\Dd}{{\mathcal D}}

\newcommand{\Hh}{{\mathcal H}}

\newcommand{\Kk}{{\mathcal K}}

\newcommand{\significant}{significant\xspace}
\newcommand{\good}{reach\=/deterministic\xspace}

\newcommand{\reach}{{\mathtt{reach}}}
\newcommand{\reset}[1]{{\color{purple}\mathtt{r}_{#1}}}

\newcommand{\circstate}[1]{\bullet #1} 
\newcommand{\sqrstate}[1]{\text{\tiny$\blacksquare$} #1}
\newcommand{\purpley}{{\color{purple}y}}

\newcommand{\Abkks}{{\mathcal{A}_{\mathtt{BKKS}}}}
\newcommand{\Astrong}{{\mathcal{A}_{\mathtt{strong}}}}
\newcommand{\Lstrong}{{L_{\mathtt{strong}}}}

\newcommand{\Aweak}{{\mathcal{A}_{\mathtt{weak}}}}
\newcommand{\Lweak}{{L_{\mathtt{weak}}}}
\newcommand{\Dstrong}{{\mathcal{D}_{\mathtt{strong}}}}
\newcommand{\Dmain}{{\mathcal{D}_{\mathtt{main}}}}
\newcommand{\Areplace}{{\mathcal{A}_{\mathtt{replace}}}}

\newcommand{\Aclassifier}{{\mathcal{A}_{\mathtt{class}}}}

\newcommand{\Amain}{{\mathcal{A}_{\mathtt{main}}}}
\newcommand{\Qsignificant}{{Q_{\mathtt{signif}}}}

\newcommand{\langBlock}[1]{L_{#1}}
\newcommand{\langInfix}[2]{L_{#1 #2 #2}}
\newcommand{\Lmain}{L_{\mathsf{main}}}

\newcommand{\safe}{{\mathtt{safe}}}
\newcommand{\Cmain}{\mathcal{C}_{{\mathtt{main}}}}
\newcommand{\safeLangA}[2]{L_{\mathtt{safe}}({#1}, {#2})}
\newcommand{\safeLang}[1]{L_{\mathtt{safe}}({#1})}
\newcommand{\ysafelangD}[1]{L_{y\mathtt{safe}}({#1})}
\newcommand{\Theju}{{\mathcal{K}}}
\newcommand{\yTheju}{{\mathcal{K}_y}}

\newcommand{\Sbus}{{\mathcal{S}}^{\mathcal{C}}_{\mathsf{safe}}}
\newcommand{\Dbus}{{\mathcal{S}}^{\mathcal{D}}_{\mathsf{safe}}}

\newcommand{\Fin}[1]{\mathsf{Fin}(#1)}
\newcommand{\Dcore}{\mathcal{D}_{\mathsf{core}}}
\newcommand{\Ptarget}{P_{\mathsf{target}}}
\newcommand{\Qtarget}{Q_{\mathsf{target}}}

\title{History-Deterministic B\"uchi Automata are Succinct}

 \author{Antonio Casares}{Aix-Marseille Université, LIS, France \and \url{https://antonio-casares.github.io/}}{antonio.casares@lis-lab.fr}{https://orcid.org/0000-0002-6539-2020}{A part of project was done when the author was employed by the University of Kaiserslautern-Landau and supported by Deutsche Forschungsgemeinschaft (grant number 522843867).}
 \author{Keya Prakash}{Aix-Marseille Université, CNRS, LIS, France \and \url{https://keya-prakash.github.io/}} {keya.prakash@proton.me}{https://orcid.org/0000-0002-2404-0707}{Funded by EU under MSCA-PF AuRora.}
 \author{K. S. Thejaswini}{Université Libre de Bruxelles, Belgium\and \url{https://thejaswiniraghavan.github.io}}{thejaswini.raghavan@ulb.be}{https://orcid.org/0000-0001-6077-7514}{A part of project was done when the author was employed by Institute of Science and Technology Austria and received funding from the European
 Research Council (ERC), grant agreement No 101020093.}

\authorrunning{A. Casares, K. Prakash and K.S. Thejaswini} 

\Copyright{Antonio Casares, Keya Prakash and K. S. Thejaswini} 



\category{} 

\relatedversion{} 

\nolinenumbers 

\begin{document}

\maketitle

\begin{abstract}
  We describe a history-deterministic Büchi automaton that has fewer states than every language-equivalent deterministic Büchi automaton. This solves a problem that had been open since the introduction of history-determinism and actively investigated for over a decade. 
  
  Our example automaton has 65 states, and proving its succinctness requires the combination of theoretical insights together with the use of solvers for NP-complete problems.
\end{abstract}
\setcounter{tocdepth}{2}
\tableofcontents

\thispagestyle{empty}

\clearpage
\setcounter{page}{1}
\section{Introduction}
\noindent 
Automata~over infinite words are a well-established tool with applications in the verification and synthesis of non-terminating~systems~\cite{Kupferman2018Handbook,EKV21Verification}. A~major bottleneck in these applications is the~exponential cost of determinising automata over infinite~words.
To~circumvent this, recent research has focused on history\=/deterministic automata (HD automata hereafter): a~class of ``mildly'' nondeterministic automata that offer the~algorithmic benefits of determinism without the~full cost of~determinisation.
Formally, an~automaton is history deterministic if there is a~strategy resolving the~nondeterministic choices on-the-fly without guessing the~future.  
HD~automata were introduced by Henzinger and Piterman~\cite{HP06} under the~name of good\=/for\=/games automata, as they are exactly the~class of automata that can be composed with infinite duration games on graphs while preserving the~winner.
This~is the~key property that enables their use in verification and synthesis (see, e.g.,~\cite[page~22]{BL23}).
Equivalent~definitions of history determinism appear in the~work of Kupferman, Safra, and Vardi~\cite{KSV96} 
and were unified by Boker, Kuperberg, Kupferman, and  Skrzypczak~\cite{BKKS13}. Moreover, Colcombet independently defined the concept of history determinism in the~setting of cost~automata~\cite{Col09}.

In this work, we focus on B\"uchi and coB\"uchi~automata. A~B\"uchi~automaton is an~automaton with some of its transitions marked as~\emph{significant}. An infinite word is accepted by a~B\"uchi~automaton if it admits a run that contains infinitely many significant~transitions. CoB\"uchi~automata are the~dual of B\"uchi~automata: a~word is accepted if it admits a run which contains only finitely many significant~transitions.
\begin{example}\label{ex:BKKS13}
\begin{wrapstuff}[r, type=figure, width=5cm, top=1]
\centering
\begin{tikzpicture}[
    base dot/.style={inner sep=0pt, minimum size=5pt, fill=blue!60!black, font=\scriptsize},
    cstate/.style={base dot, circle},
    sstate/.style={base dot, rectangle},
    node distance=1.5cm,
    tight label/.style={auto, inner sep=0.5pt, font=\scriptsize, text=black}
]
\tikzset{every state/.style = {inner sep=-3pt, minimum size=15}}

\node[cstate,fill=red!40, red!40] (q0) at (0,0) {AA};
\node (q0imposter) at (q0) {$I$};
\node[cstate,fill=red!40, red!40] (q1) at (-0.7,-0.8) {AA};
    \node at (q1) {$\iota_a$};
    \node[cstate,fill=red!40, red!40] (q2) at (0.7,-0.8) {AA};
        \node at (q2) {$\iota_b$};
    
    \node[cstate,fill=red!40, red!40] (q3) at (-0.7,-1.6) {AA};
    \node[cstate,fill=red!40, red!40] (q4) at (0.7,-1.6) {AA};
    \node (q3imposter) at (q3) {$p_a$};
    \node (q4imposter) at (q4) {$p_b$};
    \node[cstate,fill=red!40, red!40] (q5) at (-0.7,-2.4) {AA};
    \node[cstate,fill=red!40, red!40] (q6) at (0.7,-2.4) {AA};
    \node[cstate,fill=red!40, red!40] (qI) at (0,-3) {AA};
    \node at (qI) {$I$};
    \path[->] (-0.5,0) edge (q0);

    \path[-stealth]
        (q0) edge [bend right=10] node [left] {$x$} (q1)
        (q0) edge [bend left=10] node [right] {$x$} (q2)
        (q1) edge [bend right=10] node [right,pos=0.3] {$b$} (q0)
        (q2) edge [bend left=10] node [left,pos=0.35] {$a$} (q0)
        (q1) edge [thick, green!50!black, double] node [left] {$a$} (q3)
        (q2) edge [thick, green!50!black, double] node [right] {$b$} (q4)

        (q3) edge node [left] {$x$} (q5)
        (q4) edge node [right] {$x$} (q6)
        
        (q5) edge node [xshift=-0.2cm,yshift=-0.3cm] {$b$} (q4)
        (q6) edge node [xshift=0.2cm,yshift=-0.3cm] {$a$} (q3)
        
        (q5) edge [thick, green!50!black, double] node [left] {$a$} (qI)
        (q6) edge [thick, green!50!black, double] node [right] {$b$} (qI)
        ;
\end{tikzpicture}
\caption{An HD B\"uchi automaton $\Abkks$ that is not determinisable by pruning. States labelled $I$ are identified as the same state. Double green arrows represent significant transitions.}
\label{fig:BKKS13}
\end{wrapstuff} The~automaton $\Abkks$ described in~\cref{fig:BKKS13}---a~version of~\cite[Fig.2]{BKKS13}---accepts the~language of~words of the~form $(xa+xb)^\omega$ that contain at least one of the~strings $xaxa$ or $xbxb$ infinitely~often. This~automaton is~HD.  
 The~only nondeterministic choice is on~$I$ over the~letter~$x$. A~strategy for resolving this nondeterminism is simple: we alternate between choosing the~transitions $I\xrightarrow{x}\iota_a$ and~$I\xrightarrow{x}\iota_b$. If the~input~word is in the~language, it must eventually contain $xaxa$ or $xbxb$, allowing us to progress to states $p_a$ or $p_b$ and visit an accepting~transition (infinitely often).
\end{example}
 HD B\"uchi (resp. coB\"uchi) automata have the same expressive power as their deterministic counterpart. They recognise a strict subset of all $\omega$-regular languages, recognised, for instance, by deterministic parity automata.
 However, HD automata constitute an~alternative for the~use of deterministic automata in various~applications.
In~procedures such~as $\LTL$\=/synthesis, the main bottleneck is the size of the resulting~deterministic automata~\cite{ERS2024EfficientLTL}.
One~motivation for the use of HD automata is the potential reduction in~size.
For~this~reason, a~central theme in the~study of HD~automata is~\emph{succinctness}: can HD~automata be smaller than their deterministic~counterparts?

 For~coB\"uchi~automata, the~landscape is well\=/understood. HD~coB\"uchi~automata can be exponentially~smaller than the~smallest language\=/equivalent deterministic~ones~\cite{KS15}. Moreover, Abu~Radi and Kupferman showed that HD~coB\"uchi~automata (with transition\=/based acceptance) can be minimised in polynomial~time, leading to a canonical~form for languages recognised by 
coB\"uchi~automata~\cite{AK22}. These theoretical~results have subsequently enabled canonical~representations of $\omega$\=/regular~languages~\cite{ES22,Ehlers25Rerailing,CLW26Layered}, algorithms for learning~HD~coB\"uchi automata~\cite{LW25}, and practical implementation~efforts~\cite{EK24a,EK24b}.

In~contrast, the~situation for B\"uchi~automata has remained elusive. While~it~is known that every~HD~B\"uchi~automaton admits a language\=/equivalent deterministic~B\"uchi automaton of quadratic~size~\cite{KS15}, it is unknown whether this~bound could be matched or~improved. 
Actually, it was not even known if every~HD~B\"uchi~automaton admits an~equivalent deterministic~B\"uchi~automaton of the~same~size.
In~this~work, we tackle the~following question.

\begin{quote} 
\centering 
\textbf{HD-Büchi succinctness problem.} \emph{Is~there an HD Büchi automaton that is strictly smaller than every~language\=/equivalent deterministic Büchi automaton?} 
\end{quote}

In~fact, for several years since Henzinger and Piterman's introduction of HD~automata in~2006~\cite{HP06}, it was not known whether HD~B\"uchi~automata could simply be determinised by pruning transitions (as is the case for HD~automata over finite~words~\cite[Proposition 14]{Col12}).
This~was settled in the negative by Boker et al.~\cite{BKKS13} by the $7$\=/state automaton example in~\cref{fig:BKKS13}. 
However, there is an~equivalent deterministic~automaton with only $4$~states (see~\cref{fig:detForBKKS13a}). The~computational~complexity of the~determinisation of HD~B\"uchi automata was only recently proved to be in polynomial~time by two~independent~works~\cite{AJP24,AKL24}.

In~recent years, the~research~community working on HD~automata has been very~active, and recent~results---for example the~works by Abu~Radi and Kupferman~\cite{AK22} or by Lehtinen and Prakash~\cite{LP25}---have gained a lot of insight on the~structure of HD~automata. 
The~fact that this~question remained open points to a~gap of our understanding of the~structure of HD and deterministic~B\"uchi~automata. 

Several~works raise the~above question on the~succinctness of HD~Büchi~automata. A~conjecture stating that all HD~parity~automata (and in particular HD~B\"uchi~automata) are determinisable~by~pruning was stated by Colcombet in~2012~\cite[Conj.~8]{Col12}.
The~succinctness question was asked in 2013~\cite{BKKS13} (for HD automata with various acceptance conditions) and solved in~2015 for coB\"uchi~automata~\cite{KS15}. The succinctness question for Büchi automata is discussed in various recent works~\cite[p.30]{AK22}, \cite[pages~3 and 16]{AJP24}, and \cite[p.69]{AKL24}.
For~the~subclass of Muller~languages (languages that can be described by a~Boolean~condition on the~letters that appear infinitely~often), Casares et al.\ showed that HD~parity and B\"uchi~automata are not more succinct than deterministic~ones~\cite[Cor.~4.16]{CCFL24FromMtoP}, while Casares, Colcombet, and Lehtinen showed that HD~Rabin~automata can be exponentially~smaller than deterministic~ones~\cite[Thm.~21]{CCL22SizeGFG}. In her thesis, Prakash also stated a~conjecture implying that every~HD~B\"uchi~automaton admits an~equivalent deterministic Büchi automaton of the same size~\cite[Conj.\ 2]{Pra25} (see \cref{conj:weak-rewiring}).

In this paper, we disprove this conjecture by answering the succinctness question positively.
\begin{theorem}\label{thm:main}
    There is a~history\=/deterministic B\"uchi~automaton that has strictly fewer~states than every language\=/equivalent deterministic B\"uchi automaton.
\end{theorem}

Concretely, we provide an HD~Büchi~automaton with 65 states and prove that every language\=/equivalent deterministic~automaton requires at least 66 states (\cref{fig:wholeautomatonmainpaper}).

Verifying this gap is non-trivial. Although tools like SPOT exist for minimising deterministic transition-based Büchi~automata~\cite{spot2.1}, they use SAT solvers. 
Our attempt to feed the~full 66\=/state~deterministic~automaton to such a solver was perhaps overly optimistic; the~resulting CNF formula grew large enough to exhaust memory long before it exhausted the~search space and, to the~best of our knowledge, the~solver is still pondering over the existence of a 65\=/state deterministic automaton.
Furthermore, the~minimisation of deterministic Büchi automata is an $\NP$-complete problem for both state-based acceptance~\cite{Sch10} or transition-based acceptance~\cite{AE25}.

Hence, we construct our witness automaton in such a way that most of the~states of the~automaton are essential and cannot be minimised. This is done by leveraging structural~insights regarding HD~coBüchi~automata, drawing primarily on the~work of Abu Radi and Kupferman~\cite{AK22}. We use their characterisation of minimality for HD~coBüchi~automata to reason about the~minimality of a deterministic coBüchi automaton recognising the~complement~language.

Even with these theoretical reductions, we are left to prove that the smallest DFA separating two specific languages on finite words requires at least 5 states. To establish this lower bound, we employ the learner tool DFAMiner~\cite{DLS24}. This~tool formulates the~existence of a~4\=/state separator as a~SAT~instance and then proves its unsatisfiability, thereby making our~proof computer\=/aided. 

\subparagraph{Declaration concerning the use of AI and LLMs.} We declare that we only use the aid of computers to prove \cref{lemma:simpler-comp} (which is used to prove \cref{lemma:compy}). The generation of our 65-state example was done entirely by hand. In particular, no AI/ LLM was used towards the writing and research for this paper. 

\subparagraph{Organisation.} In the next section, we provide some preliminaries on history-deterministic automata. The body of the paper is then split in two parts: a general overview (\Cref{sec:overview}) presenting the main ideas, and a technical body (\Cref{sec:full-details}) with full details for the constructions and proofs outlined in \Cref{sec:overview}. 
These two parts follow the same structure in three subsections; each subsection from \Cref{sec:full-details} extends and provides full details of the corresponding subsection in \Cref{sec:overview}.
We end the paper with conclusions in \Cref{sec:conclusions}. 

\section{Preliminaries}
Before we proceed to the technical overview, we first provide some vocabulary that helps us describe the results. 
For a natural number $i$, we use $[i]$ to denote the set $\{1,2,\dots,i\}$. 
An \emph{alphabet}, often denoted by $\Sigma$, is a finite set. 
\subparagraph{Automata.} 
A~nondeterministic B\"uchi (resp.\ coB\"uchi)~automaton $\Aa=(Q,\Sigma,\Delta,q_0,F)$ consists of a~finite 
set of \emph{states} $Q$, an~\emph{initial state} $q_0\in Q$, a set of  \emph{transitions} $\Delta \subseteq Q \times \Sigma \times Q$ and a subset $F \subseteq \Delta$ of \emph{\significant} transitions. We denote  transitions $(p,a,q)$ by $p\xrightarrow{a} q$, and  transitions $(p,a,q)$ that are \significant by $p\xRightarrow{a} q$. 

A \emph{run} of $\Aa$ over an infinite word $w$ in $\Sigma^{\omega}$ is an infinite path $\rho$ of the automaton starting at~$q_0$ such that the~edges of $\rho$ are labelled by the letters of $w$ in sequence. If $\Aa$~is a~Büchi~automaton, then $\rho$ is an \emph{accepting run} if $\rho$ contains infinitely many \significant transitions. Dually, if $\Aa$ is a coB\"uchi automaton, then $\rho$ is an \emph{accepting run} if $\rho$ contains finitely many \significant transitions. Infinite runs that are not accepting are dubbed \emph{rejecting}.

In the rest of the paper, the word automaton means either a nondeterministic Büchi automaton or a nondeterministic coB\"uchi automaton, unless mentioned otherwise. 
A~\emph{subautomaton} of $\Aa$ is an automaton obtained by removing some transitions of $\Aa$.

A word $w$ is \emph{accepted} by an automaton $\Aa$ if there is an accepting run of $\Aa$ on $w$, and \emph{rejected} otherwise. The \emph{language} of $\Aa$ is the set of words that are accepted by $\Aa$, which we denote as $L(\Aa)$. 
Two automata $\Aa$ and $\Bb$ are \emph{language-equivalent} if they have the~same~language. 

For a state $q$ in $\Aa$, we write $(\Aa,q)$ to denote the automaton where we set $q$ to be the initial state. We say that two states $p,q$ of $\Aa$ are \emph{language-equivalent} if $L(\Aa,p)=L(\Aa,q)$.

We say that an automaton is \emph{deterministic} if, from every state and for every letter, there is at most one outgoing transition from that state on that letter. 

\subparagraph{History determinism.}
A \emph{resolver} for an automaton $\Aa$ is a function of type $\Delta^*\times\Sigma\to{\Delta}$ that takes as input a finite run in $\Aa$ and a letter $a$, and outputs a transition on $a$ from the endpoint of the run. This way, a \emph{resolver} \emph{induces} a run on every word in $\Sigma^*$. We say that a resolver for $\Aa$ is an~\emph{HD\=/resolver} if it induces an accepting run on every word in the language of $\Aa$. We say that an automaton is \emph{history deterministic}, HD for short, if there is an HD resolver for it.

We refer to Example~\ref{ex:BKKS13} for an example of an HD automaton. We next give a non-example.
\begin{example}
    \begin{wrapstuff}[r, type=figure, width=3.6cm, top=1]
\centering
        \begin{tikzpicture}[
    base dot/.style={
        inner sep=0pt,      
        minimum size=5pt,   
        fill=blue!60!black, 
        font=\scriptsize    
    },
    cstate/.style={base dot, circle},    
    sstate/.style={base dot, rectangle}, 
    node distance=1cm,
    tight label/.style={auto, inner sep=0.5pt, font=\scriptsize, text=black}
]
        \tikzset{every state/.style = {inner sep=-3pt,minimum size =15}}

    \node[cstate,fill=red!40, red!40] (q3) at (0,0) {AA};
    \node[cstate,fill=red!40, red!40] (q4) at (1.5,0) {AA};
    \node[cstate,fill=red!40, red!40] (q5) at (3,0) {AA};
    \node at (q3) {$A$};
    \node at (q4) {$I$};
    \node at (q5) {$B$};
    \path[->] (1.5,-0.5) edge (q4);

    \path[->,>=stealth,shorten >=1pt]
        (q3) edge [bend left = 60, thick, green!50!black, double] node [above] {$a$} (q4)
        (q3) edge [bend right = 60] node [above] {$b$} (q4)
        (q4) edge node [above] {$a,b$} (q3)
        (q4) edge node [above] {$a,b$} (q5) 
        (q5) edge [bend left = 60] node [above] {$a$} (q4)
        (q5) edge [bend right = 60, thick, green!50!black, double] node [above] {$b$} (q4)
        ;
    \end{tikzpicture}
    \caption{A B\"uchi automaton that is not HD.}
    \label{fig:ND-not-HD}
\end{wrapstuff} In Figure~\ref{fig:ND-not-HD} we show a B\"uchi automaton that is not history deterministic.
    This automaton accepts every word over the alphabet $\{a,b\}$.
    However, to resolve the nondeterminism from $I$ so as to see a \significant transition next, the resolver needs to guess whether the next letter is going to be $a$ or $b$. For every resolver for the automaton, there is a word on which the run constructed by the resolver contains only non-\significant transitions, and thus that run is rejecting.
\end{example}
\vspace{-0.5cm}
\paragraph*{Structure of HD B\"uchi automata}
Kuperberg and Skrzypczak's quadratic determinisation construction of HD Büchi automata~\cite{KS15} first involved proving that every HD Büchi automaton can be simplified so that it satisfies certain properties. Their analysis was then refined later by Acharya, Jurdzi\'nski, and Prakash~\cite{AJP24}; we will use the presentation from Prakash's PhD thesis~\cite{Pra25} in this~work. We begin by defining notions that are useful towards this. 

\subparagraph{Simulation relation.}
Simulation is a relation between two automata that is coarser than language inclusion, that is, if $\Bc$ simulates $\Ac$, then $L(\Ac) \subseteq L(\Bc)$ \cite{HKR02}. Concretely, the simulation relation can be described as a game between two players Adam and Eve, called the simulation game. 
The \emph{simulation game of automaton $\Ac$ by automaton $\Bc$} starts with an~Eve’s~token at the initial state of the automaton $\Bc$ and Adam’s token at the initial state of~$\Ac$. In each round,
\begin{enumerate}
    \item Adam selects a~letter and moves his~token along a~transition on that letter in~$\Ac$,
    \item Eve moves her~token along a transition on that letter in automaton $\Bc$.
\end{enumerate} 
Consequently, in the limit of such a~play, Adam constructs a run on his token in $\Ac$, and Eve constructs a run on her token in $\Bc$,
both on the same word chosen by Adam. Eve \emph{wins} if either Eve's run is accepting or Adam’s run is rejecting. If Eve has a winning strategy in this game, then we say that automaton $\Bc$ \emph{simulates} automaton $\Ac$.
For two states $p$ and $q$ in an automaton $\Ac$, we will often write $q$ simulates $p$ in $\Ac$ to say that $(\Ac,q)$ simulates $(\Ac,p)$. 

\subparagraph{Reach-covering.}
A \emph{reachability automaton} is a Büchi automaton where all of its \significant transitions lead to an \emph{accepting sink state} that has a \significant self-loop on every letter. 

For every B\"uchi automaton $\Aa$, we use $\reach(\Aa)$ to denote the reachability automaton that is obtained by redirecting all the \significant transitions of $\Aa$ to an accepting sink-state, and by not modifying all the transitions that are not \significant. 

We say that $q$ is a \emph{\good state} in $\Aa$ if the accessible portion of $(\reach(\Ac),q)$ is \emph{deterministic}, i.e., every word $w$ has an unique run from $q$ in $\reach(\Ac)$. 

We say that $\Aa$ has \emph{reach\=/covering} if for every state $p$ in $\Aa$, there is a \good state $q$ of $\Aa$ such that the following two conditions hold.
\begin{enumerate}
    \item $p$ and $q$ are language-equivalent in $\Aa$, and
    \item $p$ simulates $q$ in $\reach(\Aa)$.
\end{enumerate}
Note that, since \good states simulate themselves in $\reach(\Aa)$, $\Aa$ has reach\=/covering if and only if the above conditions hold for every state $p$ that is not \good.

\subparagraph{Semantic determinism.} We say that an automaton $\Aa$ is \emph{semantically deterministic}, SD for short, if outgoing transitions on the same letter from language-equivalent states lead to language-equivalent states. That is, if $p,q$ are language-equivalent states in $\Aa$ and $p\xrightarrow{a}p', q\xrightarrow{a} q'$ are  transitions on some letter $a$, then $p'$ and $q'$ are also language\=/equivalent. 

\subparagraph{Normal form.} 
Given a Büchi automaton $\Aa$, consider the graph consisting only of non-significant transitions of $\Aa$. The automaton $\Aa$ is \emph{normal} if every non-significant transition is in some SCC of this graph. We note that every Büchi automaton can be converted into a language-equivalent normal automaton by making \significant all  
non-significant transitions of $\Aa$ changing of SCC; this transformation does not introduce any accepting cycles, and thus, does not increase the language of the automaton.
\subparagraph{Simplified Büchi automata.}
We say that a Büchi automaton $\Aa$ is \emph{simplified} if $\Aa$ is semantically deterministic, normal, and has \emph{reach-covering}. 
Acharya, Jurdzi\'nski, and Prakash showed that 
every HD Büchi automaton can be \emph{simplified} in polynomial time without increasing the number of states~\cite{AJP24} (see also \cite[Lemma 5.8]{Pra25}). 
We will use their result that every simplified Büchi automaton is history deterministic as well, in order to prove history determinism of various Büchi automata.

\begin{lemma}[{\cite[Lemma 5.7]{Pra25}}]\label[lemma]{lemma:buchi-simplified-implies-hd}
    Every simplified Büchi automaton is history deterministic.
\end{lemma}

\paragraph*{Interpretation of figures} Throughout the paper, we include many figures for automata, some of them representing automata with more states than we would like. To represent these automata, we use some shorthands.
First, a state may appear several times in a figure; the final automaton is the result of merging the states with the same name. When many states have the same target over the same letter $\sigma$, we may surround these states by a box, and use transitions outgoing from the box. For example, a transition on $\sigma$ from a box to a state $q$ means that from all states of the box, the $\sigma$-transition goes to $q$.
As in the text, double green arrows represent significant transitions for B\"uchi automata. (For coB\"uchi automata, significant transitions will be double arrows, but red. These rarely occur.)
Sometimes we use different colours and shapes or sizes for states; this is mostly for aesthetic purposes (in Figures~\ref{fig:strong-but-no-weak},~\ref{fig:weak-rewiring-counterexample}, and~\ref{fig:wierdYdet}, colour differentiation also serves to identify language-equivalent states).

\renewcommand{\thesection}{\Roman{section}}
\setcounter{section}{0}

\section{General overview}\label{sec:overview}

This section is a technical overview of our contribution. We outline the main constructions and the ideas of the proofs. We start in Section~\ref{subsec:ov-examples} by giving  some~intuition and first steps towards the~construction of the 65-state HD-automaton in~Section~\ref{subsec:ov-main}. 
Then, in Section~\ref{subsec:ov-succinctness} we sketch the proof of succinctness for Theorem~\ref{thm:main}, stating that there is no equivalent deterministic automaton with only 65~states. 
We provide full details in \Cref{sec:full-details}.

\subsection{Motivating examples}\label{subsec:ov-examples}

In this section, we analyse why na\"ive determinisation techniques fail. We provide increasingly complex counterexamples to specific ``straw man''  determinisation techniques, which then help us build the intuition required for our main result. These examples demonstrate the pitfalls that our final construction must avoid. 
Although the definition of the 65-state automaton can be understood directly by skipping ahead to \cref{subsec:ov-main} (or \cref{subsec:full-main}), we strongly recommend reading this section first to understand the method behind the madness of the final construction.

\subparagraph{A motivating example.}
We start with the example automaton $\Abkks$ from Figure~\ref{fig:BKKS13}. 
We show two deterministic Büchi automata that are language-equivalent to $\Abkks$ and that have 
at most as many states as $\Abkks$ in \cref{fig:sec3-bkks-determinisationsFIRST}.


\begin{wrapstuff}[r, type=figure, width=8cm, top=1]
\centering
\captionsetup{font=footnotesize,skip=2pt}
\begin{subfigure}[t]{.20\textwidth}
  \centering
\begin{tikzpicture}[
    base dot/.style={
        inner sep=0pt,
        minimum size=5pt,
        fill=blue!60!black,
        font=\scriptsize
    },
    cstate/.style={base dot, circle},
    sstate/.style={base dot, rectangle},
    node distance=1.5cm,
    tight label/.style={auto, inner sep=0.5pt, font=\scriptsize, text=black}
]
    \node[cstate,fill=red!40, red!40] (q3) at (-0.7,-1.4) {AA};
    \node[cstate,fill=red!40, red!40] (q4) at (0.7,-1.4) {AA};
    \node[cstate,fill=red!40, red!40] (q5) at (-0.7,-2.2) {AA};
    \node[cstate,fill=red!40, red!40] (q6) at (0.7,-2.2) {AA};
    \node[cstate,fill=red!40, red!40] (qIa) at (q3) {AA};
    \node[cstate,fill=red!40, red!40] (qIb) at (q4) {AA};
    \path[->] (-0.7,-1) edge (q3);
    \path[-stealth]
        (q3) edge node [right] {$x$} (q5)
        (q4) edge node [left] {$x$} (q6)
        (q5) edge node [xshift=-0.2cm,yshift=-0.3cm] {$b$} (q4)
        (q6) edge node [xshift=0.2cm,yshift=-0.3cm] {$a$} (q3)
        (q5) edge [bend left, thick, green!50!black, double] node [left] {$a$} (qIa)
        (q6) edge [bend right, thick, green!50!black, double] node [right] {$b$} (qIb);
\node[] (space) at (0,-3) {};
    \end{tikzpicture}
\caption{Rewiring}\label{fig:detForBKKS13aFIRST}\label{fig:detForBKKS13a}
\end{subfigure}%
\hspace{0.7cm}
\begin{subfigure}[t]{.55\textwidth}
  \centering
  \begin{tikzpicture}[
    base dot/.style={
        inner sep=0pt,
        minimum size=5pt,
        fill=blue!60!black,
        font=\scriptsize
    },
    cstate/.style={base dot, circle},
    sstate/.style={base dot, rectangle},
    node distance=1.5cm,
    tight label/.style={auto, inner sep=0.5pt, font=\scriptsize, text=black}
]
    \node[cstate,fill=red!40, red!40] (q0) at (0,0) {AAA};
    \node at (q0) {$I$};
    \path[->] (-0.5,0) edge (q0);
     \node[cstate,fill=red!40, red!40] (q1) at (0,-0.8) {AA};
     \node[cstate,fill=red!40, red!40] (q3) at (-0.7,-1.4) {AA};
     \node[cstate,fill=red!40, red!40] (q4) at (0.7,-1.4) {AA};
     \node[cstate,fill=red!40, red!40] (q5) at (-0.7,-2.2) {AA};
     \node[cstate,fill=red!40, red!40] (q6) at (0.7,-2.2) {AA};
     \node[cstate,fill=red!40, red!40] (qI) at (0,-2.8) {AAA};
     \node at (qI) {$I$};
    \path[-stealth]
        (q0) edge node [right] {$x$} (q1)
        (q1) edge [thick, green!50!black, double] node [above] {$a$} (q3)
        (q1) edge [thick, green!50!black, double] node [above] {$b$} (q4)
        (q3) edge node [left] {$x$} (q5)
        (q4) edge node [right] {$x$} (q6)
        (q5) edge node [xshift=-0.2cm,yshift=-0.3cm] {$b$} (q4)
        (q6) edge node [xshift=0.2cm,yshift=-0.3cm] {$a$} (q3)
        (q5) edge [thick, green!50!black, double] node [left,pos=0.8,inner sep=1.5mm] {$a$} (qI)
        (q6) edge [thick, green!50!black, double] node [right,pos=0.9,inner sep=2mm] {$b$} (qI);
    \end{tikzpicture}
\caption{Replacing the nondeterministic component}\label{fig:detForBKKS13bFIRST}\label{fig:detForBKKS13b}
\end{subfigure}
\caption{Two language-equivalent deterministic Büchi automata to the automaton from Figure~\ref{fig:BKKS13}.}
\label{fig:sec3-bkks-determinisationsFIRST}
\end{wrapstuff} 
First, observe that the $4$ bottom states of $\Abkks$ are \good: this subautomaton is deterministic until reaching a \significant transition.
The states $I,\iota_a$ and $\iota_b$ are not, 
we refer to these $3$ states as the ``nondeterministic component'' of $\Abkks$.

The automaton in \cref{fig:detForBKKS13aFIRST} is obtained by \emph{rewiring}: keeping the deterministic component and redirecting significant transitions. 
This is formally defined below.
The second automaton, on the other hand, is obtained by replacing the `nondeterministic component' of $\Aa$ with a deterministic part  with no more states. 

Below, we describe automata where these and other determinisation techniques do not work (without strictly increasing the number of states).
We then use these examples and insights  to motivate our succinct HD Büchi automaton in \cref{subsec:ov-main}.

\subsubsection{The rewiring conjectures}\label{subsec:ov-rewiring}
For a simplified Büchi automaton $\Aa$, a \emph{rewiring} of $\Aa$ is a deterministic Büchi automaton~$\Bb$ that satisfies the following three conditions:
\wrapstuffclear
\begin{enumerate}
    \item The states of $\Bb$ are the \good states of $\Aa$, and the non-significant transitions of $\Bb$ are the same as the non-significant transitions that are outgoing from the \good states of $\Aa$.
    \item The initial state of $\Bb$ is set to some \good state of $\Aa$ such that it simulates the initial state of $\Aa$ in $\reach(\Aa)$.
    \item The significant transitions of $\Bb$ are of the form $p\xRightarrow{a}q$, such that $L(\Aa,q)=a^{-1}L(\Aa,p)$.
\end{enumerate}
In other words, $\Bb$ is obtained by redirecting the significant transitions from \good states of $\Aa$ to other \good states of $\Aa$, so that only \good states are reached. 
In particular, $\Bb$ has no more states than $\Aa$.
Prakash~\cite[Conjecture 2]{Pra25} thus proposed  what we call the \emph{weak\=/rewiring conjecture}. 
\begin{conjecture}[Weak-rewiring conjecture]\label[conjecture]{conj:weak-rewiring}
Every simplified HD Büchi automaton has a language-equivalent rewiring.
\end{conjecture}

Towards disproving the weak\=/rewiring conjecture, we first disprove a stronger version of the~conjecture.
The strong\=/rewiring conjecture states that we can redirect significant transitions towards \good states,  while preserving the language, one-by-one and in an arbitrary order. Whereas in the weaker conjecture, we are allowed to redirect all significant transitions at once, which may be necessary to preserve language-equivalence.

\begin{conjecture}[Strong-rewiring conjecture]\label[conjecture]{conj:strong-rewiring}
Let $\Aa$ be a simplified HD Büchi automaton, and let $p$ be a \good state in $\Aa$ such that there is transition of the form $p\xRightarrow{a}p'$, where $p'$ is not a \good state. Then there exists a \good state $q$ in $\Aa$ such that $L(\Aa,q) = a^{-1} L(\Aa,p)$, and the automaton $\Bb$, obtained by removing all transitions of the form $p\xrightarrow{a}p''$ in $\Aa$ and adding the transition $p\xRightarrow{a}q$, is language-equivalent to $\Aa$.
\end{conjecture}

\subparagraph{Strong rewiring.}
To disprove the strong\=/rewiring conjecture, we consider the automaton $\Astrong$ in 
\cref{fig:strong-rewiring-counterexampleCOPY}, over the alphabet $\Sigma = \{a,b,c,x,\purpley\}$.
For $\alpha,\beta \in \{a,b,c\}$, let
\[L_{\alpha} =  ((a+b+c+\purpley)^* (\varepsilon + x^+\purpley) )^*  x^+ \alpha \, , \text{ and } \;
L_{\alpha\beta\beta} = L_\alpha L_\beta L_\beta.\]
\begin{figure}[ht]
\centering
{\footnotesize\thinmuskip=1.5mu 
\colorlet{sname}{blue!60!black} 
\begin{tikzpicture}[
    base dot/.style={inner sep=0pt, minimum size=5pt, fill=blue!60!black, font=\scriptsize},
    cstate/.style={base dot, circle},
    sstate/.style={base dot, rectangle},
    rstate/.style={base dot, circle, fill=red!40, draw=red!40, text=black,
                   inner sep=1pt, minimum size=14pt, font=\footnotesize},
    acc/.style={thick, green!50!black, double}      
]

\node[rstate] (I) at (0,0) {$I$};

\node[cstate,label=left:$\color{sname}\iota_a$]  (ia) at (-2.2,-1.15) {};
\node[cstate,label=left:$\color{sname}\iota_b$]  (ib) at (  0 ,-1.15) {};
\node[cstate,label=right:$\color{sname}\iota_c$] (ic) at ( 2.2,-1.15) {};

\node[rstate] (pa) at (-2.2,-2.15) {$p_a$};
\node[rstate] (pb) at (  0 ,-2.15) {$p_b$};
\node[rstate] (pc) at ( 2.2,-2.15) {$p_c$};

\node[cstate,label=left:$\color{sname}p_a'$] (pa2) at (-2.2,-3.10) {};
\node[cstate,label=left:$\color{sname}p_b'$] (pb2) at (  0 ,-3.10) {};
\node[cstate,label=left:$\color{sname}p_c'$] (pc2) at ( 2.2,-3.10) {};

\node[sstate,label=below:$\color{sname}q_b$] (lab) at (-2.70,-4.00) {};
\node[sstate,label=below:$\color{sname}q_c$] (lac) at (-1.70,-4.00) {};
\node[sstate,label=below:$\color{sname}q_a$] (lba) at (-0.55,-4.00) {};
\node[sstate,label=below:$\color{sname}q_c$] (lbc) at ( 0.55,-4.00) {};
\node[sstate,label=below:$\color{sname}q_a$] (lca) at ( 1.70,-4.00) {};
\node[sstate,label=below:$\color{sname}q_b$] (lcb) at ( 2.70,-4.00) {};

\path[-stealth]
    (I)  edge [loop above]    node [above] {$\purpley,a,b,c$} (I)
    (I)  edge [bend left=15]  node [above] {$x$} (ia)
    (ia) edge [bend left=15]  node [above,yshift=0.1cm] {$b,c,\purpley$} (I)
    (I)  edge [bend right=15] node [left]  {$x$} (ib)
    (ib) edge [bend right=15] node [right] {$a,c,$}
                              node [right,yshift=-0.25cm] {$\purpley$} (I)
    (I)  edge [bend right=15] node [above] {$x$} (ic)
    (ic) edge [bend right=15] node [above] {$a,b,\purpley$} (I)

    (ia) edge [acc] node [left]  {$a$} (pa)
    (ib) edge [acc] node [right] {$b$} (pb)
    (ic) edge [acc] node [right] {$c$} (pc)

    (ia) edge [loop above] node [above] {$x$} (ia)
    (ib) edge [loop right] node [right] {$x$} (ib)
    (ic) edge [loop above] node [above] {$x$} (ic)

    (pa) edge [bend left=15] node [right] {$x$} (pa2)
    (pb) edge [bend left=15] node [right] {$x$} (pb2)
    (pc) edge [bend left=15] node [right] {$x$} (pc2)

    (pa) edge [loop right] node [right] {$\purpley,a,b,c$} (pa)
    (pb) edge [loop right] node [right] {$\purpley,a,b,c$} (pb)
    (pc) edge [loop right] node [right] {$\purpley,a,b,c$} (pc)

    (pa2) edge [bend left=15] node [left] {$a,\purpley$} (pa)
    (pb2) edge [bend left=15] node [left] {$b,\purpley$} (pb)
    (pc2) edge [bend left=15] node [left] {$c,\purpley$} (pc)

    (pa2) edge [loop right] node [right] {$x$} (pa2)
    (pb2) edge [loop right] node [right] {$x$} (pb2)
    (pc2) edge [loop right] node [right] {$x$} (pc2)

    (pa2) edge [acc] node [left,pos=.6]  {$b$} (lab)
    (pa2) edge [acc] node [right,pos=.6] {$c$} (lac)
    (pb2) edge [acc] node [left,pos=.6]  {$a$} (lba)
    (pb2) edge [acc] node [right,pos=.6] {$c$} (lbc)
    (pc2) edge [acc] node [left,pos=.6]  {$a$} (lca)
    (pc2) edge [acc] node [right,pos=.6] {$b$} (lcb)
    ;

\node[sstate,label=left:$\color{sname}q_a$] (qa) at (5.20,0) {};
\node[sstate,label=left:$\color{sname}q_b$] (qb) at (7.30,0) {};
\node[sstate,label=left:$\color{sname}q_c$] (qc) at (9.40,0) {};

\node[cstate,label=left:$\color{sname}q_a'$] (qa2) at (5.20,-1.10) {};
\node[cstate,label=left:$\color{sname}q_b'$] (qb2) at (7.30,-1.10) {};
\node[cstate,label=left:$\color{sname}q_c'$] (qc2) at (9.40,-1.10) {};

\node[sstate,label=below:$Y$]                (aY) at (4.58,-2.10) {};
\node[sstate,label=below:$\color{sname}q_b$] (ab) at (5.20,-2.10) {};
\node[sstate,label=below:$\color{sname}q_c$] (ac) at (5.82,-2.10) {};

\node[sstate,label=below:$\color{sname}q_a$] (ba) at (6.68,-2.10) {};
\node[sstate,label=below:$Y$]                (bY) at (7.30,-2.10) {};
\node[sstate,label=below:$\color{sname}q_c$] (bc) at (7.92,-2.10) {};

\node[sstate,label=below:$\color{sname}q_a$] (ca) at (8.78,-2.10) {};
\node[sstate,label=below:$\color{sname}q_b$] (cb) at (9.40,-2.10) {};
\node[sstate,label=below:$Y$]                (cY) at (10.02,-2.10) {};

\node[sstate,label=left:$Y$] (Y) at (7.30,-3.15) {};
\node[rstate] (I2) at (7.30,-4.15) {$I$};

\path[-stealth]
    (qa) edge [loop above] node [above] {$\purpley,a,b,c$} (qa)
    (qb) edge [loop above] node [above] {$\purpley,a,b,c$} (qb)
    (qc) edge [loop above] node [above] {$\purpley,a,b,c$} (qc)

    (qa) edge node [right] {$x$} (qa2)
    (qb) edge node [right] {$x$} (qb2)
    (qc) edge node [right] {$x$} (qc2)

    (qa2) edge [purple,bend left] node [left] {$y$} (qa)
    (qb2) edge [purple,bend left] node [left] {$y$} (qb)
    (qc2) edge [purple,bend left] node [left] {$y$} (qc)

    (qa2) edge [loop right] node [right] {$x$} (qa2)
    (qb2) edge [loop right] node [right] {$x$} (qb2)
    (qc2) edge [loop right] node [right] {$x$} (qc2)

    (qa2) edge [acc] node [left,pos=.55]  {$a$} (aY)
    (qa2) edge       node [right,pos=.68] {$b$} (ab)
    (qa2) edge       node [right,pos=.55] {$c$} (ac)

    (qb2) edge       node [left,pos=.55]  {$a$} (ba)
    (qb2) edge [acc] node [right,pos=.68] {$b$} (bY)
    (qb2) edge       node [right,pos=.55] {$c$} (bc)

    (qc2) edge       node [left,pos=.55]  {$a$} (ca)
    (qc2) edge       node [right,pos=.68] {$b$} (cb)
    (qc2) edge [acc] node [right,pos=.55] {$c$} (cY)

    (Y) edge [loop right] node [right] {$x,a,b,c$} (Y)
    (Y) edge [acc]        node [left]  {$y$} (I2)
    ;
\end{tikzpicture}
}
\caption{HD B\"uchi automaton $\Astrong$ for which the strong-rewiring conjecture is false.
To improve readability, we represent the automaton ``by chunks''; note that some states
appear multiple times (in total, the automaton has $17$ states).}
\label{fig:strong-rewiring-counterexampleCOPY}
\end{figure}

The automaton $\Astrong$ recognises the words that contain infinitely many factors in some
$L_{\alpha\beta\beta}$ with $\alpha\neq\beta$, as well as infinitely many $\purpley$s:
\[L(\Astrong) = \Lstrong =\Big(\Sigma^{*} \cdot \bigcup_{\alpha\neq\beta\in \{a,b,c\}} L_{\alpha\beta\beta}\, \cdot  \Sigma^{*}\cdot \purpley\Big)^{\omega}.\]
For comparison, the language of $\Abkks$ is $(xa+xb)^*(xaxa+xbxb)^{\omega}$, with the two-letter words
$xa$ and $xb$ in the role of blocks played by $L_\alpha$ here. 

Indeed, on reading a factor of $L_\alpha L_\beta L_\beta$ with
$\alpha\neq\beta$ from $I$, the automaton can produce the run
$I \xrightarrow{L_\alpha} p_\alpha \xrightarrow{L_\beta} q_\beta \xrightarrow{L_\beta} Y$, and from
$Y$ it returns to $I$ on the letter $\purpley$ along a \significant transition. 
An HD strategy for $\Astrong$ needs to ensure reaching some \significant transition from $I$ if many factors $L_\alpha L_\beta L_\beta$ are read; this can be done similarly as in  $\Abkks$.



The strong-rewiring conjecture fails for $\Astrong$ because the transition $Y\xRightarrow{\purpley}I$
cannot be redirected on its own. Indeed, for \emph{every} \good state $s$, replacing it by
$Y\xRightarrow{\purpley}s$ changes the language (see \cref{lemma:no-weak-rewiring}). The weak-rewiring
conjecture, however, survives here (see \cref{fig:strong-but-no-weak}),
which is possible because the \good states of $\Astrong$ are
spread over several SCCs of non-\significant transitions.


\subparagraph{Weak rewiring.}
\input{figures_sec2and3/SidebysideWeakdet.tex}
We  now present the automaton $\Aweak$ shown in \cref{fig:weak-rewiring-counterexampleFIRST}, 
which we will prove to be a counterexample to the weak-rewiring conjecture. 
The automaton $\Aweak$ is obtained by modifying $\Astrong$ as follows.
\begin{enumerate}
    \item The \significant transitions from the $p'$s to the $q$s are made non-\significant. 
    \item So that these transitions remain non-\significant once the automaton is made normal, we add three \emph{reset letters} $\reset{a},\reset{b}$ and $\reset{c}$: from every state $s$ in $\{I,\iota_a,\iota_b,\iota_c\}$ there is a \significant transition $s\xRightarrow{\reset{\alpha}}p_\alpha$, from $Y$ there are non-\significant self-loops on the reset letters, and from every other state $s$ there is a non-\significant transition $s\xrightarrow{\reset{\alpha}}p_\alpha$.
\end{enumerate}
Writing $L'_{\alpha\beta\beta} = (L_{\alpha}+\reset{\alpha})L_{\beta}L_{\beta}$ and
$\Sigma' = \{x,a,b,c,\purpley,\reset{a},\reset{b},\reset{c}\}$, the language of $\Aweak$ is
$$\Lweak = \Big((x+a+b+c+\purpley)^{*} \cdot  \bigcup_{\alpha\neq\beta\in\{a,b,c\}} L'_{\alpha\beta\beta} \cdot  \Sigma'^{*}\cdot \purpley\Big)^{\omega}.$$
These changes make sure every \good state apart from $Y$ is a single SCC of $\reach(\Aweak)$, 
the blue box in \cref{fig:weak-rewiring-counterexampleFIRST}. A rewiring is 
then specified by the target of the transition out of $Y$.
We prove that none of these rewirings is language-equivalent by 
providing a distinguishing word for each of them (c.f. \cref{lemma:weak-rewiring-dead}).
\subsubsection{Copying the state $Y$}\label{subsec:ov-weird-Y-det}
The previous automaton $\Aweak$ may momentarily fill an optimistic researcher with the hope of having found an 
HD Büchi automaton strictly smaller than every language-equivalent deterministic one.
Unfortunately, there are smaller deterministic automata recognising $\Lweak$.

The deterministic automaton $\Dd_{\mathtt{weak}}$ of \cref{fig:wierdYdetFIRST} manages to get rid of the nondeterministic component $\{I,\iota_a,\iota_b,\iota_c\}$ by adding
three copies of the state $Y$ ($Y_a,Y_b$ and $Y_c$), thereby using $2$ extra states to remove the $4$ states of the nondeterministic component.

This determinisation strategy works here because the number of $Y$-copies ($2$ in this case) is smaller than the number of non-\good states ($4$ in this case).
To defeat it, it suffices to enlarge our automaton with many ``blocks'', so that we would require to introduce $>4$ copies of the state $Y$.

\subsubsection{Determinising the nondeterministic component}\label{subsec:ov-determinising}
A third tactic remains: replacing the nondeterministic component by a deterministic one doing the
same job. For $\Abkks$ this succeeds with $2$ states in place of $3$ (\cref{fig:detForBKKS13bFIRST}), and
the same trick applies to both $\Astrong$ and $\Aweak$. To defeat this tactic, we use a nondeterministic gadget with $4$ states that cannot be simulated by a deterministic one with less than $5$ states. 
We depict this gadget and its deterministic version in \Cref{fig:beatReplaceDFAs}.
This was inspired by an example that arose in a different context~\cite[Figure~6]{HPT25}.

\begin{figure}[H]
    \centering
    {\footnotesize\thinmuskip=1.5mu 

        \begin{tikzpicture}[
    base dot/.style={
        inner sep=0pt,      
        minimum size=5pt,   
        fill=red!50!white, 
        font=\scriptsize    
    },
    cstate/.style={base dot, circle},    
    sstate/.style={base dot, rectangle}, 
    node distance=1.5cm,
    tight label/.style={auto, inner sep=0.5pt, text=black}
]
        \tikzset{every state/.style = {inner sep=-3pt,minimum size =15}}

\node[cstate, fill=red!40, red!40] (q1) at (-1,1.5) {AA};

 \node[cstate, fill=red!40, red!40] (q0) [above=0.8cm of q1] {AA};

 \node[cstate, fill=red!40, red!40] (q3) [right=0.8cm of q0] {AA};
 \node[cstate, fill=red!40, red!40] (q2) [below=0.8cm of q3] {AA};
 \path[->] (-1.6,2.8) edge (q0); 
    
    \node[cstate, fill=red!40, red!40] (r1) at (2.8,1.5) {AA};
    \node[cstate, fill=red!40, red!40] (r0) [above=0.8cm of r1] {AA};
    \node[cstate, fill=red!40, red!40] (r2) [right=0.8cm of r1] {AA};
    \node[cstate, fill=red!40, red!40] (r3) [right=0.8cm of r0] {AA};
    \node[cstate, fill=red!40, red!40] (r4) [right=0.8cm of r2] {AA}; 
    \path[->] (2.2,2.8) edge (r0); 

\path[-stealth]
(q1) edge[ thick, green!50!black, double] node [left,yshift=0cm,xshift=0cm] {$1$} (-1,0.8)
(q2) edge[ thick, blue!50!black, double] node [left,yshift=0cm,xshift=0cm] {$2$} (0.27,0.8)

(r1) edge[ thick, green!50!black, double] node [left,yshift=0cm,xshift=0cm] {$1$} (2.8,0.8)
(r2) edge[ thick, green!50!black, double] node [left,yshift=0cm,xshift=0cm] {$1$} (3.9,0.8)
(r2) edge[ thick, blue!50!black, double] node [right,yshift=0cm,xshift=0cm] {$2$} (4.2,0.8)
(r4) edge[ thick, blue!50!black, double] node [left,yshift=0cm,xshift=0cm] {$2$} (5.5,0.8);

    \path[-stealth]
    (q0) edge [loop above] node [above] {$a,b,1,2$} (q0)
    (q0) edge [bend left = 8] node [above] {$a$} (q3)
    (q1) edge [] node [below] {$a$} (q3)
    (q3) edge [bend left = 8] node [below] {$1,2$} (q0)
    (q0) edge [bend left = 8] node [right] {$c$} (q1)
    (q1) edge [bend left = 8] node [left] {$b,a,2$} (q0)
    (q3) edge [bend right = 8] node [left] {$b$} (q2)
    (q2) edge [bend right = 8] node [right] {$c,a$} (q3)
    (q2) edge [bend right = 90, looseness = 3] node [above] {$1$} (q0)  
    (q3) edge [loop above] node [right] {$a,c$} (q3)
    (q1) edge [loop left] node [left] {$c$} (q1)
    (q2) edge [loop right] node [right] {$b$} (q2)

    (r0) edge [loop above] node [above] {$b,1,2$} (r0)
    (r0) edge [bend left = 8] node [above] {$a$} (r3)
    (r3) edge [bend left = 8] node [below] {$1,2$} (r0)
    (r1) edge  node [below, left] {$a$} (r3)
    (r0) edge [bend left = 8] node [tight label,right] {$c$} (r1)
    (r1) edge [bend left = 8] node [tight label,left] {$b,2$} (r0)    
    (r3) edge [bend left = 8] node [tight label,right] {$c$} (r2)
    (r2) edge [bend left = 8] node [tight label,left] {$a$} (r3)
    (r3) edge [bend right = 8] node [tight label,below=3pt,left] {$b$} (r4)
    (r4) edge [bend right = 8] node [tight label,above=3pt, right] {$a$} (r3)
    (r2) edge [bend right = 8] node [tight label,below=0.5pt] {$b$} (r4)
    (r4) edge [bend right = 8] node [tight label,above=0.5pt] {$c$} (r2)
    (r3) edge [in=30,out=60,loop]  node [above] {$a$} (r3)
    (r1) edge [loop left] node [left] {$c$} (r1)
    (r2) edge [loop left] node [left] {$c$} (r2)
    (r4) edge [loop right] node [right] {$b$} (r4)
    
    (r4) edge [bend right = 85, looseness = 1.6] node [above] {$1$} (r0);

    \end{tikzpicture}
}
    \caption{Two automata distinguishing the languages $L_1'$ and $L_2'$.}
    \label{fig:beatReplaceDFAs}
\end{figure}

To understand the functioning of this gadget, let $\Sigma=\{a,b,c,1,2\}$ and define the languages $L'_1 = \Sigma^{*}c1$, and $L'_2 = \Sigma^* a  (a+b+c)^* b 2$. The automata in \cref{fig:beatReplaceDFAs} can recognise whether a word belongs to $L_1'$ (building a run that ends by a green arrow) or to $L_2'$ (building a run that ends by a blue arrow). The automaton on the left uses nondeterminism for it: if a word belongs to $L_1'$, it chooses to stay on the states on the left, while if the word belongs to $L_2'$, it chooses to switch to the top rigth state upon reading $a$. The righthand automaton manages to distinguish these languages in a deterministic way by using an extra state.
We refer to \cref{fig:HDBuchinotMRAgain} (page~\pageref{fig:HDBuchinotMRAgain}) for a full description of an HD automaton $\Areplace$ that uses this nondeterministic gadget.


\subsection{The main automaton}\label{subsec:ov-main}
\input{figures_sec2and3/bacopycopy}
We now describe the automaton $\Amain$, depicted in \cref{fig:wholeautomatonmainpaper}, which combines the three
ingredients above: the skeleton of $\Aweak$ with its reset letters rules out rewirings, using $6$
internal blocks rules out copying $Y$, and having the gadget of \Cref{fig:beatReplaceDFAs} as the nondeterministic component rules
out determinising the nondeterministic component.

\subparagraph{The language.}
We fix the alphabet $\Sigma = \{a,b,c,1,2,3,4,5,6,\reset{1},\dots,\reset{6},\purpley\},
$ and write 
$E = \{1,2,3\}$, and $F=\{4,5,6\}$. 
The role played by $L_a,L_b,L_c$ for $\Astrong$ is now played by six languages of finite words
$\langBlock{1},\dots,\langBlock{6}$.
The final language is then given by:
\[ 
\Lmain = \Big(\Sigma^* \cdot \!\!\bigcup_{\substack{\alpha\neq \beta \\ 1\leq \alpha,\beta \leq 6}}\!\! \langInfix{\alpha}{\beta} \cdot \Sigma^* \cdot \purpley \Big)^\omega,
\;\; \text{ with } \;\;
\langInfix{\alpha}{\beta} = (\langBlock{\alpha}+\reset{\alpha})\langBlock{\beta}\langBlock{\beta}
\ \text{ for } \alpha\neq\beta \in [6].\]
\begin{wrapstuff}[l, type=figure, width=5.5cm, top=1]
\centering
\captionsetup{font=footnotesize,skip=2pt}
\scalebox{0.8}{
\begin{tikzpicture}[
        base dotred/.style={
        inner sep=0pt,      
        minimum size=5pt,   
        fill=red!50!white, 
        font=\scriptsize    
    },
    base dot/.style={
        inner sep=0pt,      
        minimum size=5pt,   
        fill=blue!60!black, 
        font=\scriptsize    
    },
    cstate/.style={base dotred, circle},  
    node distance=1.5cm,
    tight label/.style={auto, inner sep=0.5pt, text=black},
    sstate/.style={base dot, rectangle}, 
    house peak H/.store in=\peakH,
    house peak H=-1.6cm,
    house eaves H/.store in=\eavesH, 
    house eaves H=0.15cm,
    house base H/.store in=\baseH,   
    house base H=-0.7cm,
    house width/.store in=\houseW,   
    house width=0.5cm,
    group sep x/.store in=\gSepX, 
    group sep x=2.6cm, 
    group sep y/.store in=\gSepY, 
    group sep y=3.2cm, 
        row top y/.store in=\rowTopY,
    row bot y/.store in=\rowBotY,
    top x off/.store in=\topX,
    bot x off/.store in=\botX,
    row top y=0.6cm,
    row bot y=-0.9cm,
    top x off=0.6cm,
    bot x off=1.2cm
]
    \filldraw [fill=violet!10, draw=violet!30] (-2.5,3.2) rectangle (2.8,-0.5);

    \filldraw [fill=violet!30, draw=violet!50] (-2.2,0.4) rectangle (0.3,-0.3);
    
\foreach \i in {1} {
    \coordinate (center\i) at ({(\i-1)*\gSepX}, 0);

    \node[cstate, label=above:$p$] (p\i) at (-1,2) {};
    \node[cstate] (q\i) at (1,2)  {};

    \node[cstate] (t\i) at (-1.5, 0) {};
    \node[cstate] (s\i) at (0,0)      {};
    \node[cstate] (r\i) at (2,0)  {};
}
    \node[sstate, label=below:$\color{blue!60!black}\ell_1$] (accept1) at (-2.2,-1.3) {};
    \node[sstate, label=below:$\color{blue!60!black}\ell_2$] (accept2) at (-1.2,-1.3) {};
    \node[sstate, label=below:$\color{blue!60!black}\ell_3$] (accept3) at (-0.2,-1.3) {};
    \node[sstate, label=below:$\color{blue!60!black}\ell_4$] (accept4) at (0.8,-1.3) {};
    \node[sstate, label=below:$\color{blue!60!black}\ell_5$] (accept5) at (1.8,-1.3) {};
    \node[sstate, label=below:$\color{blue!60!black}\ell_6$] (accept6) at (2.8,-1.3) {};




 \foreach \i in {1} {
  \path[->,>=stealth, thick, shorten >=1pt]
    (p\i) edge [loop left] node [above,tight label] {$b,{\color{purple}E,F}$} (p\i)
    (q\i) edge [loop right] node [tight label] {$a$} (q\i)
    
    (r\i) edge [loop right] node [tight label] {$b$} (r\i)
    (t\i) edge [loop left] node [tight label, yshift=0cm, xshift=-0.05cm] {$c$} (t\i)

    (p\i) edge  node [tight label,below] {$a$} (q\i)
    (q\i) edge [purple,bend right] node [tight label,above] {${\color{purple}{E,F}}$} (p\i)

    (q\i) edge [bend left = 10] node [tight label] {$b$} (r\i)   
    (r\i) edge [bend left = 10] node [tight label,yshift=0.1cm, xshift=-.05cm] {$a$} (q\i)
     
    (r\i) edge [bend left = 10] node [tight label] {$c$} (s\i)  
    (s\i) edge [bend left = 10] node [tight label] {$b$} (r\i)       
    (p\i) edge [bend right =10] node [ tight label, left] {$c$} (t\i)
    (t\i) edge [bend right=10] node [tight label, right] {$b$} (p\i)

    (q\i) edge [bend left = 10] node [tight label, yshift=0cm] {$c$} (s\i)  
    (s\i) edge [bend left = 10] node [tight label,yshift=-0.15cm] {$a$} (q\i)   
    (t\i) edge  node [tight label, left, yshift=-0.20cm, xshift=-.35cm] {$a$} (q\i)
    
    (t\i) edge [bend left=40, purple] node [tight label, left] {${\color{purple}F}$} (p\i)

    (r\i) edge [purple,out=60,in=60, looseness=1.8] node [tight label, above,xshift=0.1cm] {{\color{purple}$E$}} (p\i)  
    (s\i) edge [purple] node [tight label,xshift=0.1cm,yshift=0.6cm] {$\color{purple}F$} (p\i)   




    (s\i) edge [loop left] node [tight label,left,xshift=-0.1cm,yshift=0cm] {$c$} (s\i);
    ;

    \path[-stealth] 
    (-1.6,-0.3) edge [thick, green!50!black, out=200,in=90,looseness=1] node [left] {$1$} (accept1) 
    (-1.2,-0.3) edge [thick, green!50!black, out=-90,in=90,looseness=1] node [left] {$2$} (accept2) 
    (-0.8,-0.3) edge [thick, green!50!black, out=-60,in=90,looseness=1] node [right, xshift= 0.10cm] {$3$} (accept3) 

    (r\i) edge [thick, green!50!black, out=-120,in=90,looseness=1] node [left, xshift=-0.10cm] {$4$} (accept4) 
    (r\i) edge [thick, green!50!black, out=-90,in=90,looseness=1] node [left] {$5$} (accept5) 
    (r\i) edge [thick, green!50!black, out=-60,in=90,looseness=1] node [right] {$6$} (accept6) 
    
    (-1.2,3.2) edge [thick, out=110,in=130,looseness=3] node [xshift=0.0cm,left] {$y$} (p\i);
}
\end{tikzpicture}
}\caption{The DFA-classifier $\Aclassifier$.} 
    \label{fig:LanguageL1copy}
\end{wrapstuff}
That is, the language $\Lmain$ consists of the words containing infinitely many factors in some
$\langInfix{\alpha}{\beta}$, as well as infinitely many $\purpley$s. 


The six languages $\langBlock{1},\dots,\langBlock{6}$ are described by the DFA-classifier $\Aclassifier$ of
\cref{fig:LanguageL1copy}: a word is in $\langBlock{\alpha}$ if its run over $\Aclassifier$ from $p$ ends
in the accepting state~$\sqrstate{\ell_\alpha}$.
In particular, the languages $\langBlock{\alpha}$ are pairwise disjoint, the letter $\purpley$ brings back
$\Aclassifier$ to its initial state, and no word containing a reset letter $\reset{\alpha}$ lies in any
$\langBlock{\alpha}$.
We note that this DFA mimics the gadget from \Cref{fig:beatReplaceDFAs}
with the languages $L_1, L_2$ and $L_3$ corresponding to $L_1'$ from \cref{subsec:ov-weird-Y-det}, and the languages $L_4, L_5$ and~$L_6$ playing the role of $L_2'$. 



\subparagraph{The automaton.} The automaton $\Amain$ (\cref{fig:wholeautomatonmainpaper}) consists of an initial nondeterministic component, similar to that of \cref{fig:beatReplaceDFAs} (the initial state $I$ together with $I_a,I_b,I_c$); twelve \emph{basic blocks}  (the groups of $5$ \emph{circular} states $\circstate{p_i},\circstate{q_i},\circstate{r_i},\circstate{s_i},\circstate{t_i}$ for $i\in[6]$, or $5$
\emph{square} states $\sqrstate{p_i},\sqrstate{q_i},\sqrstate{r_i},\sqrstate{s_i},\sqrstate{t_i}$ for
$i\in[6]$) each of them being a copy of the $5$-state gadget on the right of \cref{fig:beatReplaceDFAs};
and the state $\sqrstate{Y}$.
This gives $4 + 2\times 5\times 6 + 1 = 65$ states in total.
Restricted to non-\significant transitions,
$\Amain$ has exactly three SCCs: the four initial states, the $60$ block states, and $\sqrstate{Y}$.

To visualise the automaton, we split the transitions across two figures and use a few graphical shorthands similar to the earlier section.
The transitions of the automaton over the letters $a,b,c,$ and $\purpley$ are described in \cref{fig:transitionsoverabcyCOPY}, whereas the transitions of the automaton over the rest of the letters $1,\dots,6$ and $\reset{1},\dots,\reset{6}$ are depicted in \cref{fig:transitionsoverEFrCOPY}. 
We use variables to represent sets of transitions,
therefore, a single arrow like $\xrightarrow{f}\circstate{p_f}$ in fact represents the transitions $\xrightarrow{4} \circstate{p_4}$, $\xrightarrow{5} \circstate{p_5}$, and $\xrightarrow{6} \circstate{p_6}$.
Finally, the rectangles around states represent grouped sources. An arrow originating from a box around a set of states implies a transition from every state within that box.

To accept a word in $\Lmain$, the automaton reads it until it sees an infix in
 $\langBlock{\alpha}$ or $\reset{\alpha}$ followed by $\langBlock{\beta}\langBlock{\beta}$ with
$\alpha\neq\beta$, at which point it moves to $\sqrstate{Y}$ and waits for letter $\purpley$ to restart.


\subparagraph{Resolving the nondeterminism.} The only nondeterminism is on $I$ and $I_c$ and over the
letter $a$. It guesses whether the infix being read will end with a letter
in $E$, in which case the run should stay in $I, I_c$, or with a letter in $F$, in which case it should
move to $I_a$. An HD-resolver makes its first guess arbitrarily and afterwards uses the information from the latest infix in some $L_\alpha$. If $u$ is the prefix read so far and $\alpha\in[6]$ is the index of the latest factor of
$u$ lying in some $\langBlock{\alpha}$, then on reading $a$ from $I$ or $I_c$ the resolver moves to
$I$ if $\alpha\in E$, and to $I_a$ if $\alpha\in F$. Every word of $\Lmain$ contains infinitely many
factors $\langBlock{\beta}\langBlock{\beta}$; after the first factor of such a pair the
resolver's guess is correct and leaves the
top box. As this recurs after every visit to $\sqrstate{Y}$, the induced run is accepting.

We prove in \cref{thm:A65isHD} that $\Amain$ is simplified,
which shows that the automaton is history deterministic via \cref{lemma:buchi-simplified-implies-hd}
and is also used in the proof of succinctness below.

\subsection{Succinctness proof}\label{subsec:ov-succinctness}
It remains to argue that no deterministic Büchi automaton with $65$ states recognises $\Lmain$
(\cref{thm:succinctness}). Minimising deterministic Büchi automata is $\NP$-complete for both
state-based~\cite{Sch10} and transition-based~\cite{AE25} acceptance, and due to the size of our automaton, this task is not at the reach of solvers. We instead use structural results 
on coB\"uchi automata, and finally use a solver for a different $\NP$-complete problem over smaller instances.

The main observation is that a deterministic B\"uchi automaton for $\Lmain$ with $n$ states is, 
when interpreted as a coB\"uchi automaton with the same transitions, 
a deterministic coBüchi automaton for $\Lmain^c$ with
$n$ states. For coBüchi automata, Abu Radi and Kupferman~\cite{AK22} showed many structural properties of the minimal HD coB\"uchi automaton, which in particular impose constraints on the structure of any other equivalent HD coB\"uchi automaton, deterministic ones included.

Applying the complementation of Abu Radi, Kupferman and Leshkowitz~\cite{AKL24} to $\Amain$ yields an HD coBüchi automaton $\Cmain$ for $\Lmain^c$ whose
states are precisely the $61$ \good states of $\Amain$, that is, the $60$ block states and
$\sqrstate{Y}$ (\cref{lemma-sec5-cmain-iscomplementandhd}). We verify that
$\Cmain$ is the minimal HD coB\"uchi automaton for this language (\cref{lemma:cmain-minimality}).

The intuitive idea of our proof is then as follows: Let $\Dmain$ be an equivalent deterministic coB\"uchi automaton for
$\Lmain^c$. Using the structural properties from~\cite{AK22} we show that $\Dmain$ must contain at least $61$ states corresponding to the states of $\Cmain$ (\cref{lemma:Dcore-one-state-for-everythingin60}), and moreover one extra deterministic part that mimics the behaviour of the $4$ states of the non-deterministic component in $\Amain$ (\cref{cor:5states-dcore}).
To finish the proof, it then suffices to show that this extra deterministic part must have at least $5$ states, for which we use the help of the SAT-based tool DFAMiner~\cite{DLS24}.
More precisely, we let $\Theju$ be the language of infinite words containing no prefix in any $\langBlock{i}$. Our only computer-aided step, \cref{lemma:compy}, 
states that every deterministic safety automaton $\Dd$ with 
$\Theju \subseteq L(\Dd) \subseteq \safeLangA{\Cmain}{q}$, for a block
state $q$ of $\Cmain$, has at least $5$ states, which allows to conclude.

\section{Technical content}\label{sec:full-details}
This section contains full details for our constructions and proofs. It is structured in three subsections (\ref{subsec:full-examples},~\ref{subsec:full-main} and~\ref{subsec:full-succinctness}), mirroring the structure of the previous; each subsection contains full details on the content of the corresponding subsection of \Cref{sec:overview}.


\subsection{Motivating examples (full details)}\label{subsec:full-examples}
We begin by analysing why na\"ive determinisation techniques fail. We provide increasingly complex counterexamples to specific ``straw man''  determinisation techniques, which then help us build the intuition required for our main result. These examples demonstrate the pitfalls that our final construction must avoid. 
Although the definition of the 65-state automaton can be understood directly by skipping ahead to \cref{subsec:full-main}, we strongly recommend reading this section first to understand the method behind the madness of the final construction.

\subsubsection{The rewiring conjectures}

\paragraph{Strong rewiring}

\begin{figure}[h]
\centering
{\footnotesize\thinmuskip=1.5mu 
\colorlet{sname}{blue!60!black} 
\begin{tikzpicture}[
    base dot/.style={inner sep=0pt, minimum size=5pt, fill=blue!60!black, font=\scriptsize},
    cstate/.style={base dot, circle},
    sstate/.style={base dot, rectangle},
    rstate/.style={base dot, circle, fill=red!40, draw=red!40, text=black,
                   inner sep=1pt, minimum size=14pt, font=\footnotesize},
    acc/.style={thick, green!50!black, double}      
]

\node[rstate] (I) at (0,0) {$I$};

\node[cstate,label=left:$\color{sname}\iota_a$]  (ia) at (-2.2,-1.15) {};
\node[cstate,label=left:$\color{sname}\iota_b$]  (ib) at (  0 ,-1.15) {};
\node[cstate,label=right:$\color{sname}\iota_c$] (ic) at ( 2.2,-1.15) {};

\node[rstate] (pa) at (-2.2,-2.15) {$p_a$};
\node[rstate] (pb) at (  0 ,-2.15) {$p_b$};
\node[rstate] (pc) at ( 2.2,-2.15) {$p_c$};

\node[cstate,label=left:$\color{sname}p_a'$] (pa2) at (-2.2,-3.10) {};
\node[cstate,label=left:$\color{sname}p_b'$] (pb2) at (  0 ,-3.10) {};
\node[cstate,label=left:$\color{sname}p_c'$] (pc2) at ( 2.2,-3.10) {};

\node[sstate,label=below:$\color{sname}q_b$] (lab) at (-2.70,-4.00) {};
\node[sstate,label=below:$\color{sname}q_c$] (lac) at (-1.70,-4.00) {};
\node[sstate,label=below:$\color{sname}q_a$] (lba) at (-0.55,-4.00) {};
\node[sstate,label=below:$\color{sname}q_c$] (lbc) at ( 0.55,-4.00) {};
\node[sstate,label=below:$\color{sname}q_a$] (lca) at ( 1.70,-4.00) {};
\node[sstate,label=below:$\color{sname}q_b$] (lcb) at ( 2.70,-4.00) {};

\path[-stealth]
    (I)  edge [loop above]    node [above] {$\purpley,a,b,c$} (I)
    (I)  edge [bend left=15]  node [above] {$x$} (ia)
    (ia) edge [bend left=15]  node [above,yshift=0.1cm] {$b,c,\purpley$} (I)
    (I)  edge [bend right=15] node [left]  {$x$} (ib)
    (ib) edge [bend right=15] node [right] {$a,c,$}
                              node [right,yshift=-0.25cm] {$\purpley$} (I)
    (I)  edge [bend right=15] node [above] {$x$} (ic)
    (ic) edge [bend right=15] node [above] {$a,b,\purpley$} (I)

    (ia) edge [acc] node [left]  {$a$} (pa)
    (ib) edge [acc] node [right] {$b$} (pb)
    (ic) edge [acc] node [right] {$c$} (pc)

    (ia) edge [loop above] node [above] {$x$} (ia)
    (ib) edge [loop right] node [right] {$x$} (ib)
    (ic) edge [loop above] node [above] {$x$} (ic)

    (pa) edge [bend left=15] node [right] {$x$} (pa2)
    (pb) edge [bend left=15] node [right] {$x$} (pb2)
    (pc) edge [bend left=15] node [right] {$x$} (pc2)

    (pa) edge [loop right] node [right] {$\purpley,a,b,c$} (pa)
    (pb) edge [loop right] node [right] {$\purpley,a,b,c$} (pb)
    (pc) edge [loop right] node [right] {$\purpley,a,b,c$} (pc)

    (pa2) edge [bend left=15] node [left] {$a,\purpley$} (pa)
    (pb2) edge [bend left=15] node [left] {$b,\purpley$} (pb)
    (pc2) edge [bend left=15] node [left] {$c,\purpley$} (pc)

    (pa2) edge [loop right] node [right] {$x$} (pa2)
    (pb2) edge [loop right] node [right] {$x$} (pb2)
    (pc2) edge [loop right] node [right] {$x$} (pc2)

    (pa2) edge [acc] node [left,pos=.6]  {$b$} (lab)
    (pa2) edge [acc] node [right,pos=.6] {$c$} (lac)
    (pb2) edge [acc] node [left,pos=.6]  {$a$} (lba)
    (pb2) edge [acc] node [right,pos=.6] {$c$} (lbc)
    (pc2) edge [acc] node [left,pos=.6]  {$a$} (lca)
    (pc2) edge [acc] node [right,pos=.6] {$b$} (lcb)
    ;

\node[sstate,label=left:$\color{sname}q_a$] (qa) at (5.20,0) {};
\node[sstate,label=left:$\color{sname}q_b$] (qb) at (7.30,0) {};
\node[sstate,label=left:$\color{sname}q_c$] (qc) at (9.40,0) {};

\node[cstate,label=left:$\color{sname}q_a'$] (qa2) at (5.20,-1.10) {};
\node[cstate,label=left:$\color{sname}q_b'$] (qb2) at (7.30,-1.10) {};
\node[cstate,label=left:$\color{sname}q_c'$] (qc2) at (9.40,-1.10) {};

\node[sstate,label=below:$Y$]                (aY) at (4.58,-2.10) {};
\node[sstate,label=below:$\color{sname}q_b$] (ab) at (5.20,-2.10) {};
\node[sstate,label=below:$\color{sname}q_c$] (ac) at (5.82,-2.10) {};

\node[sstate,label=below:$\color{sname}q_a$] (ba) at (6.68,-2.10) {};
\node[sstate,label=below:$Y$]                (bY) at (7.30,-2.10) {};
\node[sstate,label=below:$\color{sname}q_c$] (bc) at (7.92,-2.10) {};

\node[sstate,label=below:$\color{sname}q_a$] (ca) at (8.78,-2.10) {};
\node[sstate,label=below:$\color{sname}q_b$] (cb) at (9.40,-2.10) {};
\node[sstate,label=below:$Y$]                (cY) at (10.02,-2.10) {};

\node[sstate,label=left:$Y$] (Y) at (7.30,-3.15) {};
\node[rstate] (I2) at (7.30,-4.15) {$I$};

\path[-stealth]
    (qa) edge [loop above] node [above] {$\purpley,a,b,c$} (qa)
    (qb) edge [loop above] node [above] {$\purpley,a,b,c$} (qb)
    (qc) edge [loop above] node [above] {$\purpley,a,b,c$} (qc)

    (qa) edge node [right] {$x$} (qa2)
    (qb) edge node [right] {$x$} (qb2)
    (qc) edge node [right] {$x$} (qc2)

    (qa2) edge [purple,bend left] node [left] {$y$} (qa)
    (qb2) edge [purple,bend left] node [left] {$y$} (qb)
    (qc2) edge [purple,bend left] node [left] {$y$} (qc)

    (qa2) edge [loop right] node [right] {$x$} (qa2)
    (qb2) edge [loop right] node [right] {$x$} (qb2)
    (qc2) edge [loop right] node [right] {$x$} (qc2)

    (qa2) edge [acc] node [left,pos=.55]  {$a$} (aY)
    (qa2) edge       node [right,pos=.68] {$b$} (ab)
    (qa2) edge       node [right,pos=.55] {$c$} (ac)

    (qb2) edge       node [left,pos=.55]  {$a$} (ba)
    (qb2) edge [acc] node [right,pos=.68] {$b$} (bY)
    (qb2) edge       node [right,pos=.55] {$c$} (bc)

    (qc2) edge       node [left,pos=.55]  {$a$} (ca)
    (qc2) edge       node [right,pos=.68] {$b$} (cb)
    (qc2) edge [acc] node [right,pos=.55] {$c$} (cY)

    (Y) edge [loop right] node [right] {$x,a,b,c$} (Y)
    (Y) edge [acc]        node [left]  {$y$} (I2)
    ;
\end{tikzpicture}
}
\caption{HD B\"uchi automaton $\Astrong$ for which the strong-rewiring conjecture is false.
To improve readability, we represent the automaton ``by chunks''; note that some states
appear multiple times (in total, the automaton has $17$ states).}
\label{fig:strong-rewiring-counterexample}
\end{figure}
  To disprove the strong\=/rewiring conjecture, we will consider the automaton $\Astrong$ in 
\cref{fig:strong-rewiring-counterexample}, over the alphabet $\Sigma = \{a,b,c,x,\purpley\}$.

We again describe the language recognised by $\Astrong$ for ease here. Towards this, we define, for each $\alpha \in\{a,b,c\}$, the language $L_{\alpha}$ as the language recognised by the DFA in \cref{fig:sec3-DFAforLemma3point3}.

\begin{figure}[H]
\centering

\begin{tikzpicture}[
    base dot/.style={
        inner sep=0pt,      
        minimum size=5pt,   
        fill=blue!60!black, 
        font=\scriptsize    
    },
    cstate/.style={base dot, circle, fill=red!40, minimum size = 15},    
    sstate/.style={base dot, rectangle}, 
    node distance=1.5cm,
    tight label/.style={auto, inner sep=0.5pt, font=\scriptsize, text=black}   
]
    \tikzset{every state/.style = {inner sep=-3pt,minimum size =15}}

    \node[cstate] (Ia) at (0,0) {$\iota$};
    \node[cstate] (pa) at (1.5,0) {$p$};
    \node[state, accepting, fill=red!40] (fa) at (3,0) {$q$};
    
    \path[->] (-0.5,0) edge (Ia);

    \path[-stealth]
        (Ia) edge node [above] {$x$} (pa)
        (Ia) edge [loop above] node [above]{$a,b,c,\purpley$} (Ia)
        (pa) edge [loop above] node [above] {$x$} (pa)
        (pa) edge node [above] {$\alpha$} (fa)
        (pa) edge [bend left] node [below] {$\purpley$} (Ia)
        ;
\end{tikzpicture}
 \caption{The DFA $D_{\alpha}$ for each $\alpha$ in $\{a,b,c\}$. Note that the DFA $D_{\alpha}$ is not complete.}
\label{fig:sec3-DFAforLemma3point3}
\end{figure}

Equivalently,  $L_{\alpha} =  ((a+b+c+\purpley)^* (\varepsilon + x^+\purpley) )^*  x^+ \alpha$,
for each $\alpha \in \{a,b,c\}$. 
For each $\alpha,\beta \in \{a,b,c\}$, we define $L_{\alpha\beta\beta}$ as 
$L_{\alpha\beta\beta} = L_{\alpha} L_{\beta}L_{\beta}.$

Then, the language recognised by $\Astrong$ in  \cref{fig:strong-rewiring-counterexample} consists of words that contain infinitely often infixes in $L_{\alpha\beta\beta}$, for  some $\alpha\neq \beta$, and also contain infinitely many $y$s. Formally, the language $L(\Astrong)$ is the following:
$$\Lstrong =(\Sigma^{*} \cdot \bigcup_{\alpha\neq\beta\in \{a,b,c\}} L_{\alpha\beta\beta}\, \cdot  \Sigma^{*}\cdot \purpley)^{\omega}.$$

To intuitively see that this automaton indeed recognises the language $\Lstrong$, note that whenever an infix $L_\alpha L_\beta L_\beta$ is read from $I$, with $\alpha \neq \beta$, the automaton can produce a run
$$ I \xrightarrow{ L_\alpha} p_\alpha \xrightarrow{ L_\beta} q_\beta  \xrightarrow{ L_\beta} Y.$$
From $Y$, the automaton takes a significant transition and moves to $I$ upon reading letter~$\purpley$. 

\begin{remark}
We remark that $\Astrong$ has some redundancy; one can construct a deterministic automaton that recognises $\Astrong$ by merging the states $\iota_a,\iota_b,\iota_c$ into a single state $\iota$, such that transition from $\iota$ on the letter $a$ (resp.\ $b$ or $c$) is to the state $p_a$ (resp.\ $p_b$ or $p_c$). This redundancy exists to disprove the rewiring conjectures; we will discuss  similar determinisation heuristics for HD Büchi automata in detail later in \cref{subsec:determinising-nondeterministic-comp}. We also note that since every accepting run in $\Astrong$ must take the significant transition $Y\xRightarrow{y}I$ infinitely often, having this be the only \significant transition will not change the language. The additional significant transitions in $\Astrong$ are nevertheless present to ensure that $\Astrong$ is normal, which is necessary for $\Astrong$ to be simplified.
\end{remark}

Although not necessary for establishing the succinctness of our 65-state HD Büchi automaton (\cref{thm:main}), we provide full proofs of the next result, as this is a similar, but simpler, version of the proof of correctness for our main automaton in~\cref{subsec:full-main}.

\begin{lemma}\label{lemma:astrong-simplified}
The automaton $\Astrong$ recognises the language $\Lstrong$. Moreover, $\Astrong$ is history deterministic and simplified.   
\end{lemma}
\begin{proof}
To show $\Astrong$ is simplified and HD, we need to show that $\Astrong$ is semantically deterministic and has reach-covering (which implies HDness by Lemma~\ref{lemma:buchi-simplified-implies-hd}). We start by showing that $\Astrong$ is SD and recognises $\Lstrong$ in \cref{subclaim-sd-astrong}. We will then show that $\Astrong$ has reach-covering in \cref{claim:astrong-RC}, which will prove the lemma.
We start with a remark about the languages $L_\alpha$.
\begin{claim}\label{subclaim-sd-Lalpha}
Let $u\in \Sigma^*$ ending by $x\alpha$, with $\alpha \in \{a,b,c\}$. Then, $u\in L_\alpha$ if and only if it does not contain any other factor $x\beta$, for $\beta \in \{a,b,c\}$.
In particular, an infinite word $w\in \Sigma^\omega$ containing infinitely many factors in $x(a+b+c)$ factorises in a unique way in $(L_a + L_b + L_c)^\omega$.
\end{claim}

\begin{claim}\label{subclaim-sd-astrong}
$L(\Astrong) = \Lstrong$ and $\Astrong$ is semantically deterministic.
\end{claim}
\begin{claimproof}
First, we observe that $\Lstrong$ is \emph{prefix-independent}; that is, for every letter $\sigma\in \Sigma$, $\sigma^{-1}\Lstrong=\Lstrong$. This follows from the definition of $\Lstrong$. 

We will next prove that for all states~$q$ in~$\Astrong$, we have that $\Lstrong \subseteq L(\Astrong,q)$.
Towards this, let $w$ be a word in $\Lstrong$. We will build an accepting run from $q$ on $w$ inductively, by proving that  there is a run from $q$ on $w$ that contains at least one accepting transition. To show this, it suffices to show $\Lstrong\subseteq L(\reach(\Astrong),q)$ for each state $q$.
It then follows, from the prefix-independence of $\Lstrong$, and induction that there is a run from $q$ on $w$ that contains infinitely many \significant transitions, as desired.

Since $w$ contains infinitely many $\purpley$'s, an arbitrary run from $q$ eventually reaches one of the states $I,p_a,p_b,p_c,q_a,q_b,q_c$ (after reading a $y$). Then, we have the following cases.

\begin{enumerate}
    \item From $I$. Note that the word $w$ contains a prefix $u$ such that $u$ ends in a suffix $L_{\alpha}$, for some $\alpha$. That is, $w\in \Sigma^{*}(L_{a}\cup L_{b}\cup L_c) \cdot \Sigma^{\omega}$. We claim that 
    \[ L(\reach(\Astrong),I) \supseteq \Sigma^{*}(L_{a}\cup L_{b}\cup L_c)  \cdot \Sigma^{\omega}.\]
    Indeed, a word in $L_\alpha$ ends by $x\alpha$. A run from $I$ can simply stay in the states $I$ and $\iota_\alpha$, and whenever $x\alpha$ is read, it will reach a significant transition.

    \item From $p_\alpha$. Note that the word $w$ contains a prefix $u$ such that $u$ ends in a suffix $L_{\beta}$, for some $\beta\neq \alpha$. That is, $w\in \Sigma^* \cdot \bigcup_{\beta\neq\alpha}L_{\beta}\cdot \Sigma^{\omega}$. We claim that 
    \[L(\reach(\Astrong),p_{\alpha}) \supseteq \Sigma^*  \cdot \bigcup_{\beta\neq\alpha}L_{\beta}  \cdot \Sigma^{\omega}.\] 
    As before, a word in $L_\beta$ ends by $x\beta$. A (deterministic) run from $p_\alpha$ stays in the states $p_\alpha, p_\alpha'$, and reaches a significant transition whenever $x\beta$ is read.
    
    \item From $q_{\alpha}$. Since $w$ is in $\Lstrong$, by \cref{subclaim-sd-Lalpha}, $w$ has a unique decomposition as $w=u_1 u_2 \dots$, where each $u_i\in L_{\alpha_i}$ for some $\alpha_i\in\{a,b,c\}.$
    Note that, from $q_{\beta}$ for $\beta\in\{a,b,c\}$, upon reading a finite word $L_{\beta}$, the run takes a significant transition, or, if $\beta'\neq \beta$, upon reading a finite word $L_{\beta'}$, the run reaches $q_{\beta'}$. Therefore, the deterministic run on $w$ from $q_{\alpha}$ until a significant transition occurs is of the form:
    \[q_\alpha \xrightarrow{u_1} q_{\alpha_1} \xrightarrow{u_2} q_{\alpha_2} \dots \]
    Whenever $\alpha_i = \alpha_{i+1}$, this run reaches a significant transition. But note that since $w$ is in $\Lstrong$, $w$ contains infinitely many infixes in $L_{\beta}L_{\beta}$ for some $\beta\in\{a,b,c\}$, and thus $\alpha_i = \alpha_{i+1}$ for some $i$.
\end{enumerate}
Thus, we conclude that $\Lstrong\subseteq L(\Astrong)$.


For the other language-inclusion $L(\Astrong,q) \subseteq \Lstrong$ for every state $q$, observe that every accepting run in $\Astrong$ must take the transition $Y\xRightarrow{y} I$ infinitely often. Let $L_{\mathsf{infix}}$ be the~set of finite words on which there is a run that starts at $I$ and ends when the transition~$Y\xrightarrow{y} I$ occurs for the first time. 
Note that from $I$ to $Y$, each time that the automaton takes a significant transition, the prefix that has been read ends by $x\alpha$, for some $\alpha\in \{a,b,c\}$. Using \Cref{subclaim-sd-Lalpha} and decomposing the finite words of $L_{\mathsf{infix}}$ in such factors, we observe that: 
$$L_{\mathsf{infix}} = \Sigma^{*} \cdot \bigcup_{\alpha\neq\beta\in \{a,b,c\}} L_{\alpha\beta\beta}\, \cdot  \Sigma^{*}\cdot y.$$
Thus, $L(\Astrong,q) \subseteq L_{\mathsf{infix}}^{\omega} = \Lstrong$, as desired.
\end{claimproof}

\begin{claim}\label{claim:astrong-RC}
The automaton $\Astrong$ has reach-covering.
\end{claim}
\begin{claimproof}
Observe that every state in $\Astrong$ apart from $I,\iota_a,\iota_b,\iota_c$ is \good. We show that \begin{enumerate}
    \item $I$ simulates $q_a$ in $\reach(\Astrong)$, and
    \item $\iota_\alpha$ simulates $q'_{\alpha}$  in $\reach(\Astrong)$, for each $\alpha\in\{a,b,c\}$.
\end{enumerate}
This implies reach-covering for $\Astrong$. 

To see this, note that the only nondeterministic choice appears on state $I$, over the letter $x$. Eve has the following strategy to win the simulation-game: If Adam's token is in $q_{\alpha}$ and Adam gives the letter $x$, then Eve takes the transition on $x$ to $\iota_{\alpha}$. If Eve uses this strategy and the simulation game in $\reach(\Astrong)$ starts from the starting positions $(q_a,I)$ or $(q_\alpha',\iota_\alpha)$, then
\begin{enumerate}
    \item whenever Eve's state is in $I$, Adam's state is $q_{\alpha}$ for some $\alpha$ in $\{a,b,c\}$, and
    \item Eve's state is in $\iota_{\alpha}$ if and only if Adam's state is in $q'_{\alpha}$.
\end{enumerate} 
In particular, the only way for Adam's token to reach the accepting sink state in $\Aa_{\reach}$ is to read $\alpha$ from $q'_{\alpha}$, but then Eve's token at $\iota_{\alpha}$ in the simulation game would also reach the accepting sink state. Eve's strategy is therefore winning in the simulation games in $\reach(\Astrong)$, and it follows that $\Astrong$ has reach-covering. 
\end{claimproof}
We conclude that $\Astrong$ is simplified, and therefore, due to \cref{lemma:buchi-simplified-implies-hd}, $\Astrong$ is HD. 
\end{proof}
The following lemma then disproves the strong\=/rewiring conjecture.
\begin{lemma}\label[lemma]{lemma:no-weak-rewiring}
    Let $s$ be a \good state in $\Astrong$.
    Let $\Bb$ be the automaton obtained by replacing the transition $ Y \xRightarrow{y}I$ in $\Astrong$ by the transition $Y\xRightarrow{y} s$, with the initial state set as some state that was \good in $\Astrong$.
    Then, 
    $$L(\Bb) \neq L(\Astrong).$$
\end{lemma}

\begin{proof}
First, we observe that no matter the state $s$, it is still reachable from the initial state of $\Bb$: this is the case because $Y$ is reachable from every state in $\Bb$, and $s$ is reachable from $Y$ using the added transition $Y\xRightarrow{y}s$. We thus let $u$ be a finite word on which there is a run of $\Bb$ that reaches $s$.  

We will show that there is a word accepted by $\Bb$ that is rejected by $\Astrong$. 
We have the following three cases for the state~$s$:
    \begin{description}
    \item[Case 1.] $s = Y$. We note that $u y^{\omega}$ is accepted by $\Bb$, but is not in $\Lstrong$.
    \item[Case 2.]$s = p_{\alpha}$ or $s= p'_{\alpha}$. For every $\beta\neq \alpha$ in $\{a,b,c\}$, $\Bb$ accepts the word $u (y \cdot x \beta \cdot x \beta \cdot y)^{\omega}$, which is not in $\Lstrong$.
    \item[Case 3.]$s = q_{\alpha}$ or $q'_{\alpha}$. 
    We note that the word
    $u \cdot(x\alpha y)^{\omega}$ is accepted by $\Bb$ but is not in $\Lstrong$.\qedhere
\end{description}
\end{proof}

\cref{lemma:astrong-simplified,lemma:no-weak-rewiring} together imply that the strong-rewiring conjecture is false.

Our counterexample for the strong-rewiring conjecture, however, is not a counterexample for the weak-rewiring conjecture. We present a language-equivalent rewiring $\Dstrong$ in \cref{fig:strong-but-no-weak}.

\begin{figure}[h]
\centering
{\footnotesize\thinmuskip=1.5mu 
        \begin{tikzpicture}[
    base dot/.style={
        inner sep=0pt,      
        minimum size=5pt,   
        fill=blue!60!black, 
        font=\scriptsize    
    },
    cstate/.style={base dot, circle},    
    sstate/.style={base dot, rectangle}, 
    node distance=1.5cm,
    tight label/.style={auto, inner sep=0.5pt, font=\scriptsize, text=black}
]
        \tikzset{every state/.style = {inner sep=-3pt,minimum size =15}}

    \node[cstate, fill=red!40, red!40] (A1) at (-2.3, 2) {AAA};
    \node[cstate, fill=red!40, red!40] (B1) at (0, 2) {AAA};
    \node[cstate, fill=red!40, red!40] (C1) at (2.3,2) {AAA};
    \node (d1) at (A1) {$p_a$};
    \node (d2) at (B1) {$p_b$};
    \node (d3) at (C1) {$p_c$};
    \node[cstate,label=left:$\color{blue!60!black}p_a'$] (A2) at (-2.3,1) {};
     \node[cstate,label=left:$\color{blue!60!black}p_b'$] (B2) at (0, 1) {};
     \node[cstate,label=right:$\color{blue!60!black}p_c'$] (C2) at (2.3, 1) {};

    \node[sstate] (s1) at (-3,0) {B};
    \node[sstate] (s2) at (0,0) {B};
    \node[sstate] (s3) at (3, 0) {B};

    \node[cstate,label=left:$\color{blue!60!black}q_a'$] (m1) at (-3, -1) {A};
    \node[cstate,label=left:$\color{blue!60!black}q_b'$] (m2) at (0, -1) {B};
    \node[cstate,,label=left:$\color{blue!60!black}q_c'$] (m3) at (3,-1) {C};

    \node[cstate, fill=red!40, red!40] (final) at (0,-3.5) {AAA};
    \node (Imp) at (final) {$p_a$};
    
    \node[sstate,label=above:$q_a$] (Q1) at (-3,0) {B};
    \node[sstate,label=above:$q_b$] (Q2) at (0,0) {B};
    \node[sstate,label=above:$q_c$] (Q3) at (3,0) {B};
    
    \node[sstate,label=below:$Y$] (QB1Y) at (-4,-2) {B};
    \node[sstate,label=below:$q_b$] (QB12) at (-3, -2) {B};
    \node[sstate,label=below:$q_c$] (QB13) at (-2, -2) {B};

    \node[sstate,label=below:$q_a$] (QB21) at (-1,-2) {B};
    \node[sstate,label=left:$Y$] (QB2Y) at (0, -2.5) {B};
    \node[sstate,label=below:$q_c$] (QB23) at (1, -2) {B};
    
    \node[sstate,label=below:$q_a$] (QB31) at (2,-2) {B};
    \node[sstate,label=below:$q_b$] (QB32) at (3, -2) {B};
    \node[sstate,label=below:$Y$] (QB3Y) at (4, -2) {B};

     \path[-stealth]
         (A1) edge [purple, loop left] node [left] {$y{\color{black},a,b,c}$} (A1)
         (B1) edge [purple, loop left] node [left] {$y{\color{black},a,b,c}$} (B1)
         (C1) edge [purple, loop left] node [left] {$y{\color{black},a,b,c}$} (C1)

         (A1) edge  node [right] {$x$} (A2)
         (B1) edge  node [right] {$x$} (B2)
         (C1) edge  node [left] {$x$} (C2)

         (A2) edge [bend left=15] node [left] {$a$, $\color{purple}y$} (A1)
         (B2) edge [bend left=15] node [left] {$b$,$\color{purple}y$} (B1)
         (C2) edge [bend right=15] node [right] {$c$,$\color{purple}y$} (C1)

        (A2) edge [thick, green!50!black, double] node [above, pos=0.4, rotate = 25] {$b,c$} (B1)
        (B2) edge [thick, green!50!black, double] node [above, pos=0.4, rotate = 25] {$a,c$} (C1)

        (C2) edge [thick, green!50!black, double] node [above, pos=0.6, rotate = 10] {$a$} (s1)
        (C2) edge [thick, green!50!black, double] node [below, rotate = 20] {$b$} (s2)
        ;

        \path[-stealth]
  
         (s1) edge [loop left, purple] node [left] {$y{\color{black},a,b,c}$} (s1)
         (s1) edge  node [right] {$x$} (m1)
               
         (m1) edge [purple,bend left] node [left] {$y$} (s1)
         (m1) edge [thick, green!50!black, double] node [left] {$a$} (QB1Y)
         (m1) edge  node [left] {$b$} (QB12)
         (m1) edge  node [left] {$c$} (QB13) 

         (s2) edge [loop left, purple] node [left] {$y{\color{black},a,b,c}$} (m2)
         (s2) edge  node [right] {$x$} (m2)
         (m2) edge [purple,bend left] node [left] {$y$} (s2)
         (m2) edge  node [left] {$a$} (QB21)
         (m2) edge  node [left] {$c$} (QB23)
         (m2) edge [thick, green!50!black, double] node [left] {$b$} (QB2Y)

         (s3) edge [loop left, purple] node [left] {$y{\color{black},a,b,c}$} (m3)
         (s3) edge  node [right] {$x$} (m3)
         (m3) edge [purple, bend left] node [left] {$y$} (s3)
         (m3) edge  node [left] {$a$} (QB31)
         (m3) edge  node [left] {$b$} (QB32)
         (m3) edge [thick, green!50!black, double] node [left] {$c$} (QB3Y)

         (QB2Y) edge [loop, in=-30,out=-60,looseness=20] node [right] {$x,a,b,c$}
         (QB2Y) edge [thick, green!50!black, double] node [left] {$y$} (final)
         ;
    \end{tikzpicture}
}
\caption{A language-equivalent rewiring $\Dstrong$ for $\Astrong$. }\label{fig:strong-but-no-weak}
\end{figure}
\begin{restatable}{lemma}{strongbutnotweak}
    \label{lemma:strong-but-not-weak}
    The automaton $\Dstrong$ is a language-equivalent rewiring of $\Astrong$.
\end{restatable}

\begin{proof}
The fact that $\Dstrong$ is a rewiring of $\Astrong$ follows from the definition.

We can prove that every word in $\Lstrong$ is also in $L(\Dstrong)$ similar to how we proved $\Astrong$ is semantically deterministic in \cref{lemma:astrong-simplified}. Indeed, for every state $q$, $\reach(\Dstrong,q)$ and $\reach(\Astrong,q)$ are structurally-isomorphic. 
For the other direction,  consider the following languages:

$$ L_{\neq} = \bigcup_{\alpha\neq\beta\in\{a,b,c\}} L_{\alpha} L_{\beta} \;\;, \;\;  L_{=} = \bigcup_{\alpha\in\{a,b,c\}} L_{\alpha} L_{\alpha}, $$
and
$$ L_{{\mathtt{Inf}}(\neq)} = (\Sigma^* \cdot L_{\neq})^{\omega} \;\;, \;\;  L_{{\mathtt{Inf}}(=)} = (\Sigma^* L_{=})^{\omega}   \;\;, \;\; L_{{\mathtt{Inf}}(y)} = (\Sigma^* y)^{\omega}.$$

Then, note that  
$$\Lstrong = L_{{\mathtt{Inf}}(\neq)} \cap  L_{{\mathtt{Inf}}(=)}  \cap L_{{\mathtt{Inf}}(y)}.$$
We show that every word $w\in L(\Dstrong)$ is indeed in each of these three languages.

For two states $s,s'$ in $\Dstrong$, we let $$L(s,s') = \{\text{Words that have a run from }s \text{ to }s'\}.$$ 

\begin{enumerate}
    \item $L(\Dstrong) \subseteq L_{{\mathtt{Inf}}(\neq)}$. This is true because every accepting  run in $\Dstrong$ must visit $p_a$ infinitely many times, as well as $q_a$ or $q_b$ infinitely many times. We observe that  $$L(p_a,q_a)\cup L(p_a,q_b) \subseteq \Sigma^* \cdot L_{\neq} \cdot \Sigma^*$$
    \item $L(\Dstrong) \subseteq L_{{\mathtt{Inf}}(=)}$. Observe that this follows because $$L(p'_c,Y)\subseteq \Sigma^{*} \cdot L_{=} \cdot \Sigma^*,$$ and every accepting run in $\Dstrong$ must visit $p'_c$ and $Y$ infinitely often.
    \item $L(\Dstrong)\subseteq L_{{\mathtt{Inf}}(y)}$. This follows because $$L(Y,p_a) = \Sigma^* \cdot y \cdot \Sigma^*,$$ and every accepting run in $\Dstrong$ must visit $p_a$ and $Y$ infinitely often.
\end{enumerate}
We conclude that $L(\Dstrong)=L(\Astrong)$.
\end{proof}

We will next build upon $\Astrong$  to disprove the weak-rewiring conjecture (\cref{conj:weak-rewiring}).

\paragraph{Weak rewiring}
We  now present the automaton $\Aweak$ shown in \cref{fig:weak-rewiring-counterexample}, which we will prove to be a counterexample to the weak-rewiring conjecture. The automaton $\Aweak$ is obtained by modifying $\Astrong$ as follows.
\begin{enumerate}
    \item The \significant transitions from $p'$s to $q$s are made non-\significant. 
    \item For these transition to remain non-significant when the automaton is made normal, we add three new \emph{reset letters} $\reset{a},\reset{b},$ and $\reset{c}$. The transitions in $\Aweak$ over these reset letters look like the following:
    \begin{enumerate}
        \item From every state $s$ in $\{I,\iota_a,\iota_b,\iota_c\}$ we have a \significant transition $s\xRightarrow{\reset{\alpha}}p_\alpha$ for each $\alpha\in \{a,b,c\}$.
        \item From state $Y$, we have non-significant self-loops on the reset letters.
        \item From every other state $s$, we have a non-significant transition $s\xrightarrow{\reset{\alpha}} p_\alpha$, for each $\alpha\in\{a,b,c\}$.
    \end{enumerate}
\end{enumerate}
\begin{figure}[h]
\centering
{\footnotesize\thinmuskip=1.5mu 
  \begin{tikzpicture}[
    base dot/.style={
        inner sep=0pt,      
        minimum size=5pt,   
        fill=blue!60!black, 
        font=\scriptsize    
    },
    cstate/.style={base dot, circle},    
    sstate/.style={base dot, rectangle}, 
    node distance=1.5cm,
    tight label/.style={auto, inner sep=0.5pt, font=\scriptsize, text=black}
]
        \tikzset{every state/.style = {inner sep=-3pt,minimum size =15}}

    \filldraw [fill=pink!10, draw=pink] (-2.6,1.25) rectangle (2.5,-1.25);

    \filldraw [fill=cyan!30, draw=cyan!30] (-4,-1.75) rectangle (4.1,-7.5);

    \node[cstate, fill=red!40, red!40] (q0) at (0,0) {AAA};
    \node (q0fake) at (q0) {$I$};
    \node[cstate,label=left:$\color{blue!60!black}\iota_a$] (A) at (-2, -1) {A};
    \node[cstate,label=left:$\color{blue!60!black}\iota_b$] (B) at (0, -1) {B};
    \node[cstate,label=right:$\color{blue!60!black}\iota_c$] (C) at (2,-1) {C};
    \node[cstate, fill=red!40, red!40] (A1) at (-2, -2) {AAA};
    \node[cstate, fill=red!40, red!40] (B1) at (0, -2) {AAA};
    \node[cstate, fill=red!40, red!40] (C1) at (2,-2) {AAA};

    \node[cstate,label=left:$\color{blue!60!black}p_a'$] (A2) at (-2,-3) {};
    \node[cstate,label=left:$\color{blue!60!black}p_b'$] (B2) at (0, -3) {};
    \node[cstate,label=left:$\color{blue!60!black}p_c'$] (C2) at (2, -3) {};

    \node[sstate,label=below:$\color{blue!60!black}q_b$] (Q12) at (-2.6,-4) {B};
    \node[sstate,label=below:$\color{blue!60!black}q_c$] (Q13) at (-1.5, -4) {B};
    \node[sstate,label=below:$\color{blue!60!black}q_a$] (Q21) at (-0.6, -4) {B};
    \node[sstate,label=below:$\color{blue!60!black}q_c$] (Q23) at (0.6,-4) {B};
    \node[sstate,label=below:$\color{blue!60!black}q_a$] (Q31) at (1.5, -4) {B};
    \node[sstate,label=below:$\color{blue!60!black}q_b$] (Q32) at (2.6, -4) {B};
    
    \node (d1) at (A1) {$p_a$};
    \node (d2) at (B1) {$p_b$};
    \node (d3) at (C1) {$p_c$};

    \path[-stealth] 
    (-2.6,-1) edge [thick, green!50!black, double, out=200,in=180,looseness=3] node [left] {$\reset{a}$} (A1)
    ( 0.35, -1.25) edge [thick, green!50!black, double, out=-75,in=45,looseness=1] node [above, right] {$\reset{b}$} (B1)
    (2.5,-1) edge [thick, green!50!black, double, out=-20,in=0,looseness=3] node [right] {$\reset{c}$} (C1)
    
    (-1.6,-1.75) edge [out=90,in=0,looseness=8] node [right] {$\reset{a}$} (A1)
    ( -0.5, -1.75) edge [out=90,in=180,looseness=7] node [xshift=0.15cm,right] {$\reset{b}$} (B1)
    (1.6,-1.75) edge [out=90,in=180,looseness=8] node [xshift=0.2cm,right] {$\reset{c}$} (C1)
    
    ;

     \path[-stealth]
        (q0) edge [purple,loop above] node [above] {$y{\color{black},a,b,c}$} (q0)
         (q0) edge [bend left=15] node [above] {$x$} (A)
         (A) edge [bend left=15] node [above] {$b,c, \color{purple}y$} (q0)
         (q0) edge [bend right=15] node [left] {$x$} (B)
         (B) edge [bend right=15] node [right,xshift=-0.1cm] {$a,c$,} 
         node [right,yshift=-0.25cm] {$\color{purple}y$}
         (q0)
         (q0) edge [bend right=15] node [above] {$x$} (C)
         (C) edge [bend right=15] node [above] {$a,b, \color{purple}y$} (q0)
         (A) edge [thick, green!50!black, double] node [left] {$a$} (A1)
         (B) edge [thick, green!50!black, double] node [right] {$b$} (B1)
         (C) edge [thick, green!50!black, double] node [right] {$c$} (C1)

       (A) edge [loop above] node [above] {$x$} (A)
       (B) edge [loop, out=0, in=-30,looseness=20] node [right] {$x$} (B)
       (C) edge [loop above] node [above] {$x$} (C)

       (A2) edge [loop right] node [right] {$x$} (A2)
       (B2) edge [loop right] node [right] {$x$} (B2)
       (C2) edge [loop right] node [right] {$x$} (C2)

       (A1) edge [loop, out=300, in=330, looseness=8] node [right,yshift=+0.14cm] {$\purpley,a,$} node [right,yshift=-0.14cm] {$b,c$} (A1)
        (B1) edge [loop, out=300, in=330, looseness=8] node [right,yshift=+0.14cm] {$\purpley,a,$} node [right,yshift=-0.14cm] {$b,c$} (B1)
        (C1) edge [loop, out=300, in=330, looseness=8] node [right,yshift=+0.14cm] {$\purpley,a,$} node [right,yshift=-0.14cm] {$b,c$} (C1)

         (A1) edge [bend left=15] node [right,xshift=-0.1cm] {$x$} (A2)
         (B1) edge [bend left=15] node [right,xshift=-0.1cm] {$x$} (B2)
         (C1) edge [bend left=15] node [right,xshift=-0.1cm] {$x$} (C2)

         (A2) edge [bend left=15] node [left] {$a, \color{purple}y$} (A1)
         (B2) edge [bend left=15] node [left] {$b, \color{purple}y$} (B1)
         (C2) edge [bend left=15] node [left] {$c, \color{purple}y$} (C1)

         (A2) edge  node [left] {$b$} (Q12)
         (A2) edge  node [right] {$c$} (Q13)
         (B2) edge  node [left] {$a$} (Q21)
         (B2) edge  node [right] {$c$} (Q23)
         (C2) edge  node [left] {$a$} (Q31)
         (C2) edge  node [right] {$b$} (Q32)
         ;
    
    \node[sstate] (s1) at (-3,-5) {B};
    \node[sstate] (s2) at (0,-5) {B};
    \node[sstate] (s3) at (3, -5) {B};

    \node[cstate,label=left:$\color{blue!60!black}q_a'$] (m1) at (-3, -6) {A};
    \node[cstate,label=left:$\color{blue!60!black}q_b'$] (m2) at (0, -6) {B};
    \node[cstate,label=left:$\color{blue!60!black}q_c'$] (m3) at (3,-6) {C};

    \node[sstate, label=north west:$Y$] (Y) at (0,-8) {A};
    \node[cstate, fill=red!40, red!40] (I) at (0,-9) {AAA};
    \node (Imp) at (I) {$I$};
    
    \node[sstate,label=above:$\color{blue!60!black}q_a$] (Q1) at (-3,-5) {B};
    \node[sstate,label=above:$\color{blue!60!black}q_b$] (Q2) at (0,-5) {B};
    \node[sstate,label=above:$\color{blue!60!black}q_c$] (Q3) at (3,-5) {B};
    
    \node[sstate,label=below:$\color{blue!60!black}q_b$] (QB12) at (-3, -7) {B};
    \node[sstate,label=below:$\color{blue!60!black}q_c$] (QB13) at (-2, -7) {B};

    \node[sstate,label=below:$\color{blue!60!black}q_a$] (QB21) at (-1,-7) {B};
    \node[sstate,label=below:$\color{blue!60!black}q_c$] (QB23) at (1, -7) {B};
    
    \node[sstate,label=below:$\color{blue!60!black}q_a$] (QB31) at (2,-7) {B};
    \node[sstate,label=below:$\color{blue!60!black}q_b$] (QB32) at (3, -7) {B};

        \path[-stealth]

         (s1) edge [loop right] node [right] {$\purpley,a,b,c$} (s1)
         (s1) edge  node [right] {$x$} (m1)
               
         (m1) edge [purple,bend left] node [left] {$y$} (s1)
         (m1) edge [thick, green!50!black, double,out=240,in=180,looseness=1.5] node [below] {$a$} (Y)
         (m1) edge  node [right] {$b$} (QB12)
         (m1) edge  node [right] {$c$} (QB13) 

         (s2) edge [loop right] node [right] {$\purpley,a,b,c$} (m2)
         (s2) edge  node [right] {$x$} (m2)
         (m2) edge [purple,bend left] node [left] {$y$} (s2)
         (m2) edge  node [left] {$a$} (QB21)
         (m2) edge  node [left] {$c$} (QB23)
         (m2) edge [thick, green!50!black, double] node [left] {$b$} (Y)

         (m1) edge [loop right] node [right] {$x$} (m1)
        (m2) edge [loop right] node [right] {$x$} (m2)
        (m3) edge [loop right] node [right] {$x$} (m3)
        
         (s3) edge [loop right] node [right, yshift=+0.1cm] {$\purpley,a,$}  node [yshift=-0.15cm] {$b,c$}(m3)
         (s3) edge  node [right] {$x$} (m3)
         (m3) edge [purple, bend left] node [left] {$y$} (s3)
         (m3) edge  node [left] {$a$} (QB31)
         (m3) edge  node [left] {$b$} (QB32)
         (m3) edge [thick, green!50!black, double,out=300,in=0,looseness=1.5] node [below] {$c$} (Y)

         (Y) edge [loop, out=290, in=330, looseness=20] node [right] {$x,a,b,c, \reset{a},\reset{b},\reset{c}$}
         (Y) edge [thick, green!50!black, double] node [left] {$y$} (I)
         ;
    \end{tikzpicture}
}
\caption{HD Büchi automaton $\Aweak$ for which the weak rewiring conjecture is not true.
The transitions $\xrightarrow{\reset{i}} p_i$ outgoing from a box indicate that from every state in that box, there is a transition on $\reset{i}$ to $p_i$.}
    \label{fig:weak-rewiring-counterexample}
\end{figure}
Changing transitions from the states $p'$s to $q$s to be non-\significant and adding the reset letters to obtain $\Aweak$ from $\Astrong$ has the effect of `gluing' every \good state apart from $Y$ together  in a single strongly connected component in the automaton $\reach(\Aweak)$ (blue box in Fig.~\ref{fig:weak-rewiring-counterexample}). This eliminates the possibility of a language-equivalent rewiring similar to $\Dstrong$ in \cref{fig:strong-but-no-weak}, as we will prove in \cref{lemma:weak-rewiring-dead}.

We begin by describing the language of $\Aweak$. Let  $\Sigma' = \{x,a,b,c,y, \allowbreak \reset{a}, \allowbreak \reset{b}, \allowbreak\reset{c}\}$ and 
\[L'_{\alpha\beta\beta } = (L_{\alpha}+\reset{\alpha})L_{\beta}L_{\beta} \; , \; \text{ for } \alpha,\beta\in\{a,b,c\}, \]
where $L_\alpha = ((a+b+c+\purpley)^* (\varepsilon + x^+\purpley) )^*  x^+ \alpha$ as previously. Then, 
$$\Lweak = ((x+a+b+c+y)^{*} \cdot  \bigcup_{\alpha\neq\beta\in\{a,b,c\}} L'_{\alpha\beta\beta} \cdot  \Sigma'^{*}\cdot y)^{\omega}.$$


\begin{lemma}\label{lemma:weak-rewiring-dead}
The weak-rewiring conjecture is not true for $\Aweak$.
\end{lemma}
\begin{proof}[Proof sketch]
The proof that $\Aweak$ is simplified  and recognises $\Lweak$ is  similar to that for $\Astrong$. We next show that there is no language-equivalent rewiring for $\Aweak$. 

Let $\Bb$ be a rewiring of $\Aweak$. We  first show that if $\Bb$ has a significant transition of the form $q'_{\alpha}\xRightarrow{\alpha} s$ for some $\alpha \in \{a,b,c\}, s \neq Y$, then $\Bb$ is not language-equivalent to $\Aweak$. 

We argue that $s$ is reachable from the initial state of $\Bb$.
If $Y$ is the initial state, then the significant transition from $Y$ on $y$ must not be redirected to $Y$, since otherwise $\Bb$ recognises $y^{\omega}\notin \Lweak$. Otherwise, since $s$ is a \good state that is not $Y$, we can deduce by the structure of non-significant transitions in $\Bb$ that $s$ must be reachable from $Y$. Otherwise, if the initial state is a \good state that is not $Y$, then $s$ is reachable from it since $s$ and the initial state are in the same SCC of non-significant transitions. Thus, let $u$ be a word for which there is a run of $\Bb$ to $s$. 

Then $\Bb$ accepts the word $(\reset{\beta} x\alpha x\alpha)^{\omega}$, with $\beta \neq \alpha$, which is not in $\Lweak$ because it does not contain infinitely many $y$. Thus, in this case, $\Bb$ is not language-equivalent to $\Aweak$. 

Thus, for $\Bb$ to be language-equivalent to $\Aweak$, all significant transitions outgoing  $q'_a,q'_b,$ and $q'_c$ must be towards $Y$. Suppose $\Bb$ has the transition $Y\xRightarrow{y}s$ for some state $s$ in $\Bb$. The proof that $\Bb$ then accepts a word that is not recognised by $\Aweak$ is verbatim to the proof of \cref{lemma:no-weak-rewiring}. We conclude that there is no language-equivalent rewiring of $\Aweak$, as desired.
\end{proof}

\subsubsection{Copying \good states}\label{subsec:weird-Y-det}
Ruling out rewirings is not enough. We provide another way of deterministing
by deleting the nondeterministic component altogether and instead duplicate and create 
keeping several copies of the state $Y$, one per they way in which $Y$ is entered. For
$\Aweak$ this works. 

In \cref{fig:wierdYdet}, we present a  deterministic automaton $\Dd_{{\mathtt{weak}}}$ recognising $\Lweak$. 
This automaton is obtained by removing the non-\good states $\{I,\iota_a,\iota_b,\iota_c\}$ from $\Aweak$  and adding two copies of the state $Y$, to have $Y_a,Y_b$, and $Y_c$. 
\begin{figure}[h]
\centering
{\footnotesize\thinmuskip=1.5mu 
  \begin{tikzpicture}[
    base dot/.style={
        inner sep=0pt,      
        minimum size=5pt,   
        fill=blue!60!black, 
        font=\scriptsize    
    },
    cstate/.style={base dot, circle},    
    sstate/.style={base dot, rectangle}, 
    node distance=1.5cm,
    tight label/.style={auto, inner sep=0.5pt, font=\scriptsize, text=black}
]
        \tikzset{every state/.style = {inner sep=-3pt,minimum size =15}}

    \filldraw [fill=cyan!30, draw=cyan!30] (-4,-1.6) rectangle (4.2,-7.5);

    \node[cstate, fill=red!40, red!40] (A1) at (-2, -2) {AAA};
    \node[cstate, fill=red!40, red!40] (B1) at (0, -2) {AAA};
    \node[cstate, fill=red!40, red!40] (C1) at (2,-2) {AAA};
    
    \node (d1) at (A1) {$p_a$};
    \node (d2) at (B1) {$p_b$};
    \node (d3) at (C1) {$p_c$};
    
    \node[cstate,label=left:$\color{blue!60!black}p_a'$] (A2) at (-2,-3) {};
    \node[cstate,label=left:$\color{blue!60!black}p_b'$] (B2) at (0, -3) {};
    \node[cstate,label=left:$\color{blue!60!black}p_c'$] (C2) at (2, -3) {};

    \node[sstate,label=below:$\color{blue!60!black}q_b$] (Q12) at (-2.6,-4) {B};
    \node[sstate,label=below:$\color{blue!60!black}q_c$] (Q13) at (-1.5, -4) {B};
    \node[sstate,label=below:$\color{blue!60!black}q_a$] (Q21) at (-0.6, -4) {B};
    \node[sstate,label=below:$\color{blue!60!black}q_c$] (Q23) at (0.6,-4) {B};
    \node[sstate,label=below:$\color{blue!60!black}q_a$] (Q31) at (1.5, -4) {B};
    \node[sstate,label=below:$\color{blue!60!black}q_b$] (Q32) at (2.6, -4) {B};

    \path[-stealth] 
    (-2.5,-1.6) edge [out=90,in=180,looseness=7] node [left] {$\reset{a}$} (A1)    ( -0.5, -1.6) edge [out=90,in=180,looseness=6] node [xshift=0.25cm,right] {$\reset{b}$} (B1)
    (1.5,-1.6) edge [out=90,in=180,looseness=6] node [xshift=0.25cm,right] {$\reset{c}$} (C1)
    
    ;

     \path[-stealth]
         (A1) edge  node [right] {$x$} (A2)
         (B1) edge  node [right] {$x$} (B2)
         (C1) edge  node [right] {$x$} (C2)

         (A2) edge [bend left=15] node [left] {$a$, $\color{purple}y$} (A1)
         (B2) edge [bend left=15] node [left] {$b$,$\color{purple}y$} (B1)
         (C2) edge [bend left=15] node [left] {$c$,$\color{purple}y$} (C1)

            (A2) edge [loop right] node [right] {$x$} (A2)
            (B2) edge [loop right] node [right] {$x$} (B2)
            (C2) edge [loop right] node [right] {$x$} (C2)

        (A1) edge [loop, out=330, in=360, looseness=8] node [right,yshift=+0.14cm] {$\purpley,a,$} node [right,yshift=-0.14cm] {$b,c$} (A1)
        (B1) edge [loop, out=330, in=360, looseness=8] node [right,yshift=+0.14cm] {$\purpley,a,$} node [right,yshift=-0.14cm] {$b,c$} (B1)
        (C1) edge [loop, out=330, in=360, looseness=8] node [right,yshift=+0.14cm] {$\purpley,a,$} node [right,yshift=-0.14cm] {$b,c$} (C1)

         (A2) edge [thick, green!50!black, double] node [left] {$a$} (Q12)
         (A2) edge [thick, green!50!black, double] node [right] {$b$} (Q13)
         (B2) edge [thick, green!50!black, double] node [left] {$c$} (Q21)
         (B2) edge [thick, green!50!black, double] node [right] {$a$} (Q23)
         (C2) edge [thick, green!50!black, double] node [left] {$b$} (Q31)
         (C2) edge [thick, green!50!black, double] node [right] {$c$} (Q32)
         ;
    
    \node[sstate] (s1) at (-3,-5) {B};
    \node[sstate] (s2) at (0,-5) {B};
    \node[sstate] (s3) at (3, -5) {B};

    \node[cstate, label=left:$\color{blue!60!black}q_a'$] (m1) at (-3, -6) {A};
    \node[cstate, label=left:$\color{blue!60!black}q_b'$] (m2) at (0, -6) {B};
    \node[cstate, label=left:$\color{blue!60!black}q_c'$] (m3) at (3,-6) {C};

    \node[sstate, label=left:$\color{blue!60!black}Y_a$] (Y1) at (-3,-8) {A};
    \node[sstate, label=left:$\color{blue!60!black}Y_b$] (Y) at (0,-8) {A};
    \node[sstate, label=left:$\color{blue!60!black}Y_c$] (Y3) at (3,-8) {A};

    \node[cstate, fill=red!40, red!40] (reA1) at (-3, -9) {AAA};
    \node[cstate, fill=red!40, red!40] (reB1) at (0, -9) {AAA};
    \node[cstate, fill=red!40, red!40] (reC1) at (3,-9) {AAA};
    
    \node (rd1) at (reA1) {$p_a$};
    \node (rd2) at (reB1) {$p_b$};
    \node (rd3) at (reC1) {$p_c$};

    \node[sstate,label=above:$\color{blue!60!black}q_a$] (Q1) at (-3,-5) {B};
    \node[sstate,label=above:$\color{blue!60!black}q_b$] (Q2) at (0,-5) {B};
    \node[sstate,label=above:$\color{blue!60!black}q_c$] (Q3) at (3,-5) {B};
    
    \node[sstate,label=below:$\color{blue!60!black}q_b$] (QB12) at (-3, -7) {B};
    \node[sstate,label=below:$\color{blue!60!black}q_c$] (QB13) at (-2, -7) {B};

    \node[sstate,label=below:$\color{blue!60!black}q_a$] (QB21) at (-1,-7) {B};
    \node[sstate,label=below:$\color{blue!60!black}q_c$] (QB23) at (1, -7) {B};
    
    \node[sstate,label=below:$\color{blue!60!black}q_a$] (QB31) at (2,-7) {B};
    \node[sstate,label=below:$\color{blue!60!black}q_b$] (QB32) at (3, -7) {B};

        \path[-stealth]
        
         (s1) edge [loop right] node [right] {$\purpley,a,b,c$}  (s1)
         (s1) edge  node [right] {$x$} (m1)
               
         (m1) edge [purple,bend left] node [left] {$y$} (s1)
         (m1) edge [thick, green!50!black, double,out=240,in=120,looseness=1.5] node [left] {$a$} (Y1)
         (m1) edge  node [right] {$b$} (QB12)
         (m1) edge  node [right] {$c$} (QB13) 

         (s2) edge [loop right] node [right] {$\purpley,a,b,c$} (m2)
         (s2) edge  node [right] {$x$} (m2)
         (m2) edge [purple,bend left] node [left] {$y$} (s2)
         (m2) edge  node [left] {$a$} (QB21)
         (m2) edge  node [left] {$c$} (QB23)
         (m2) edge [thick, green!50!black, double] node [left] {$b$} (Y)

         (m1) edge [loop right] node [right] {$x$} (m1)
        (m2) edge [loop right] node [right] {$x$} (m2)
        (m3) edge [loop right] node [right] {$x$} (m3)

         (s3) edge [loop right] node [right] {$\purpley,a,$} node [right,yshift=-0.25cm] {$b,c$} (s1) (m3)
         (s3) edge  node [right] {$x$} (m3)
         (m3) edge [purple, bend left] node [left] {$y$} (s3)
         (m3) edge  node [left] {$a$} (QB31)
         (m3) edge  node [left] {$b$} (QB32)
         (m3) edge [thick, green!50!black, double,out=300,in=60,looseness=1.5] node [right] {$c$} (Y3)

         (Y) edge [loop, out=290, in=340, looseness=10] node [right] {$\Sigma\setminus\{{\color{purple}y}\}$} (Y)

         (Y1) edge [loop, out=290, in=340, looseness=10] node [right] {$\Sigma\setminus\{{\color{purple}y}\}$} (Y1)

         (Y3) edge [loop, out=290, in=340, looseness=10] node [right] {$\Sigma\setminus\{{\color{purple}y}\}$} (Y3)

        (Y1) edge [thick, green!50!black, double] node [left] {$y$} (reA1)
        (Y) edge [thick, green!50!black, double] node [left] {$y$} (reB1)
        (Y3) edge [thick, green!50!black, double] node [left] {$y$} (reC1)
         ;
    \end{tikzpicture}
}
\caption{Determinisation of $\Aweak$ by introducing copies of the state $Y$.}
    \label{fig:wierdYdet}
\end{figure}

This determinisation strategy works here because the number of $Y$-copies ($2$ in this case) is smaller than the number of non-\good states ($4$ in this case).
To defeat it, it suffices to enlarge our automata with alphabets similar to $\{a,b,c\}$, so that we would require to introduce $>4$ copies of the state $Y$.

\subsubsection{Determinising the nondeterministic component}\label{subsec:determinising-nondeterministic-comp}
We discuss a final determinisation technique.
Recall the automaton $\Abkks$ in \cref{fig:BKKS13}.
In \cref{fig:detForBKKS13b}, we showed a language-equivalent deterministic Büchi automaton of the same size that is not obtained by a rewiring, but instead by `\emph{determinising the nondeterministic component'}.
It is easy to come up with a similar language-equivalent deterministic Büchi automata for both $\Aweak$ and $\Astrong$, where we replace the states $I,\iota_a,\iota_b,\iota_c$ that are not \good with two states as in \cref{fig:detForBKKS13b}, and add appropriate transitions.
This provides yet a different proof of non-succinctness for $\Aweak$.

To demonstrate that such a determinisation heuristic does not always result in smaller deterministic Büchi automata, we propose the automaton $\Areplace$ in \cref{fig:HDBuchinotMRAgain}.
This~example is inspired by a recent example that arose in a different context~\cite[Figure~6]{HPT25}. The automaton $\Areplace$ recognises the language 
\[((L_1+L_2)^{*} (L_1L_1 + L_2 L_2))^{\omega},\] 
for some appropriate $L_1$ and $L_2$ that are languages over finite words. For comparison, we remark that the language of $\Abkks$ is $(xa+xb)^*(xaxa+xbxb)^{\omega}$.

To describe $L_1$ and $L_2$ concretely, let $A=\{a,b,c\}$, and $\Sigma=\{a,b,c,1,2\}.$  Define the languages $L'_1 = \Sigma^{*}c1$, and $L'_2 = \Sigma^* a  A^* b 2$. Then, $L_1$ is given by 
\begin{equation}\label{eqn:Areplace1}
    L_1 = L'_1 \setminus (L'_2 \Sigma^{*}),
\end{equation}

i.e., the set of words in $L'_1$ with no prefix in $L'_2$. Analogously, $L_2$ is defined as 
\begin{equation}\label{eqn:Areplace2}
    L_2=L'_2\setminus (L'_1 \Sigma^{*}).
\end{equation}

The only nondeterminism in $\Areplace$ is on $I$ and $\iota_c$, over the letter $a$. On these states, upon reading an $a$, we  `guess'  whether we are going to read a word in $L_1$ or $L_2$, and take the transition to $I$ and $\iota_a$, respectively. Consider a resolver makes its first guess arbitrarily, and afterwards, considers the latest occurrence of an infix in $L_1$ or $L_2$ in the word read so far, and uses this to take the transition to the $I$ or $\iota_a$, respectively. The fact that this resolver is an HD-resolver is proved similarly to the case of $\Abkks$ (\cref{ex:BKKS13}).

Determinising the nondeterministic component of $\Areplace$, given by the states $I,\iota_a,\iota_b$ and $\iota_c$, requires 5 states (copying one of the $5$-state components below).
Therefore, such a determinisation strategy yields a deterministic automaton with strictly more states than~$\Areplace$.

There is, of course, a language-equivalent rewiring for $\Areplace$, which is similar in spirit to the language-equivalent rewiring in \cref{fig:detForBKKS13a} for $\Abkks$.
This provides a smaller language-equivalent deterministic automaton.

\begin{figure}[H]
\centering
{\footnotesize\thinmuskip=1.5mu 

        \begin{tikzpicture}[
    base dot/.style={
        inner sep=0pt,      
        minimum size=5pt,   
        fill=red!50!white, 
        font=\scriptsize    
    },
    cstate/.style={base dot, circle},    
    sstate/.style={base dot, rectangle}, 
    node distance=1.5cm,
    tight label/.style={auto, inner sep=0.5pt, text=black}
]
        \tikzset{every state/.style = {inner sep=-3pt,minimum size =15}}

\node[cstate, fill=red!40, red!40] (q1) at (-1,1.5) {AAA};

\node at (q1) {$\iota_c$};

 \node[cstate, fill=red!40, red!40] (q0) [above=0.8cm of q1] {AAA};
\node at (q0) {$I$};

 \node[cstate, fill=red!40, red!40] (q3) [right=0.8cm of q0] {AAA};
\node at (q3) {$\iota_a$};
 \node[cstate, fill=red!40, red!40] (q2) [below=0.8cm of q3] {AAA};
\node at (q2) {$\iota_b$};
 \path[->] (-1.7,2.8) edge (q0);
    
    \node[cstate,label=left:$\color{blue!60!black}p_2$] (r0) at (1.7,0.2) {};
    \node[cstate] (r1) at (1.7,-1.3) {};
    \node[cstate] (r2) at (2.6,-1.3) {};
    \node[cstate] (r3) at (3.2,0.2) {};    
    \node[cstate] (r4) at (4.2,-1.3) {};    

    \node[cstate,label=left:$\color{blue!60!black}p_1$] (l0) at (-3,0.2) {};
    \node[cstate] (l1) at (-3,-1.3) {};
    \node[cstate] (l2) at (-2.2,-1.3) {};
    \node[cstate] (l3) at (-1.5,0.2) {};    
    \node[cstate] (l4) at (-0.5,-1.3) {};    

    \node[cstate, fill=red!40, red!40] (l5) at (-2.8,-2.3) {AAA};  
    \node at (l5) {$I$};
    \node[cstate,label=below:$\color{blue!60!black}p_2$] (l6) at (-0.65,-2.3) {};  

    \node[cstate,label=below:$\color{blue!60!black}p_1$] (l7) at (-1.65,-2.3) {};  
    \node[cstate,label=below:$\color{blue!60!black}p_2$] (r5) at (2.15,-2.3) {};  

    \node[cstate,label=below:$\color{blue!60!black}p_1$] (r7) at (3.15,-2.3) {}; 
    \node[cstate, fill=red!40, red!40] (r6) at (4.2,-2.3) {AAA};    
     \node at (r6) {$I$};    


    \path[-stealth]
    (q0) edge [loop above] node [above] {$a,b,1,2$} (q0)
    (q0) edge [bend left = 8] node [above] {$a$} (q3)
    (q1) edge [] node [below] {$a$} (q3)
    (q3) edge [bend left = 8] node [below] {$1,2$} (q0)
    (q0) edge [bend left = 8] node [right] {$c$} (q1)
    (q1) edge [bend left = 8] node [left] {$b,a,2$} (q0)
    (q3) edge [bend right = 8] node [left] {$b$} (q2)
    (q2) edge [bend right = 8] node [right] {$c,a$} (q3)
    (q2) edge [bend right = 90, looseness = 3] node [above] {$1$} (q0)  
    (q3) edge [loop above] node [right] {$a,c$} (q3)
    (q1) edge [loop left] node [left] {$c$} (q1)
    (q2) edge [loop right] node [right] {$b$} (q2)

    (l0) edge [loop above] node [above] {$b,1,2$} (l0)
    (l0) edge [bend right = 8] node [below] {$a$} (l3)
    (l3) edge [bend right = 8] node [tight label, above] {$1,2$} (l0)
    (l1) edge  node [below, left] {$a$} (l3)
    (l0) edge [bend right = 8] node [tight label,left] {$c$} (l1)
    (l1) edge [bend right = 8] node [tight label,right] {$b$} (l0)
    (l3) edge [bend right = 8] node [tight label,left] {$c$} (l2)
    (l2) edge [bend right = 8] node [tight label,right] {$a$} (l3)
    (l3) edge [bend right = 8] node [tight label,below=3pt,left] {$b$} (l4)
    (l4) edge [bend right = 8] node [tight label,above=3pt,right] {$a$} (l3)
    (l4) edge [bend right = 8] node [tight label,above=0.5pt] {$c$} (l2)
    (l2) edge [bend right = 8] node [tight label,below=0.5pt] {$b$} (l4)
    
    (l3) edge [in=30,out=60,loop] node [right] {$a$} (l3)
    (l1) edge [loop left] node [left] {$c$} (l1)
    (l2) edge [in=-60,out=-90,loop]  node [below] {$c$} (l2)
    (l4) edge [loop above] node [above] {$b$} (l4)

    (r0) edge [loop above] node [above] {$b,1,2$} (r0)
    (r0) edge [bend right = 8] node [below] {$a$} (r3)
    (r3) edge [bend right = 8] node [tight label, above] {$1,2$} (r0)
    (r1) edge  node [below, left] {$a$} (r3)
    (r0) edge [bend right = 8] node [tight label,left] {$c$} (r1)
    (r1) edge [bend right = 8] node [tight label,right] {$b$} (r0)    
    (r3) edge [bend right = 8] node [tight label,left] {$c$} (r2)
    (r2) edge [bend right = 8] node [tight label,right] {$a$} (r3)
    (r3) edge [bend right = 8] node [tight label,below=3pt,left] {$b$} (r4)
    (r4) edge [bend right = 8] node [tight label,above=3pt, right] {$a$} (r3)
    (r2) edge [bend right = 8] node [tight label,below=0.5pt] {$b$} (r4)
    (r4) edge [bend right = 8] node [tight label,above=0.5pt] {$c$} (r2)
    (r3) edge [in=30,out=60,loop]  node [above] {$a$} (r3)
    (r1) edge [loop left] node [left] {$c$} (r1)
    (r2) edge [in=-60,out=-90,loop] node [below] {$c$} (r2)
    (r4) edge [loop above] node [above] {$b$} (r4)
    (r2) edge [purple,out=-30,in=90] node [tight label,right] {$\color{purple}1$} (r7)
    (r1) edge [purple,out=-40,in=180] node [tight label,below,xshift=0.2cm,yshift=-0.15cm] {$\color{purple}1$} (r7)
    (r4) edge [purple] node [tight label,below] {$\color{purple}1$} (r7)

    (q1) edge  node [right] {$1$} (l0)
    (q2) edge node [left] {$2$} (r0)

    (l1) edge[ thick, green!50!black, double] node [left,yshift=0cm,xshift=0cm] {$1$} (l5)
    (l2) edge[thick, green!50!black, double] node [left,pos=0.2] {$1$} (l5)
    (l4) edge node [left] {$2$} (l6)
    (l2) edge [purple,out=-30,in=90] node [tight label,right] {$\color{purple}2$} (l7)
    (l1) edge [purple,out=-40,in=180] node [tight label,below,xshift=0.15cm,yshift=-0.2cm] {$\color{purple}2$} (l7)
    (l4) edge [purple] node [tight label,below] {$\color{purple}1$} (l7)
    
    (r1) edge node [right,xshift=-0.1cm] {$2$} (r5)
    (r2) edge node [left,pos=0.25] {$2$} (r5)
    (r4) edge [thick, green!50!black, double] node [left] {$2$} (r6)
    
;
    \end{tikzpicture}
}
\caption{The HD B\"uchi automaton $\Areplace$.}\label{fig:HDBuchinotMRAgain}
\end{figure}


\vskip1em

We are ready to propose an HD B\"uchi automaton strictly smaller than every language-equivalent deterministic one.
In the next section, we combine our examples to construct the succinct HD Büchi automaton that defeats all these determinisation techniques at the same time. 

\subsection{The main automaton (full details)}\label{subsec:full-main}
In this section, we introduce the protagonist of our succinctness result: a 65-state Büchi automaton $\Amain$, and  prove that it is in fact history deterministic.
In Section~\ref{subsec:full-succinctness}, we will prove that there is no language-equivalent deterministic automaton with at most 65 states. 
\subparagraph{Some intuition.} To build the automaton $\Amain$, we combine the ideas and gadgets explained in the previous section. 
In particular, a succinct HD B\"uchi automaton must satisfy the following (informal) properties,
\begin{description}
    \item[No rewiring.] It does not admit a language-equivalent rewiring. 
    \item[No copying.] A determinisation strategy based on making copies of a $Y$-state requires making more copies than states in the nondeterministic part. 
    \item[No replacement.] The nondeterministic part does not admit a determinisation with fewer states.
\end{description}
We ensure the no-rewiring property by using the ``skeleton structure'' of the automaton $\Aweak$ from \cref{fig:weak-rewiring-counterexample}.
That is, the general structure of the automaton $\Amain$ consists of 
\begin{itemize}
    \item a nondeterministic initial component with $4$ states, 
    \item several deterministic ``basic blocks'', each of them recognising a language of finite words~$L_i$,
    \item a state $Y$ checking for appearances of the letter $\purpley$.
\end{itemize}
Moreover, we include reset letters $\reset{i}$ as in $\Aweak$ to `glue together' the basic blocks in a single SCC of $\reach(\Amain)$. 

We ensure the no-copying property by having $12 = 2\cdot 6$ ``basic blocks'' in the skeleton structure. A determinisation via making copies of a $\purpley$-state would incur in adding $5$ states, more than what we gain from removing the nondeterministic part.

To ensure the no-replacement property, we use the gadget from \cref{fig:HDBuchinotMRAgain}.
Each basic block, therefore, will be a subautomaton of $5$ states similar to the one at the bottom left of \cref{fig:HDBuchinotMRAgain}.

\subparagraph{The language.} Before introducing the automaton, we describe the language it recognises. We fix the alphabet $$\Sigma = \{a,b,c,1,2,3,4,5,6,\reset{1},\reset{2},\reset{3},\reset{4},\reset{5},\reset{6}, {\color{purple}y}\}.$$
For notational convenience, we let $A=\{a,b,c\}$, $E = \{1,2,3\}$, and $F = \{4,5,6\}$.

The language $\Lmain = L(\Amain)$ is built similarly to the language of $\Aweak$. 
The ``building blocks'' of the language are $6$ languages of finite words, $\langBlock{1},\dots, \langBlock{6}$, which we describe below. The languages $\langBlock{1},\langBlock{2},\langBlock{3}$, and the languages  $\langBlock{4},\langBlock{5},\langBlock{6}$ are similar to $L_1$ and $L_2$ from that of $\Areplace$ (see \cref{eqn:Areplace1,eqn:Areplace2}).

We let:
\[ \langInfix{\alpha}{\beta}  =  (\langBlock{\alpha} + \reset{\alpha})\langBlock{\beta}\langBlock{\beta}, \;\; \text{ for } \alpha\neq \beta, \;\; \alpha,\beta\in \{1,\dots,6\}.\] 

Then, the final language is:
\[ \Lmain = (\Sigma^* \cdot \bigcup_{\substack{\alpha\neq \beta \\ 1\leq \alpha,\beta \leq 6}} \langInfix{\alpha}{\beta} \cdot \Sigma^* \cdot y )^\omega  \ .\]

That is, $\Lmain$ consists of the words containing infinitely often infixes in some $\langInfix{\alpha}{\beta}$ as well as containing infinitely many $\purpley$'s. 


It remains to describe the languages $\langBlock{\alpha}$ for $\alpha = 1,\dots,6$.
This is best done by providing the \emph{DFA-classifier} $\Aclassifier$ from Figure~\ref{fig:LanguageL1}.
This is a DFA over $\Sigma$ with multiple final states, in this case, the six states $\sqrstate{\ell_1},\dots, \sqrstate{\ell_6}$.
\begin{figure}
    \centering
\begin{tikzpicture}[
        base dotred/.style={
        inner sep=0pt,      
        minimum size=5pt,   
        fill=red!50!white, 
        font=\scriptsize    
    },
    base dot/.style={
        inner sep=0pt,      
        minimum size=5pt,   
        fill=blue!60!black, 
        font=\scriptsize    
    },
    cstate/.style={base dotred, circle},  
    node distance=1.5cm,
    tight label/.style={auto, inner sep=0.5pt, text=black},
    sstate/.style={base dot, rectangle}, 
    house peak H/.store in=\peakH,
    house peak H=-1.6cm,
    house eaves H/.store in=\eavesH, 
    house eaves H=0.15cm,
    house base H/.store in=\baseH,   
    house base H=-0.7cm,
    house width/.store in=\houseW,   
    house width=0.5cm,
    group sep x/.store in=\gSepX, 
    group sep x=2.6cm, 
    group sep y/.store in=\gSepY, 
    group sep y=3.2cm, 
        row top y/.store in=\rowTopY,
    row bot y/.store in=\rowBotY,
    top x off/.store in=\topX,
    bot x off/.store in=\botX,
    row top y=0.6cm,
    row bot y=-0.9cm,
    top x off=0.6cm,
    bot x off=1.2cm
]
    \filldraw [fill=violet!10, draw=violet!30] (-2.5,3.2) rectangle (2.8,-0.5);

    \filldraw [fill=violet!30, draw=violet!50] (-2.2,0.4) rectangle (0.3,-0.3);
    
\foreach \i in {1} {
    \coordinate (center\i) at ({(\i-1)*\gSepX}, 0);

    \node[cstate, label=above:$p$] (p\i) at (-1,2) {};
    \node[cstate] (q\i) at (1,2)  {};

    \node[cstate] (t\i) at (-1.5, 0) {};
    \node[cstate] (s\i) at (0,0)      {};
    \node[cstate] (r\i) at (2,0)  {};
}
    \node[sstate, label=below:$\color{blue!60!black}\ell_1$] (accept1) at (-2.2,-1.3) {};
    \node[sstate, label=below:$\color{blue!60!black}\ell_2$] (accept2) at (-1.2,-1.3) {};
    \node[sstate, label=below:$\color{blue!60!black}\ell_3$] (accept3) at (-0.2,-1.3) {};
    \node[sstate, label=below:$\color{blue!60!black}\ell_4$] (accept4) at (0.8,-1.3) {};
    \node[sstate, label=below:$\color{blue!60!black}\ell_5$] (accept5) at (1.8,-1.3) {};
    \node[sstate, label=below:$\color{blue!60!black}\ell_6$] (accept6) at (2.8,-1.3) {};




 \foreach \i in {1} {
  \path[->,>=stealth, thick, shorten >=1pt]
    (p\i) edge [loop left] node [above,tight label] {$b,{\color{purple}E,F}$} (p\i)
    (q\i) edge [loop right] node [tight label] {$a$} (q\i)
    
    (r\i) edge [loop right] node [tight label] {$b$} (r\i)
    (t\i) edge [loop left] node [tight label, yshift=0cm, xshift=-0.05cm] {$c$} (t\i)

    (p\i) edge  node [tight label,below] {$a$} (q\i)
    (q\i) edge [purple,bend right] node [tight label,above] {${\color{purple}{E,F}}$} (p\i)

    (q\i) edge [bend left = 10] node [tight label] {$b$} (r\i)   
    (r\i) edge [bend left = 10] node [tight label,yshift=0.1cm, xshift=-.05cm] {$a$} (q\i)
     
    (r\i) edge [bend left = 10] node [tight label] {$c$} (s\i)  
    (s\i) edge [bend left = 10] node [tight label] {$b$} (r\i)       
    (p\i) edge [bend right =10] node [ tight label, left] {$c$} (t\i)
    (t\i) edge [bend right=10] node [tight label, right] {$b$} (p\i)

    (q\i) edge [bend left = 10] node [tight label, yshift=0cm] {$c$} (s\i)  
    (s\i) edge [bend left = 10] node [tight label,yshift=-0.15cm] {$a$} (q\i)   
    (t\i) edge  node [tight label, left, yshift=-0.20cm, xshift=-.35cm] {$a$} (q\i)
    
    (t\i) edge [bend left=40, purple] node [tight label, left] {${\color{purple}F}$} (p\i)

    (r\i) edge [purple,out=60,in=60, looseness=1.8] node [tight label, above,xshift=0.1cm] {{\color{purple}$E$}} (p\i)  
    (s\i) edge [purple] node [tight label,xshift=0.1cm,yshift=0.6cm] {$\color{purple}F$} (p\i)   




    (s\i) edge [loop left] node [tight label,left,xshift=-0.1cm,yshift=0cm] {$c$} (s\i);
    ;

    \path[-stealth] 
    (-1.6,-0.3) edge [thick, green!50!black, out=200,in=90,looseness=1] node [left] {$1$} (accept1) 
    (-1.2,-0.3) edge [thick, green!50!black, out=-90,in=90,looseness=1] node [left] {$2$} (accept2) 
    (-0.8,-0.3) edge [thick, green!50!black, out=-60,in=90,looseness=1] node [right, xshift= 0.10cm] {$3$} (accept3) 

    (r\i) edge [thick, green!50!black, out=-120,in=90,looseness=1] node [left, xshift=-0.10cm] {$4$} (accept4) 
    (r\i) edge [thick, green!50!black, out=-90,in=90,looseness=1] node [left] {$5$} (accept5) 
    (r\i) edge [thick, green!50!black, out=-60,in=90,looseness=1] node [right] {$6$} (accept6) 
    
    (-1.2,3.2) edge [thick, out=110,in=130,looseness=3] node [xshift=0.0cm,left] {$y$} (p\i);;

    
    

}
\end{tikzpicture}
    \caption{The DFA-classifier $\Aclassifier$ defining the languages $L_1,\dots, L_6$.}
    \label{fig:LanguageL1}
\end{figure}
Its initial state is $p$.
The language $L_\alpha$ is the set of words for which the run over $\Aclassifier$ ends up in $\sqrstate{\ell_\alpha}$.
We note that $\sqrstate{\ell_\alpha}$ is \emph{not} a sink: words that admit a proper prefix ending in some $\sqrstate{\ell_\alpha}$ are rejected.
By definition, the sets $L_1,\dots, L_6$ are pairwise disjoint, and furthermore, every word in $L_i$ has no proper prefix in $L_j$, for every $i,j \in [6]$.

We give some intuition for the DFA-classifier $\Aclassifier$. We first remark that $\Aclassifier$ contains no $\reset{\alpha}$ transitions: all words containing such a letter are rejected. Note that $\purpley$ acts as a reset letter: from all states, the transition on $\purpley$ goes back to state $p$.
While no letter $a$ is produced, the automaton stays in the two leftmost states; if the factor $ce$ is read for some $e\in E = \{1,2,3\}$ is then read, it goes to $\sqrstate{\ell_e}$.
When letter $a$ is read, we go to 
the right part of the automaton, and remain there while only letters in  $A = \{a,b,c\}$ appear.
From this right part:
\begin{itemize}
    \item $A^*ce$ make us advance to $\sqrstate{\ell_e}$ for $e\in E$, while  $A^*cF$ resets to $p$,
    \item $A^*bf$ make us advance to $\sqrstate{\ell_f}$ for $f\in F$,  while $A^*bE$ resets to $p$, and
    \item $A^*a(E+F)$ resets to $p$.
\end{itemize}

\subparagraph{The automaton $\Amain$.}
We describe the HD B\"uchi automaton recognising $\Lmain$.
The states of the automaton are:
\begin{itemize}
    \item the initial state $I$, and states $I_a$, $I_b$ and $I_c$;
    \item the ``circular states'': $\circstate{p_i}, \circstate{q_i}, \circstate{r_i}, \circstate{s_i},\circstate{t_i}$ for $i\in \{1,\dots,6\}$;
    \item the ``square states'': $\sqrstate{p_i}, \sqrstate{q_i}, \sqrstate{r_i},\sqrstate{s_i},\sqrstate{t_i}$ for $i\in \{1,\dots,6\}$; 
    \item and finally the  state $\sqrstate{Y}$.
\end{itemize}
In total, it has $4 + 2\times 5 \times 6 + 1 = 65$ states.

\input{figures_sec2and3/beastAutomaton}

\subparagraph{Visualising the automaton.} 
The automaton is pictured in \cref{fig:wholeautomaton}.
To visualise the 65-state automaton without creating a tangle of spaghetti, we split the transitions across two figures and use a few graphical shorthands similar to the earlier section.
The transitions of the automaton over the letters $a,b,c,$ and $\purpley$ are described in \cref{fig:transitionsoverabcy}, whereas the transitions of the automaton over the rest of the letters $1,\dots,6$ and $\reset{1},\dots,\reset{6}$ are depicted in \cref{fig:transitionsoverEFr}. 
Note that to maintain the planarity and readability of the diagrams, some states appear in multiple locations.
We use variables to represent sets of transitions.
\begin{itemize} 
\item The variable $i$ ranges over the set $\{1,\dots,6\}$,
\item the variable $e$ ranges over the set $E = \{1,2,3\}$, and
\item the variable $f$ ranges over the set $F = \{4,5,6\}$.
\end{itemize}
Therefore, a single arrow like $\xrightarrow{f}\circstate{p_f}$ in fact represents the transitions $\xrightarrow{4} \circstate{p_4}$, $\xrightarrow{5} \circstate{p_5}$, and $\xrightarrow{6} \circstate{p_6}$.
Finally, the rectangles around states represent grouped sources. An arrow originating from a box around a set of states implies a transition from every state within that box. For instance, in \cref{fig:transitionsoverabcy}, the arrow labelled $\purpley$ to the state $\circstate{p_i}$   indicates that 
  on input  $\purpley$, the states  $\circstate{p_i}, \circstate{q_i}, \circstate{r_i}, \circstate{s_i}$, and $\circstate{t_i}$ transition to state $\circstate{p_i}$.

\subparagraph{Working of the automaton.} We note that there are two nondeterministic choices: when reading letter $a$ from states $I$ and $I_c$ (similar to the mechanism in Figure~\ref{fig:HDBuchinotMRAgain}).

This automaton works as follows. It reads sequences of words in languages $L_i$ until a factor in $L_\beta L_\beta$ that has been preceded by some word in $L_\alpha$ or $\reset{\alpha}$, with $\alpha\neq \beta$, appears.
Then, it goes to state $Y$ and waits for letter $\purpley$ to appear. When this happens, it restarts again.

The top box (states $I,I_a,I_b,I_c$) waits for a word in $L_\alpha$ to appear (or letter $\reset{\alpha}$). Using the nondeterministic choice on $I$ and $I_c$ over $a$, 
we will guess whether $\alpha$ is in $E$  and we are going to see a word of the form $A^*cE$, in which case we should take the $a$ transition to $I$, or whether $\alpha$ is in $F$ and we are going to see a word of the form $A^*aA^*bF$, in which case we should take the $a$-transition to $I_a$.

There is a HD resolver to resolve this nondeterminism, which we describe informally. The first guess from $I$ will be at random. If we make the right choice, we move to $\circstate{p_\alpha}$.
If we make the wrong choice, no big deal! 
We go back to $I$ and use the following strategy upon reading an $a$ from $I$ or $I_c$. Let $u\in\Sigma^{*}$ be the prefix read so far, and let $\alpha \in \{1,\dots,6\}$ be the index of the latest infix appearing in $u$ belonging to some $L_\alpha$.
Then, we resolve the nondeterminism from $I$ and $I_c$ in the following way: if $\alpha\in E$, we take the $a$-transition to $I$,  if $\alpha \in F$, we take the $a$-transition to $I_a$. This strategy will eventually advance us to the big blue box over every word in $\Lmain$, since every such word must consist of a factor in $L_\beta L_\beta$ for some $\beta \in [6]$.
Full formal details can be found in~\cref{prop:reach-covering-for-amain}.

In the big blue box (containing every state other than $I,I_a,I_b,I_c$, and $Y$), the $\alpha^\text{th}$ box on the top row resets over $L_\alpha$, and moves to the lower row over $L_\beta$ if $\beta\neq \alpha$. 
More specifically, it goes to the state $\circstate{p_\alpha}$ over $L_\alpha$ and to $\sqrstate{p_\beta}$ over $L_\beta$. 
From $\sqrstate{p_\beta}$ in the lower row, it advances to $Y$ over $L_\beta$. Over $L_\alpha$, for $\alpha\neq \beta$, it stays in the lower row, going to $\sqrstate{p_\alpha}$.
Over a letter $\reset{i}$ we reset to the upper row, to $\circstate{p_i}$; this ensures that the  big blue box is strongly connected over non\=/significant transitions.

The automaton has three strongly connected components when restricted to non-significant transitions: the top pink box, the big blue box, and the state $Y$.

\subparagraph{Correctness.} 
The proof of the below result on $\Amain$ is similar to the proof of~\cref{lemma:astrong-simplified}. 

\begin{restatable}{theorem}{thmMainisHD}\label{thm:A65isHD}
    The automaton $\Amain$ is history deterministic and recognises $\Lmain$. Moreover, it is simplified.
\end{restatable}

By \cref{lemma:buchi-simplified-implies-hd}, being a simplified B\"uchi automaton implies being history deterministic. 
Recall that a Büchi automaton is simplified if it is semantically deterministic, saturated, and has reach-covering.

We start with a property about the languages $L_\alpha$, that follows from inspection of $\Aclassifier$.
\begin{lemma}\label{lemma:infix-prefix}
    Let $u$ be finite word ending by a suffix in a language $L_\alpha$ for some $\alpha\in [6]$ (equivalently, ending by $cE$ or by $aA^*bF$). 
    Then, $u\in L_\alpha$ if and only if it contains no proper prefix in $L_{\beta}$ and no letter $\reset{\beta}$ for any~$\beta \in [6]$.
\end{lemma}

We also note the following fact, which we will use to show reach-covering and prove that the language of $\Amain$ is $\Lmain$.
\begin{lemma}\label{lemma:simulation-Amain}
The automaton $(\reach(\Amain),I)$ simulates $(\reach(\Amain),\sqrstate{p_i})$, for each $i\in[6]$.
\end{lemma}
\begin{proof}
We will describe a winning strategy for Eve in the corresponding simulation game. We denote the positions in this game $(p_i , I_\sigma)$;  until one of the players reaches a significant transition, Adam's token remains in states in the big blue box, and Eve's token in states in $B_{\mathtt{top}}$. The strategy will maintain the following invariant: 
    \vspace{-1mm}
    \begin{center}
        \begin{itemize}
            \item If Adam's token is in some state $\sqrstate{p_i}$, then Eve's token is in $I$.
            \item If Adam's token is in some state $\sqrstate{t_i}$, then Eve's token is in $I_c$.
        \end{itemize}
    \end{center}
In fact, every strategy for Eve satisfies this invariant.
    \begin{claim}
        Every possible strategy for Eve satisfies the above invariant.
    \end{claim}
    \begin{claimproof}
        First, note that Adam's token can only enter a state $\sqrstate{p_i}$ either upon reading a letter in $E+F+y$, or by taking a transition $\sqrstate{p_i}, \sqrstate{t_i}\xrightarrow{b} \sqrstate{p_i}$. 
        In the first case, for all these letters, the transitions of Eve's token are of the form $B_{\mathtt{top}} \to I$ (or a significant transition).
        Similarly, Adam's token can only enter a state $\sqrstate{t_i}$ by taking a transition $\sqrstate{p_i}, \sqrstate{t_i} \xrightarrow{c} \sqrstate{t_i}$.

        Using these remarks, we can prove the claim by induction, distinguishing several cases:
        \begin{itemize}
            \item Assume Adam's token 
        follows a transition with a letter in $E+F+y$, entering $\sqrstate{p_i}$; then, Eve's token either follows a transition $B_{\mathtt{top}} \to I$, or takes a significant transition.
         \item If Adam's token is in $\sqrstate{p_i}$ and takes a transition $\sqrstate{p_i} \xrightarrow{b} \sqrstate{p_i}$; then, by induction hypothesis, Eve's token is in $I$ and follows the transition $I \xrightarrow{b} I$.
        \item If Adam's token is in $\sqrstate{p_i}$ and takes a transition $\sqrstate{p_i} \xrightarrow{c} \sqrstate{t_i}$; then, by induction hypothesis, Eve's token is in $I$ and follows the transition $I \xrightarrow{c} I_c$.
        \item If Adam's token is in $\sqrstate{t_i}$ and takes a transition $\sqrstate{t_i} \xrightarrow{b} \sqrstate{p_i}$; then, by induction hypothesis, Eve's token is in $I_c$ and follows the transition $I_c \xrightarrow{b} I$.
        \item If Adam's token is in $\sqrstate{t_i}$ and takes a transition $\sqrstate{t_i} \xrightarrow{c} \sqrstate{t_i}$; then, by induction hypothesis, Eve's token is in $I_c$ and follows the transition $I_c \xrightarrow{c} I_c$.
        \end{itemize}
    \end{claimproof}
    The only positions where Eve has a choice to make is when her token is in state $I$ or $I_c$ and Adam gives letter $a$. 
    We define Eve's choice in positions of the form $(\sqrstate{p_i}, I)$ and $(\sqrstate{t_i}, I_c)$ over $a$.
    If $i\in E$, Eve stays in (or goes to) $I$, that is, she follows $I\xrightarrow{a} I$ or $I_c\xrightarrow{a} I$.
    If $i\in F$, Eve moves to the right, that is, she follows $I\xrightarrow{a} I_a$ or $I_c\xrightarrow{a} I_a$.
    Over all other positions of the game, we define Eve's choice arbitrary (note that in most of them she has no choice at all).

    We show that this is a winning strategy. Assume that in a play, Adam reaches a significant transition. We show that Eve does too.
    Consider the last moment during the play when Adam's token is in a state of the form $\sqrstate{p_i}$, and let $u_i$ be the suffix of the word of the play from this moment.
    We have that $u_i\in L_i$, and moreover, $u_i \in A^* i$ (as $\sqrstate{p_i}$ is not visited again).
    Assume that $i\in E$ (Eve decides to stay on the left part of $B_{\mathtt{top}}$). If $u_i$ does not contain letter $a$, then it is of the form $u_i = c^+i$ (this is immediate from the transitions in \cref{fig:transitionsoverabcy} and in the $i^\text{th}$ box in the bottom row of \cref{fig:transitionsoverEFr}).
    Therefore, Eve sees the significant transition $I_c \xRightarrow{i}$ at the same time as Adam.
    If $u_i$ contains letter $a$, then it must be of the form $u_i = u_0 a u_i' c i$, with $u_i'\in (b+c)^*$. Then, Adam necessarily follows the path:
    $$ \sqrstate{p_i} \xrightarrow{u_0}\cdot \xrightarrow{a} \sqrstate{q_i} \xrightarrow{u_i'} \cdot \xrightarrow{c} \sqrstate{s_i} \xRightarrow{i}.  $$
    On the left part of $B_{\mathtt{top}}$, Eve's token follows the path:
    $$ I \xrightarrow{u_0}\cdot \xrightarrow{a} I \xrightarrow{u_i'} \cdot \xrightarrow{c} I_c \xRightarrow{i},  $$ as wanted.

    Finally, assume that $i\in F$ (Eve decides to go to the right part of $B_{\mathtt{top}}$). In this case, $u_i$ must be of the form $(b+c)^* a A^* b i$.
    Then Eve's token follows the path
    $$ I \xrightarrow{(b+c)^*}\cdot \xrightarrow{a} I_a \xrightarrow{A^*} \cdot \xrightarrow{b} I_b \xRightarrow{i},  $$
    and also produces a significant transition.
\end{proof}

In \cref{prop:language-A65,prop:normal-Amain,prop:reach-covering-for-amain}, we prove that $\Amain$ is simplified and recognises $\Lmain$.

\begin{restatable}{proposition}{LanguagemainisSD}\label{prop:language-A65}
    Every state in $\Amain$ recognises the language $\Lmain$ and $\Amain$ is semantically~deterministic.
\end{restatable}
\begin{proof}
    We remark that the language $\Lmain$ is prefix-independent. Therefore, semantic determinism of $\Amain$ follows from the first part of the claim.
    
    We first show that $\Lmain \subseteq L(\Amain,q)$ for every state $q$. We will do this by inductively building an accepting run from $q$ on an arbitrary word in $\Lmain$. 
    At each step (except possibly the initialization step), we will make sure that a significant transition is taken.
    Note that all significant transitions of $\Amain$ end in some state in 
    $$ \Qsignificant = \{I\} \;  \cup \bigcup_\alpha \{ \circstate{p_{\alpha}} \}   \; \cup \{Y\}.$$
    As $\Lmain$ is prefix-independent, every suffix of  a word in $\Lmain$ also belongs to $\Lmain$.
    Therefore, it suffices to show that every word in $\Lmain$ admits a run reaching a significant transition from every state $q\in \Qsignificant$, that is, $\Lmain \subseteq L(\reach(\Amain),q)$, for all $q\in \Qsignificant$.
    \textit{Initialization step.} We show that a prefix of an arbitrary run reaches some state in $\Qsignificant$.
    Note that every word $w\in \Lmain$ contains infinitely many $\purpley$'s.
    Upon reading the letter $\purpley$, every run reaches some state in $$\Qsignificant \cup \bigcup_\alpha \{ \sqrstate{p_{\alpha}}\}.$$ 
    If the run reaches $\Qsignificant$ we are done. Assume that we reach a state $\sqrstate{p_{\alpha}}$. Let $w'$ be the suffix of $w$ from this point. Note that $w'$ must contain an infix in $L_\beta L_\beta$ for some $\beta \in [6]$: we consider the earliest  occurrence of such an infix. Then, from \cref{lemma:infix-prefix}, $w'$ either contains a reset letter or this prefix of $w'$ can be decomposed as $L_{\alpha_1}\dots L_{\alpha_k}L_{\beta} L_{\beta}$. If $w'$ contains a reset letter, the run reaches a state in $\Qsignificant$ upon reading that letter. Otherwise, if $\alpha_1 \neq \alpha$, then the run reaches $Y$ upon reading the prefix in $L_{\alpha_1}$. If $\alpha=\alpha_1$, then the run reaches $Y$ upon reading the first infix for which $\alpha_i = \alpha_{i+1}$ or after $L_{\beta} L_{\beta}$ if  $\alpha_i \neq \alpha_{i+1}$ for all $i$.

    \textit{Induction step.} To build the run in the induction step, we distinguish three cases, according to the type of state. 
    From $Y$, we see a significant transition upon reading $\purpley$.
    From $I$, we note the following claim.
    \begin{claim}
            $L(\reach(\Amain),I)$ contains the words in  $\bigcup_{\alpha} (L_\alpha + \Sigma^*\reset{\alpha})$.
    \end{claim}
    \begin{claimproof}
        Note that $L_{i} \subseteq L(\reach(\Amain),\sqrstate{p_i})$, for each $i\in[6]$. Since $(\reach{\Amain},I)$ simulates $(\reach(\Amain),\sqrstate{p_i})$ for each $i$ (see \cref{lemma:simulation-Amain}), we obtain that $L_i\subseteq L(\reach{\Amain},I)$. It is clear that each $\reset{i}$ is also in $L(\reach{\Amain},I)$, and thus, the claim follows.
    \end{claimproof}
    Since $w$ contains an infix in some $L_\alpha$, by \cref{lemma:infix-prefix}, the word $w$ contains a prefix in this language.
    The most challenging case is when $q = \circstate{p_{i}}$, for some $i$. 
    Since $w\in \Lmain$, it has a prefix of the form 
    \[ u_0 u_\alpha u_\beta u_\beta', \;\; \text{ with } \; u_\alpha \in (L_\alpha + \reset{\alpha}), \; u_\beta, u_\beta'\in L_\beta, \text{ and } \alpha\neq \beta. \]
    Consider an arbitrary run from $\circstate{p_{i}}$ over this word.
    We only need to show that some prefix of this run contains a significant transition. If there is no such significant transition seen, the run stays in the (deterministic) blue box.
    If $u_\alpha =  \reset{\alpha}$, the run over $u_0u_\alpha$ finishes in $\circstate{p_\alpha}$. Since $\alpha\neq \beta$, we have the following run  
    $$ \circstate{p_\alpha} \xrightarrow{u_\beta} \sqrstate{p_\beta} \xrightarrow{u_\beta'} \;\; \xRightarrow{}.$$
    We, therefore, focus on the case that $u_\alpha \in L_\alpha$. The last letter of $u_\alpha$ then must be $\alpha$.
    After reading $\alpha$, the automaton is necessarily in some state $\circstate{p_\gamma}$ or $\sqrstate{p_{\gamma}}$.
    \footnote{Recall that $\alpha\in \{1,2,3,4,5,6\}$, and also $E,F\subseteq [6]$, so $\alpha$ may appear in the transitions $E,F$ of the automaton.}
    From all these states, except from $\circstate{p_\beta}$, we reach a significant transition upon reading $u_\beta u_\beta'$.
    We argue that after reading $u_\alpha$ we cannot end up in $\circstate{p_\beta}$.
    Assume by contradiction that this is the case. 
    Then, it must be the case that the two predecessors of $\circstate{p_\beta}$ in the run are in the box $B_{\beta}$ on the top row, where 
    $$B_\beta = \{\circstate{p_\beta}, \circstate{q_\beta}, \circstate{t_\beta}, \circstate{s_\beta}, \circstate{r_\beta}\}.$$

    We distinguish the cases $\alpha\in E$ and $\alpha\in F$.
    If $\alpha\in E$, then $u_\alpha$ ends by $c\alpha$. However, we have 
    $$B_\beta \xrightarrow{c} \{ \circstate{t_\beta}, \circstate{s_\beta} \} \xrightarrow{\alpha} \circstate{p_\alpha} \neq \circstate{p_\beta} \;, \; \text{ a contradiction.}$$

    If $\alpha\in F$, then $u_\alpha$ ends by $aA^*b\alpha$. However, we have 
    $$B_\beta \xrightarrow{aA^*} \{ \circstate{q_\beta}, \circstate{s_\beta}, \circstate{r_\beta} \} \xrightarrow{b} \circstate{r_\beta} \xrightarrow{\alpha} \circstate{p_\alpha} \neq \circstate{p_\beta}, \ \ \text{a contradiction}.$$ This shows that $\Lmain\L(\Amain,q)$ for every state $q$.

    For the converse inclusion $L(\Amain,q) \subseteq \Lmain$ for every state $q$, let $w$ be a word in $L(\Amain,q)$.
   Then $w$ contains infinitely many $\purpley$'s, as every accepting run of $\Amain$ must go through $Y$ infinitely often.
    To show that $w$ contains infinitely many infixes in $(L_\alpha + \reset{\alpha})L_\beta L_\beta$, we observe the following facts. By ``consecutive occurrence'' we mean that no state of type $\circstate{p_\gamma}$ or $\sqrstate{p_\gamma}$ appears in between.
    \begin{itemize}
        \item If $u$ labels a path $I\xrightarrow{u} \; \xRightarrow{} \circstate{p_\alpha}$, then $u$ contains a suffix in $L_\alpha + \reset{\alpha}$.
        \item If $u$ labels a path between two consecutive occurrences of $\circstate{p_\alpha}$ and $\sqrstate{p_\beta}$, then $u\in L_\beta$ and~$\alpha\neq \beta$.
        \item If $u$ labels a path between two consecutive occurrences of $\sqrstate{p_\alpha}$ and $\sqrstate{p_\beta}$, then $u\in L_\beta$ and~$\alpha\neq \beta$.
        \item If $u$ labels a path between two consecutive occurrences of $\sqrstate{p_\alpha}$ and $\circstate{p_\beta}$, then $u \in \Sigma^*\reset{\beta}$.
        \item If $u$ labels a path between two consecutive occurrences of $\sqrstate{p_\beta}$ and $Y$, then $u\in L_\beta$.
    \end{itemize}
    Consider a significant transition ending in $Y$ on an accepting run over $w$ in $\Amain$. 
    Then,  $w$ contains an infix in $(L_\alpha + \reset{\alpha})L_\beta L_\beta$ for $\alpha\neq \beta$.
    This easily follows using the above properties, by considering the $3$ last appearances of $p$-states or the state $I$, and doing a distinction of six cases according to whether these states are $\sqrstate{p}$, $\circstate{p}$ or $I$.
\end{proof}

\begin{proposition}\label{prop:normal-Amain}
    $\Amain$ is normal.
\end{proposition}
\begin{proof}
    It suffices to observe that each rectangle in the image represents a strongly connected component induced by non-significant transitions. The rectangle containing the states $I,I_a,I_b$, and $I_c$ forms one such strongly connected component,  all the other states but $Y$
 forms  another such strongly connected component, and finally the state $Y$ alone also forms the final component. Every transition  between these three components are significant. 
\end{proof}

\begin{restatable}{proposition}{LanguagemainReachCover}\label{prop:reach-covering-for-amain}
    $\Amain$ has reach-covering.
\end{restatable}


\begin{proof}
Since all states in $\Amain$ are language-equivalent (by \cref{prop:language-A65}), $\Amain$ is SD. The only nondeterminism on $\Amain$ is over the transitions $I\xrightarrow{a}$ and $I_c\xrightarrow{a}$. Therefore, the only states in $\Amain$ that are not reach-deterministic are the $4$ states in the top pink box: $B_{\mathtt{top}} = \{I,I_a,I_b,I_c\}$. We have already proved that $I$ simulates $\sqrstate{p_i}$ in $\reach(\Amain)$ in \cref{lemma:simulation-Amain}, for each $i\in[6]$. We will show that $I_a,I_b,$ and $I_c$ simulate $\circstate{p_1}$ in $\reach{\Amain}$. Consider a strategy for Eve in the corresponding simulation game where she picks arbitrary transitions until Adam's token reaches $\sqrstate{p_i}$ for some $i$. This must eventually happen, for Adam's word to be in $L(\reach{\Amain},\circstate{p_1})$. Then, the letter in the corresponding round must be in $E$ or in $F$. Thus, Eve's token must take a significant transition in that round, as desired, or reach the state $I$. Since $I$ simulates $\sqrstate{p_i}$ in $\reach{\Amain}$, Eve can thus now use her simulation strategy from \cref{lemma:simulation-Amain}. 

\end{proof}

We conclude that $\Amain$ is simplified, finishing the proof of Theorem~\ref{thm:A65isHD}.

\subsection{Succinctness proof (full details)}\label{subsec:full-succinctness}
In this Section, we prove that $\Amain$ is succinct. 

\begin{theorem}\label{thm:succinctness}
    Every deterministic  B\"uchi automaton that is language-equi\-valent to~$\Amain$ has at least $66$ states.
\end{theorem}
To prove \cref{thm:succinctness} above, we show the equivalent result that every deterministic coBüchi~automaton recognising the~complement language $\Lmain^c$ has at least 66 states. Towards this, we rely on results from Abu Radi and Kupferman on the minimisation of HD coBüchi automata~\cite{AK22}, which we describe in \cref{subsec:hd-cobuchi-prelims}. 
We then present computer-aided lemmas about the size of minimal DFAs separating regular languages over finite words, an NP-complete problem~\cite{Gold78}. Using them, we conclude our proof in \cref{subsec:proof-of-succinctness}.

\subsubsection{Structure of HD coBüchi automata}\label{subsec:hd-cobuchi-prelims}
We briefly recall the results of Abu Radi and Kupferman's seminal paper \cite{AK22} that showed a canonical form for HD coB\"uchi automata. We will then use these results to construct a state-minimal history-deterministic coBüchi automaton for the complement language of $\Lmain$. 
We begin by introducing a few notions from their work that are necessary to state their results. 

We say that a coB\"uchi automaton is a \emph{safety automaton} if it has no \significant transitions. For every coB\"uchi automaton $\Aa$, we use $\safe(\Aa)$ to denote the safety automaton that is obtained by deleting all the \significant transitions of $\Aa$. 

A \emph{safe component} of $\Ac$ is a  strongly~connected~component 
of $\safe(\Ac)$. We write $\safeLangA{\Ac}{q}$ to denote the language $L(\safe(\Ac),q)$, for a state $q$ in $\Ac$, which we call the \emph{safe language} of $q$.

We say that a coBüchi automaton is \emph{normal} if every transition that is not in a safe component of $\Ac$ is a \significant transition. Every coBüchi automaton can be made normal while preserving the acceptance of each run, by simply making significant every non-\significant transitions of $\Ac$ that is not part of a safe component.
We assume that all states of automata are reachable from the initial state.

A coBüchi automaton $\Ac$ is \emph{safe deterministic} if the automaton $\safe(\Ac)$ is deterministic. 
A coBüchi automaton $\Ac$ is \emph{safe minimal} if there are no two different states that are language-equivalent in $\Ac$ as well as in $\safe(\Ac)$. 
It is \emph{safe centralised} if there are no two language-equivalent states $p,q$ in different safe components of $\Ac$ such that $\safeLangA{\Ac}{p} \subseteq \safeLangA{\Ac}{q}$.

Abu Radi and Kupferman proved the following result on minimality for HD coBüchi automata. 
\begin{lemma}[{\cite[Theorem 3.6]{AK22}}]\label{akimport-minimality}
Let $\Ac$ be a history-deterministic coBüchi automaton that is normal, semantically deterministic, safe deterministic, safe minimal, and safe centralised.
Then, it is statewise minimal amongst all HD coBüchi automaton recognising $L(\Ac)$.
\end{lemma}
We will also use the following result from Abu Radi and Kupferman's work.
\begin{restatable}[\protect{\cite[Proposition 3.4]{AK22}}]{lemma}{akimportantprop}
\label{AK22import-propthreepointfour}
Let $\Ac$ and $\Bc$ be two language-equivalent history\=/deterministic coBüchi automata that are normal, semantically deterministic, and safe deterministic. Then for every state $p$ in $\Ac$, there are states $q$ in $\Ac$ and $s$ in $\Bc$ such that $L(\Ac,p)=L(\Ac,q)=L(\Bc,s)$ and $$\safeLangA{\Ac}{p}\subseteq \safeLangA{\Ac}{q}=\safeLangA{\Bc}{s}.$$
\end{restatable}

\subsubsection{Complementation of HD B\"uchi automata}\label{subsec:complementation}
In their 2024 paper, Abu Radi, Kupferman, and Leshkowitz describe a complementation procedure for HD Büchi automata that produces  HD coBüchi automata with at most as many states~\cite[Theorem 2]{AKL24}. This procedure builds an HD coBüchi automata for the complement language that has as set of states the \good states of a simplified HD Büchi automata. 

If we use their procedure on our 65-state automaton $\Amain$, we obtain a 61-state history-deterministic coBüchi automaton $\Cmain$. The states of the automaton $\Cmain$ are the \good states of $\Amain$. The non-\significant transitions of $\Cmain$ are the non-\significant transitions outgoing from \good states of $\Amain$. The \significant transitions of $\Cmain$ are transitions $p\xRightarrow{\sigma}q$ for every  pair of states $p,q$ in $\Cmain$ and every letter $\sigma$. 
A rendering of $\Cmain$ is given in \cref{fig:coBuchi61automaton}.
\input{figures_sec2and3/coBuchi61state}
It follows from Abu Radi, Kupferman, and Leshkowitz's work~\cite[Proof of Theorem 9]{AKL24} that $L(\Cmain)=\Lmain^c$; we include a proof for self-containment nevertheless. 
\begin{restatable}{lemma}{cmainiscomplement}\label{lemma-sec5-cmain-iscomplementandhd}
    The automata $\Cmain$ is history deterministic, and recognises the complement of the language recognised by $\Amain$.
\end{restatable}
\begin{proof}
We will prove  that 
\begin{enumerate}
    \item  $L(\Cmain) \subseteq L(\Amain)^c$,
    \item $L(\Amain)^c \subseteq L(\Cmain)$,
    \item and that there is an HD-resolver for $\Cmain$.  
\end{enumerate}
  

\emph{Proof of 1.} Let $w$ be a word in $L(\Cmain)$, and $\rho$ be an accepting run of $\Cmain$ on $w$. Then, there is a suffix $u$ of $w$ and a state $p$ in $\Cmain$ such that the (deterministic) run from $p$ on $u$ stays in the safe-component containing $p$
and does not contain a significant~transition. Then $u$ is in $\safeLangA{\Cmain}{p}$, and thus, $u\notin L(\reach(\Amain),p)$. Since $L(B) \subseteq L(\reach(B))$ for every Büchi automaton $B$,  it follows that $u\notin L(\Amain,p)$. Since $L(\Amain,p)$ is prefix\=/independent and $\Amain$ is semantically deterministic, $w\notin L(\Amain)$, as desired. 

\emph{Proof of 2.} Let $w=a_0 a_1 a_2 \dots$ be a word that is not in $L(\Amain)$. We claim that there is a state $p$ of $\Amain$ and a suffix $w_i = a_i a_{i+1}\dots$ of $w$ such that every run of $(\Amain,p)$ on $w_i$ contains no \significant transition. If this is not the case, then it is easy to inductively build a run of $\Amain$ on $w$ that contains infinitely many \significant transitions, which contradicts the fact that $w \notin L(\Amain,p)$.

Thus, $w_i\notin L(\reach(\Amain),p)$ for some state $p$. We know, due to reach-covering for $\Amain$, that there is a \good state $q$ such that $p$ simulates $q$ in $\reach(\Amain)$.  Then, $$L(\reach(\Amain),p) \supseteq L(\reach(\Amain),q),$$ and thus $w_i \notin L(\reach(\Amain),q)$. This implies, by construction of $\Cmain$, that $w_i$ is in $\safeLangA{\Cmain}{q}$. Thus, $w\in L(\Cmain)$, since there is a run $\rho$ of $\Cmain$ on $w$ that takes arbitrary transitions on $a_0a_1\dots a_{i-1}$, takes the (possibly \significant) transition to $q$ on $a_i$, and then stays in the safe-component containing $q$ by taking deterministic non-\significant transitions. 

\emph{Proof of 3.} The fact that $\Cmain$ is HD follows from the fact that it is safe-deterministic and has \significant transitions going everywhere. We explicitly describe a resolver $\gamma$ for $\Cmain$. The resolver $\gamma$, after reading the prefix $u=a_0 a_1 \dots a_{i-1}$, will be in some state $q_i$. On the letter $a_{i}$, $\gamma$ takes the deterministic non-\significant transition $q_i \xrightarrow{a_i} q_{i+1}$, if it exists. Otherwise, $\gamma$ takes the \significant transition $q_i\xRightarrow{a_i}p$ to a state $p$ in $\Cmain$, where $p$ is the state with the longest suffix $v$ of $ua_i$ such that there is a run of $\safe(\Cmain)$ on $v$ that starts at some state $p'$ and ends at $p$. 

If the word $w$ on which $\gamma$ constructs a run is in $L(\Cmain)$, then there is a state $p$ and a suffix $w'$ of $w$ for which $w'\in L(\safe(\Cmain),p)$. The run of $\gamma$ on $w$ can only contain finitely many \significant transitions before it eventually coincides with the deterministic run of $(\Cmain,p)$ on $w'$ consisting of only non-\significant transitions. Thus, $\gamma$ is an HD-resolver for $\Cmain$, as desired.
\end{proof}

We now show statewise-minimality for $\Cmain$.
\begin{restatable}{lemma}{cmainminimal}\label{lemma:cmain-minimality}
    The automaton $\Cmain$ is the statewise minimal HD coBüchi automaton recognising $L(\Cmain)$.
\end{restatable}

\begin{proof}
    To apply \cref{akimport-minimality}, we need to prove that $\Cmain$ is HD, semantically deterministic, safe deterministic, normal, safe minimal, and safe centralised. We argued that $\Cmain$ is HD in \cref{lemma-sec5-cmain-iscomplementandhd}, $\Cmain$ is SD because every pair of states has a transition from one to the other on every letter (self-loop $\Sigma$), and $\Cmain$ is safe deterministic and normal by construction. We thus need to argue that $\Cmain$ is safe centralised and safe minimal.

    In the following, $\safeLang{p}$ denotes $\safeLangA{\Cmain}{p}$, for a state $p$ in $\Cmain$.

    \emph{Safe centralised.} We note that there are only two safe components in $\Cmain$: the safe component consisting of the state $\{Y\}$, and the safe component consisting of the other 60 states, which we call $\Sbus$.
    Note that $\safeLang{Y} =(\Sigma \setminus \{\purpley\})^{\omega}$, 
    while $\purpley^{\omega} \in \safeLang{s}$ for every state $s$ in $\Sbus$. Thus, 
    $$\safeLang{s}\nsubseteq \safeLang{Y}$$ 
    for every state $s\in \Sbus$.

    Similarly, we note that for each state $s$ in $\Sbus$, there is a word $w \in \left( \Sigma\setminus\{\purpley\} \right)^\omega$ such that $w\notin \safeLang{s}$, but $w$ is in $ \safeLang{\purpley}$ (as it does not contain $\purpley$). It follows that 
    $$\safeLang{Y} \nsubseteq \safeLang{s}$$ for every state $s\in \Sbus$. Thus, $\Cmain$ is safe centralised.

    \emph{Safe minimal.} We need to show that for every pair of distinct states $p',q'\in \Sbus$,  $\safeLang{p'} \neq \safeLang{q'}$. 
    We split the proof in the following 5 cases.\\
\textbf{Case 1.} $p'=\circstate{\pi_i}$ and $q'=\sqrstate{\pi'_j}$. Then, a finite word $v$ in $L_j$ satisfies
 \[yv y^{\omega} \in \safeLang{\circstate{\pi_i}} \setminus \safeLang{\sqrstate{\pi'_j}}.\] \\
\textbf{Case 2.}  $p'=\circstate{\pi_i}$ and $q'=\circstate{\pi'_j}$ with $i\neq j$. Then, a word $v_j \in L_j$ satisfies 
\[y v_jv_jy^{\omega} \in \safeLang{\circstate{\pi'_j}}\setminus \safeLang{\circstate{\pi_i}}.\]\\
\textbf{Case 3.} $p'=\sqrstate{\pi_i}$ and $q'=\sqrstate{\pi'_j}$ with $i\neq j$. Then, a word $v_j$ in $L_j$ satisfies \[y v_j y^{\omega} \in \safeLang{\sqrstate{\pi_i}}\setminus\safeLang{\sqrstate{\pi'_j}}.\]\\
\textbf{Case 4.} $p'=\circstate{\pi_i},q'=\circstate{\pi'_i}$, for some $\pi \neq \pi'\in \{p,q,r,s,t\}$. Let  $e\in E = \{1,2,3\}$, $f\in F = \{4,5,6\}$, be such that $e,f\neq i$.
We have the following 5 subcases here. 

    \textbf{Case 4a.} $\pi\in \{p,t\}, \pi' \in \{q,r,s\}$. Then, the word $bf$ is such that $\circstate{\pi_i} \xrightarrow{bf} \circstate{p_i}$, and $\circstate{\pi'_i}\xrightarrow{bf} \sqrstate{p_f}$. Since $\safeLang{\circstate{p_i}}\neq \safeLang{\sqrstate{p_f}}$ from Case 1, we obtain that $\safeLang{\circstate{\pi_i}}\neq \safeLang{\circstate{\pi'_i}}$.
    
    \textbf{Case 4b.} $\pi = p,\pi'= t.$ Then, $\circstate{\pi_i} \xrightarrow{e} \circstate{p_i}$, whereas $\circstate{\pi'_i} \xrightarrow{e} \sqrstate{p_e}$. 
    We conclude by Case 1.
    
    \textbf{Case 4c.} $\pi = q, \pi' = r$. Then $\circstate{\pi_i}\xrightarrow{f} \circstate{p_i}$, whereas $\circstate{\pi'_i}\xrightarrow{f}\sqrstate{p_f}$. 
    We conclude by Case 1.
    
    \textbf{Case 4d.} $\pi = q, \pi'= s$. Then, $\circstate{\pi_i}\xrightarrow{e} \circstate{p_i}$, whereas $\circstate{\pi'_i}\xrightarrow{e}\sqrstate{p_e}$. 
    We conclude by Case 1.
    
    \textbf{Case 4e.} $\pi = r, \pi'=s$. Then, $\circstate{\pi_i} \xrightarrow{f} \sqrstate{p_f}$, and $\circstate{\pi'_i}\xrightarrow{f}\circstate{p_i}$. 
    We conclude by Case 1.
\vskip1em

\noindent\textbf{Case 5.} $p'=\sqrstate{\pi_i},q'=\sqrstate{\pi'_i}$, with $\pi \neq \pi'\in \{p,q,r,s,t\}$. Let $e\in E, f\in F$ be such that $e,f\neq i$. We have the following subcases here. 

    \textbf{Case 5a.} $\pi\in \{p,t\}, \pi' \in \{q,r,s\}$. Then, the word $bf$ is such that $\sqrstate{\pi_i} \xrightarrow{bf} \sqrstate{p_i}$, and $\sqrstate{\pi'_i}\xrightarrow{bf} \sqrstate{p_f}$. Since $\safeLang{\sqrstate{p_i}}\neq \safeLang{\sqrstate{p_f}}$ from Case 3, we obtain that $\safeLang{\sqrstate{\pi_i}}\neq \safeLang{\sqrstate{\pi'_i}}$.
    
    \textbf{Case 5b.} $\pi = p,\pi'= t.$ Then, $\sqrstate{\pi_i} \xrightarrow{e} \sqrstate{p_i}$, whereas $\sqrstate{\pi'_i} \xrightarrow{e} \sqrstate{p_e}$. 
    We conclude by Case 3.

    \textbf{Case 5c.} $\pi = q, \pi' = r$. Then $\sqrstate{\pi_i}\xrightarrow{f} \sqrstate{p_i}$, whereas $\sqrstate{\pi'_i}\xrightarrow{f}\sqrstate{p_f}$. 
    We conclude by Case 3.
    
    \textbf{Case 5d.} $\pi = q, \pi'= s$. Then, $\sqrstate{\pi_i}\xrightarrow{e} \sqrstate{p_i}$, whereas $\sqrstate{\pi'_i}\xrightarrow{e}\sqrstate{p_e}$. 
    We conclude by Case 3.
    
    \textbf{Case 5e.} $\pi = r, \pi'=s$. Then, $\sqrstate{\pi_i} \xrightarrow{f} \sqrstate{p_f}$, and $\sqrstate{\pi'_i}\xrightarrow{f}\sqrstate{p_i}$. 
    We conclude by Case 3.\\
    Thus, $\Cmain$ is safe minimal, safe centralised, semantically deterministic, HD, and normal, and thus, due to \cref{akimport-minimality}, statewise minimal.
\end{proof}

We note that \cref{lemma:cmain-minimality} also proves that every deterministic coBüchi automaton $\Dd$ recognising $L(\Cmain)$ requires at least $61$ states. We thus need to prove that  $\Dd$ additionally has at least $5$ more states, which we do in the next subsection. 

\subsubsection{Study of deterministic coB\"uchi automata for $\Lmain^c$}
\label{subsec:proof-of-succinctness}
We now prove the following result.
\begin{lemma}\label{lemma:66-is-needed}
Every deterministic coBüchi automaton that is language\=/equivalent to $\Cmain$ has at least 66 states.
\end{lemma}
Throughout this subsection, we fix $\Dmain$ to be a deterministic coBüchi automaton that recognises $L(\Cmain)$. 
The proof is by contradiction: we assume that $\Dmain$ has at most $65$ states, and we will conclude that it cannot recognise $L(\Cmain)$.

We assume, without loss of generality, that $\Dmain$ is normal, that is,  all its non-\significant transitions occur in some SCC in $\safe(\Dmain)$. Note that because $L(\Cmain)=L(\Dmain)$ is prefix-independent, every state of $\Dmain$ recognises the language $L(\Dmain)$.

We denote the safe component of $\Cmain$ that consists of every state of $\Cmain$ apart from~$Y$ by~$\Sbus$. We first prove existence of states in $\Dmain$ whose safe languages coincide with the states of $\Sbus$ in $\Cmain$ (\cref{lemma:bus-state-one}). For this, we use \cref{AK22import-propthreepointfour}. 

\begin{lemma}\label{lemma:d_0}
    There is a state $d_0$ in $\Dmain$ such that $$\safeLangA{\Dmain}{d_0}=\safeLangA{\Cmain}{\circstate{p_1}}.$$
\end{lemma}
\begin{proof}
    In \cref{AK22import-propthreepointfour}, we let $\Cmain$ be $\Ac$, $\Dmain$ be $\Bc$ and let $p$ be $\circstate{p_1}$. Then, there is a state $q$ in $\Cmain$ and a state $r$ in $\Dmain$ such that $$\safeLangA{\Cmain}{\circstate{p_1}}\subseteq \safeLangA{\Cmain}{q}=\safeLangA{\Dmain}{r}.$$

    Because $\Cmain$ is safe centralised (see proof of \cref{lemma:cmain-minimality}),  $\circstate{p_1}$ and $q$ are in the same safe component of $\Cmain$. Thus, let $d_0$ be such that we have the non-\significant transition $r\xrightarrow{\reset{1}}d_0$ in $\Dmain$, which  exists because $\safeLangA{\Cmain}{q}=\safeLangA{\Dmain}{r}$ and $q\xrightarrow{\reset{1}}\circstate{p_1}$ is a non\=/\significant transition in~$\Cmain$. It follows that the safe language of $d_0$ in $\Dmain$ is $\safeLangA{\Cmain}{\circstate{p_1}}$, as desired.   
\end{proof}
We let $\Dbus$ be the safe component of $\Dmain$ that contains $d_0$. The following lemma follows from \cref{lemma:d_0} above, and the strong-connectivity of $\Sbus$ and $\Dbus$.
\begin{lemma}\label{lemma:bus-state-one}
    For every state $p$ in $\Sbus$, there is a state $q$ in $\Dbus$ such that $$\safeLangA{\Dmain}{q}=\safeLangA{\Cmain}{p}.$$
\end{lemma}
Concretely, let $u$ be a word on which there is a run from $\circstate{p_1}$ to $p$ in $\Sbus$. Then, since $\safeLangA{\Cmain}{\circstate{p_1}}=\safeLangA{\Dmain}{d_0}$, the state $q$ to which there is a unique-safe run from $d_0$ to $q$ satisfies  $\safeLangA{\Dmain}{q}=\safeLangA{\Cmain}{p}.$ 

Since $\Cmain$ is safe-minimal, this implies that $\Dbus$ has at least 60 states. Similar to \cref{lemma:d_0}, we can also prove that there is a state $d_y$ in $\Dmain$ whose safe language is equivalent to that of $Y$ in $\Amain$.
\begin{lemma}\label{lemma:d_y}
There is a state $d_y$ in $\Dmain$ such that $$\safeLangA{\Dmain}{d_y}=\safeLangA{\Cmain}{Y}.$$
\end{lemma}
\begin{proof}
    By \cref{AK22import-propthreepointfour}, there is a state $d_y$ in $\Dmain$ and a state $q$ in $\Cmain$ such that $$\safeLangA{\Cmain}{Y}\subseteq \safeLangA{\Cmain}{q}=\safeLangA{\Dmain}{d_y}.$$
    We observe that $q$ must be $Y$, since $\Cmain$ is safe centralised and the safe component consisting of $Y$ is a singleton. Thus, the above language inclusion is an equality, as desired. 
\end{proof}

Thus, $\Dmain$ has at least 61 states: at least 60 states in $\Dbus$, one state $d_y$, and possibly more. Towards proving that $\Dmain$ has at least $66$ states, we need to introduce a few notions.
\subparagraph{The subautomaton $\Dcore$.}
First,  we consider the SCC decomposition of the safe component~$\Dbus$ in the absence of transitions on~$y$. We call this the $y$-SCC decomposition. 
For~a state $p$ in $\Dbus$, we let $\ysafelangD{p}$ be the language of words $w$ such that $w$ does not contain~$y$, and the run on $w$ from $p$ stays in the same $y$-SCC of $\Dbus$.

Let $\Dcore$ be a subautomaton that is an end-SCC in the $y$-SCC-decomposition of $\Dbus$. We prove that $\Dcore$ has at least 60 states.
\begin{lemma}\label{lemma:Dcore-corresponding-safe-state}
    For every state $q$ in $\Dcore$, there is a state $p$ in $\Sbus$ such that $$\safeLangA{\Dmain}{q}=\safeLangA{\Cmain}{p}.$$
\end{lemma}
\begin{proof}
    Let $d_0$ be the state in $\Dbus$ given by \cref{lemma:d_0}, and $d_0 \xrightarrow{u} q$ a safe run in $\Dbus$ reaching~$q$. 
    Then, $\circstate{p_1} \xrightarrow{u} p$ is the desired state.
\end{proof}

\begin{lemma}\label{lemma:Dcore-one-state-for-everythingin60}
    For every state $p$ in $\Sbus$, the subautomaton $\Dcore$ contains a state $q$ such that $$\safeLangA{\Dmain}{q}=\safeLangA{\Cmain}{p}.$$  
\end{lemma}
\begin{proof}
    This follows from the observation that, in the automaton $\Cmain$, $\Sbus$ is strongly-connected even in the absence of transitions on $y$. 

    Concretely, let $q'$ be a state in $\Dcore$, and let $p'$ in $\Sbus$ be a state such that $$\safeLangA{\Dmain}{q'}=\safeLangA{\Cmain}{p'};$$  such a $p'$ exists due to the previous lemma.
    
    Let $p' \xrightarrow{u'} p$ be a safe run in $\Sbus$ such that $u'$ does not contain letter $y$. Then the corresponding run from $q'$ on $u'$ in $\Dmain$ leads to the desired state $q$: note that $q$ is in $\Dcore$ because $\Dcore$ is an end-$y$-SCC.
\end{proof}
Thus, we note that if the $y$-SCC decomposition of $\Dbus$ contains two end components, then the automaton $\Dmain$ has size at least 120. We thus suppose that the $y$-SCC decomposition of $\Dbus$ has exactly one end component, which we call $\Dcore$. 

\begin{lemma}\label{lemma:Dbus-dcore}
    Suppose that $\Dmain$ has at most $65$ states.
Then, the $y$-SCC decomposition of~$\Dbus$ has exactly one end component, which we denote by $\Dcore$.
\end{lemma}

\subsubsection{Computer-aided lemmas}

\subparagraph{$\Theju$-safety.}
Let $\Theju$ be the language of infinite words  over $\Sigma\setminus \{r_1,\dots,r_6\}$ for which no prefix is in $L_i$ for $i\in[6]$. This language $\Theju$ is recognised by the safety automaton shown in \cref{fig:languageTheju}. We let $\Theju_y$ be the subset of words in $\Theju$ that contain no $y$.

\begin{remark}\label{rmk:Theju-in-Cmain}
    It holds that
    \[\Theju_y\subseteq \Theju \subseteq \Lmain^c = L(\Cmain) = L(\Dmain).\]
\end{remark}

\begin{figure}[h]
    \centering
\begin{tikzpicture}[
        base dotred/.style={
        inner sep=0pt,      
        minimum size=5pt,   
        fill=red!50!white, 
        font=\scriptsize    
    },
    base dot/.style={
        inner sep=0pt,      
        minimum size=5pt,   
        fill=blue!60!black, 
        font=\scriptsize    
    },
    cstate/.style={base dotred, circle},  
    node distance=1.5cm,
    tight label/.style={auto, inner sep=0.5pt, text=black},
    sstate/.style={base dot, rectangle}, 
    house peak H/.store in=\peakH,
    house peak H=-1.6cm,
    house eaves H/.store in=\eavesH, 
    house eaves H=0.15cm,
    house base H/.store in=\baseH,   
    house base H=-0.7cm,
    house width/.store in=\houseW,   
    house width=0.5cm,
    group sep x/.store in=\gSepX, 
    group sep x=2.6cm, 
    group sep y/.store in=\gSepY, 
    group sep y=3.2cm, 
        row top y/.store in=\rowTopY,
    row bot y/.store in=\rowBotY,
    top x off/.store in=\topX,
    bot x off/.store in=\botX,
    row top y=0.6cm,
    row bot y=-0.9cm,
    top x off=0.6cm,
    bot x off=1.2cm
]
    \filldraw [fill=violet!10, draw=violet!30] (-2.5,3.2) rectangle (2.8,-0.5);

    
\foreach \i in {1} {
    \coordinate (center\i) at ({(\i-1)*\gSepX}, 0);

    \node[cstate] (p\i) at (-1,2) {}; 
    \node[cstate] (q\i) at (1,2)  {};

    \node[cstate] (t\i) at (-1.5, 0) {};
    \node[cstate] (s\i) at (0,0)      {};
    \node[cstate] (r\i) at (2,0)  {};
}




 \foreach \i in {1} {
  \path[->,>=stealth, thick, shorten >=1pt]
    (p\i) edge [loop left] node [above,tight label] {$b,{\color{purple}E,F}$} (p\i)
    (q\i) edge [loop right] node [tight label] {$a$} (q\i)
    
    (r\i) edge [loop right] node [tight label] {$b$} (r\i)
    (t\i) edge [loop left] node [tight label, yshift=0cm, xshift=-0.05cm] {$c$} (t\i)

    (p\i) edge  node [tight label,below] {$a$} (q\i)
    (q\i) edge [purple,bend right] node [tight label,above] {${\color{purple}{E,F}}$} (p\i)

    (q\i) edge [bend left = 10] node [tight label] {$b$} (r\i)   
    (r\i) edge [bend left = 10] node [tight label,yshift=0.1cm, xshift=-.05cm] {$a$} (q\i)
     
    (r\i) edge [bend left = 10] node [tight label] {$c$} (s\i)  
    (s\i) edge [bend left = 10] node [tight label] {$b$} (r\i)       
    (p\i) edge [bend right =10] node [ tight label, left] {$c$} (t\i)
    (t\i) edge [bend right=10] node [tight label, right] {$b$} (p\i)

    (q\i) edge [bend left = 10] node [tight label, yshift=0cm] {$c$} (s\i)  
    (s\i) edge [bend left = 10] node [tight label,yshift=-0.15cm] {$a$} (q\i)   
    (t\i) edge  node [tight label, left, yshift=-0.20cm, xshift=-.35cm] {$a$} (q\i)
    
    (t\i) edge [bend left=40, purple] node [tight label, left] {${\color{purple}F}$} (p\i)

    (r\i) edge [purple,out=60,in=60, looseness=1.8] node [tight label, above,xshift=0.1cm] {{\color{purple}$E$}} (p\i)  
    (s\i) edge [purple] node [tight label,xshift=0.1cm,yshift=0.6cm] {$\color{purple}F$} (p\i)   




    (s\i) edge [loop left] node [tight label,left,xshift=-0.1cm,yshift=0cm] {$c$} (s\i);
    ;


    
     \path[-stealth] 
    (-1.2,3.2) edge [thick, out=110,in=130,looseness=3] node [xshift=0.0cm,left] {$y$} (p\i);

    
    

}
\end{tikzpicture}
    \caption{The automaton for the language $\Theju$. The initial state is the top left one.
    The automaton for language $\Theju_y$ is obtained by removing the transitions on~$y$.}
    \label{fig:languageTheju}
\end{figure}


\subparagraph*{Computer-aided lemma.}
We now prove the following result by the aid of computers, namely, the tool DFAMiner~\cite{DLS24}. 

\begin{restatable}{lemma}{compaid}\label{lemma:compy}
        Let $\Dd$ be a deterministic safety automaton that separates $\Theju$ (resp.\ $\Theju_y$) and $\safeLangA{\Cmain}{q}$ for some state $q$ in $\Sbus$, i.e.,
    $$\Theju\subseteq L(\Dd) \subseteq \safeLangA{\Cmain}{q}.$$
    Then, $\Dd$ has at least $5$ states.
\end{restatable}

Finding minimal DFAs (or deterministic safety automata) separating regular languages  is an NP-complete problem~\cite{Gold78}. We thus use DFAMiner, a SAT-solver based tool for separating regular languages~\cite{DLS24}. Nonetheless, instead of plugging a 60-state~($\Sbus$) and a 5-state automaton (the automaton for $\Theju$) into this tool, we make certain reductions that allow us to considerably reduce the state-space of automata that we input. We describe precise details.

We note that for two safety automata $\Ac$ and $\Bc$, $L(\Ac)\subseteq L(\Bc)$ if and only if $\Fin{L(\Ac)}\subseteq \Fin{L(\Bc)}$, where $\Fin{L}$ is the set of prefixes of words in $L$. Thus, in this section, we view our safety automata as DFAs, where every state is accepting, together with an additional rejecting sink state for completion.

Towards proving \cref{lemma:compy}, we first restrict candidates for our state $q$ in $\Sbus$. Let~$d_0$ be the initial state of $\Dd$. Let $q\xrightarrow{1}q'$, and  $d_0\xrightarrow{1}d'$. We note that $1 \equiv_L \varepsilon$ in $\Theju$ and $\Theju_{y}$, i.e., $1$ and $\varepsilon$ are in the same Myhill-Nerode equivalence classes. Thus, $(\Dd,d')$ separates $\Theju$ and $\safeLangA{\Cmain}{q'}$. But note that $q'$ is either  $\circstate{p_i}$ or $\sqrstate{p_i}$, for some $i\in[6]$. Furthermore, in $\Amain$, the role of $\{1,2,3\}$ is symmetric, as is of $\{4,5,6\}$. Thus, it suffices to prove \cref{lemma:compy} for when $q'$ is a state in $\{\circstate{p_2},\circstate{p_5},\sqrstate{p_1},\sqrstate{p_4}\}$. 

\begin{remark}
A reader might ask: why $\circstate{p_2}$ or $\circstate{p_5}$ instead of $\circstate{p_1}$ and $\circstate{p_4}$ above? This is done for the following technical reason: when we restrict the alphabet set to $\{a,b,c,1,4\}$, \cref{lemma:simpler-comp} below is not true if $q$ is $\circstate{p_1}$ or $\circstate{p_4}$.
\end{remark}

For a DFA $\Aa$ over the alphabet $\Sigma$ and a subset $\Gamma\subseteq \Sigma$, we use $\Aa^\Gamma$ to denote the DFA obtained by restricting the transitions of $\Aa$ to the alphabet $\Gamma$.
We make another simplification, based on the following simple observation. 

\begin{proposition}\label{prop:restriction-letters}
Let $\Aa,\Bb,$ and $\Cc$ be three DFA over $\Sigma$ such that $$L(\Aa)\subseteq L(\Bc)\subseteq L(\Cc).$$ Then, for every $\Gamma\subseteq \Sigma$, 
$$L(\Aa^{\Gamma})\subseteq L(\Bc^{\Gamma})\subseteq L(\Cc^{\Gamma}).$$
\end{proposition}
For a language $L$ over infinite words, we use $L_{\Gamma}$ to denote the set of words in $L\cap\Gamma^{\omega}$.
Due to \cref{prop:restriction-letters}, it suffices to show the following result.
\begin{lemma}\label{lemma:simpler-comp}
Let $\Gamma=\{a,b,c,1,4\}$, and let $\Dd$ be a DFA that separates $\Fin{\Theju_{\Gamma}}$ and $\Fin{\safeLangA{\Cmain^{\Gamma}}{q}}$, for some $q\in\{\circstate{p_2},\circstate{p_5},\sqrstate{p_1},\sqrstate{p_4}\}.$ Then, $\Dd$ has at least $6$ states (accounting for a rejecting sink state).
\end{lemma}
Note that in \cref{lemma:simpler-comp}, the size of $(\Cmain^{\Gamma},q)$, when trimmed to reachable states, is 15 for each $q$, which is much smaller than 60. This allows us to use the tool DFAMiner.

\subparagraph{Using DFAMiner.}
One of the features of DFAMiner is that it minimises 3-DFAs. A 3-DFA has the syntax of a DFA, but each state is labelled as either accept, reject, or `don't care'. The idea is that a 3-DFA $\Aa$ accepts a word $w$ if the run of $\Aa$ on $w$ ends at an accepting state and rejects $w$ if the run of $\Aa$ on $w$ ends at rejecting sink state.

The problem of 3-DFA minimisation is the following: Given a 3-DFA $\Ac$, construct the smallest DFA $\Bc$ such that $\Bc$ accepts every word accepted by $\Ac$ and $\Bc$ rejects every word rejected by $\Ac$.

It is easy to see that the problem of separating regular languages can be reduced to the problem of minimising a 3-DFA that is obtained by a product construction. 
For proving our \cref{lemma:simpler-comp}, this product construction results in 16 state DFAs, due to the similarity of structure between the automaton recognising $\Kk$ and $\Cmain$. 
The four $3$-DFAs that we need to consider are depicted in \cref{fig:zblock1,fig:zblock2,fig:zblock3,fig:zblock4}.   

The tool DFAMiner, when asked to minimise these four 3-DFAs, outputs that the minimal 3-DFA has size 6, which coincides with the deterministic safety automaton recognising $\Theju_y$ together with a rejecting sink state. We attach an encoding of these automata in the desired format of the DFAMiner tool for the readers who wish to verify.

\subparagraph{Figures of the four 3-DFAs.}
We depict the four $3$-DFAs given to DFAMiner in \cref{fig:p1p2,fig:p3p4}.
The encoding of these 3-DFAs for DFAMiner is provided in the Appendix~\ref{appendix:code}.

\input{figures_sec2and3/zblocks12}
\input{figures_sec2and3/zblocks34}


\subsubsection{Concluding the proof}
We note two useful corollaries of \cref{lemma:compy} below.

\begin{corollary}\label{cor:5states-dbus}
    Let $q$ be a state of $\Dmain$ outside $\Dbus$. If $$\Theju\subseteq \safeLangA{\Dmain}{q},$$ then the safe component containing $q$ has at least $5$ states.
\end{corollary}
\begin{proof}
By \cref{AK22import-propthreepointfour},  there is a state $p$ in $\Cmain$ such that $\safeLangA{\Dmain}{q}\subseteq \safeLangA{\Cmain}{p}$. This state $p$ cannot be $Y$, since $\Theju\nsubseteq \safeLangA{\Cmain}{Y}$. Thus, $p$ is a state in~$\Sbus$. The conclusion then follows from \cref{lemma:compy}.
\end{proof}

\begin{corollary}\label{cor:5states-dcore}
    Let $q$ be a state of $\Dbus$ outside $\Dcore$. If $$\Theju_y \subseteq \ysafelangD{q},$$ then the  $y$-SCC containing $q$ has at least $5$ states. 
\end{corollary}
\begin{proof}
    Observe that $\ysafelangD{q} \subseteq \safeLangA{\Dmain}{q}$. Since $q$ is in $\Dbus$, there is a state $p\in\Sbus$ such that $\safeLangA{\Dmain}{q}=\safeLangA{\Cmain}{p}.$ We thus obtain that $$\Theju_y \subseteq \ysafelangD{q} \subseteq \safeLangA{\Cmain}{p},$$ and the conclusion follows from \cref{lemma:compy}.      
\end{proof}

Observe that if there is a state $q$ outside $\Dbus$ with $$\Theju \subseteq \safeLangA{\Dmain}{q},$$ then the safe component of $q$ is disjoint from the safe component consisting of $d_y$, since $\Theju \nsubseteq \safeLangA{\Dmain}{d_y}$ and $\safeLangA{\Dmain}{d_y}$ is prefix-independent. Thus, in this case, $\Dmain$ will have at least 66 states: the 60 states of $\Dcore$, the state $d_Y$, and these additional 5 states that \cref{cor:5states-dbus} implies the existence of.

Similarly, if there is a state $q$ outside $\Dcore$ with $\Theju_y \subseteq \ysafelangD{q}$ then \cref{cor:5states-dcore} implies that $\Dmain$ has at least 66 states: the 60 states of $\Dcore$, the state $d_Y$, and these additional 5 states that \cref{cor:5states-dcore} implies the existence of.


We thus obtain the following. 

\begin{lemma}\label{lemma-states-outside-dcore}
    If $\Dmain$ has at most $65$ states, then 
    there are no states $q$ outside $\Dbus$ with $\Theju\subseteq\safeLangA{\Dmain}{q}$ and there are no states $p$ in $\Dbus$ outside $\Dcore$ with $\Theju_y \subseteq \ysafelangD{p}$.
\end{lemma}

We conclude our proof in the next lemma.

\begin{restatable}{lemma}{CoverLemmaformostComplicated}\label{lemma:imp2} 
    If $\Dmain$ has at most $65$ states, then it  does not recognise $L(\Cmain)$.
\end{restatable}

\begin{proof}
    For a language $L$ of infinite words, we let $\Fin{L}$ be the language of finite words that are prefixes of words in $L$. 
    First, note the following:

    \begin{claim}\label{claim:reset-K}
        For every word $u\in \Fin{\Theju}$, there is a word $u'\in \{a,b,c,1,\dots, 6\}^*$ such that $(uu')^{-1}\Theju = \Theju$. That is, $u'$ resets the automaton in Figure~\ref{fig:languageTheju} to the initial state.
        Moreover, a word belongs to $\Theju$ if and only if all its prefixes belong to $\Fin{\Theju}$.
    \end{claim}

    We leverage \cref{lemma-states-outside-dcore} to force a run to enter $\Dcore$ while reading prefixes of $\Theju$.
    \begin{claim}\label{claim:from-D-to-Dbus}
        For every state $q$ outside $\Dbus$, there is a word $u\in \Fin{\Theju}$ such that the run $q\xrightarrow{u}$ ends in $\Dbus$  and visits a \significant transition.
    \end{claim}
    \begin{claimproof}
        By \cref{lemma-states-outside-dcore}, if $q$ is outside $\Dbus$, there is a word $u_1\in \Fin{\Theju}\setminus \Fin{\safeLangA{\Dmain}{q}}$. By the previous claim, we can moreover assume that $(u_1)^{-1}\Theju = \Theju$.  
        Then, the run $q\xrightarrow{u_1} q_1$ crosses some \significant transition. 
        While $q_i$ is not in $\Dbus$, we can repeat this process. If we never enter $\Dbus$, this produces a word in $\Theju$ whose run over $\Dmain$ is rejecting, a contradiction (see \Cref{rmk:Theju-in-Cmain}).
    \end{claimproof}

     \begin{claim}\label{claim:from-Dbus-to-Dcore}
        For every state $q$ in $\Dbus$ outside $\Dcore$, there is a word $u\in \Fin{\Theju}$ such that the run               $q\xrightarrow{u}$ ends in $\Dcore$.
    \end{claim}
    \begin{claimproof}
        Let $q_0 = q$. By induction, we produce a run
        \[ q_0 \xrightarrow{u_0} q_1 \xrightarrow{u_1} q_2 \xrightarrow{u_2} \dots ,\]
        such that, while $q_i \notin \Dcore$, $u_i\in \Fin{\Theju}$ with $u_i^{-1}\Theju = \Theju$ and either:
        \begin{enumerate}
            \item the state $q_i$ is not in $\Dbus$ and $q_i \xrightarrow{u_i} q_{i+1}$ visits a \significant transition, or
            \item the state $q_i$ is in $\Dbus\setminus \Dcore$ and $q_i \xrightarrow{u_i} q_{i+1}$ visits a significant transition, or
            \item the state $q_i$ is in $\Dbus\setminus \Dcore$, $u_i$ does not contain letter $y$ and $q_i \xrightarrow{u_i} q_{i+1}$ is a safe run that changes of $y$-SCC in $\Dbus$.
        \end{enumerate}
        In the first case ($q_i$ not in $\Dbus$) the word $u_i$ is given by the previous claim.
        If $q_i\in \Dbus$,  by the second part of \cref{lemma-states-outside-dcore}, there is $u_i\in \Fin{\yTheju} \setminus \Fin{\ysafelangD{q_i}}$. By \cref{claim:reset-K} we can assume $u_i^{-1}\Theju = \Theju$.
        Therefore, by definition of $\ysafelangD{q_i}$, either $q_i \xrightarrow{u_i} q_{i+1}$ produces a significant transition (case 2) or $q_i \xrightarrow{u_i} q_{i+1}$ is a safe run that changes of $y$-SCC in $\Dbus$.        

        Assume by contradiction that no $q_i$ is in $\Dcore$. It cannot be the case that \significant transitions are visited infinitely many times, as this would produce a rejecting run over a word in $\Theju$.
        Therefore, eventually all $q_i$ are as in the third case.
        However, since $\Dcore$ is the only end-$y$-SCC of $\Dbus$ (\cref{lemma:Dbus-dcore}), it is not possible to change of $y$-SCC infinitely often avoiding both \significant transitions and $y$-transitions, a contradiction.
    \end{claimproof}

    Let $\{o_1,\dots, o_5\}$ be the (at most) $5$ states outside $\Dcore$.
    Combining the two previous claims, 
    we get that there are words $u_1,\dots,u_5\in \Fin{\Theju}$ such that 
    $o_i \xrightarrow{u_i} q_i$ ends in $\Dcore$.
    Let $p_i$ in $\Sbus$ such that $\safeLangA{\Dmain}{q_i} = \safeLangA{\Cmain}{p_i}$ (such a state exists by \cref{lemma:Dcore-corresponding-safe-state}).
    We may assume (by \cref{claim:reset-K}) that moreover, each $p_i$ is of the form $\circstate{p_j}$ or $\sqrstate{p_j}$ for some $j$. 
    Let $\Qtarget = \{q_1,\dots,q_5\}$ the set  of possible destinations of the above paths and  $\Ptarget = \{p_1,\dots, p_5 \} \subseteq \{ \circstate{p_j}, \sqrstate{p_j} \mid j\in [6]\}$ the set (of size at most $5$) of corresponding states in $\Cmain$ recognising the same safe languages as $\Qtarget$.
    Let $j_0$ be an index such that $\circstate{p_{j_0}}, \sqrstate{p_{j_0}}\notin \Ptarget$.

    \begin{claim}\label{claim:from-Dcore-to-D}
        For every state $q$ in $\Qtarget$ there is a word $v\in L_{j_0}^2$ such that the run $q\xrightarrow{v}$ ends outside~$\Dcore$.
    \end{claim}
    \begin{claimproof}
        Let $p$ in $\Sbus$ such that $\safeLangA{\Dmain}{q} = \safeLangA{\Cmain}{p}$ and $p \in \{\circstate{p_i}, \sqrstate{p_i}\}$, with $i\neq j_0$.
        Take a word $v \in L_{j_0}^2$ that does not contain $y$.
        Then, the run $p\xrightarrow{v}$  in $\Cmain$ produces a \significant transition (indeed, $p \xrightarrow{L_{j_0}} \sqrstate{p_{j_0}} \xRightarrow{L_{j_0}}$, see \cref{fig:coBuchi61automaton}).
        By equivalence of safe languages, $v$ produces a significant transition during the run $q \xRightarrow{v} q'$.
        Assume by contradiction that $q'$ is in $\Dcore$.
        By definition, $\Dcore$ is a $y$-SCC, therefore, there is path $q' \xrightarrow{v'} q$ containing no $y$.
        We obtain that $(vv')^\omega$ is rejected by $\Dmain$, as we have found a cycle over $vv'$ containing a \significant transition.
        This contradicts the fact that $\Cmain$ only rejects words that contain $y$ infinitely often.
    \end{claimproof}

    Using the above remarks we build a rejecting run of the form
    \[ o_{i_1} \xrightarrow{u_{i_1}} q_{i_1} \xRightarrow{v} o_{i_2} \xrightarrow{u_{i_2}} q_{i_2} \xRightarrow{v} \dots,\]
    where $u_i \in \Fin{\Theju}$ and $v\in L_{j_0}^2$.
    However, this word is accepted by $\Cmain$, because it does not contain a factor of the form $\langInfix{\alpha}{\beta}$, for $\alpha\neq \beta$.
\end{proof}

The above lemma implies that every deterministic coBüchi automaton language-equivalent to $\Cmain$ requires at least 66 states, and thus, we conclude that every deterministic Büchi automaton language-equivalent to $\Amain$ requires at least 66 states. This proves \cref{thm:succinctness}. Since we showed in \cref{thm:A65isHD} that $\Amain$ is HD, we have proved our main result (\cref{thm:main}).

\renewcommand{\thesection}{\arabic{section}}
\setcounter{section}{2}

\section{Conclusions}\label{sec:conclusions}
In this paper, we have solved the succinctness question for history-deterministic Büchi automata by exhibiting a 65-state HD B\"uchi automaton such that every language-equivalent deterministic B\"uchi automaton requires at least 66 states.
The construction and the proof are involved and required computer assistance for some intermediate steps.
In order to have a better understanding of these automata and make progress in further questions, a first step would be to find a simpler example of a succinct history-deterministic Büchi automaton. The existence of such automata is immediately unclear; even the smallest Büchi automaton that we know to be not determinisable by pruning has $7$ states. 

The immediate next question is to settle the width of the gap between deterministic and HD B\"uchi automata. More precisely: is there a family of HD B\"uchi automata such that language-equivalent deterministic B\"uchi automata need at least quadratic number of states? While our result  kills the ``no gap'' conjecture (and rules out approaches like conjecture~\ref{conj:weak-rewiring}), constructing a family with a quadratic gap remains a challenge, especially given the complexity required just to establish a separation of one state. 

We conjecture that our example can be easily used to generate a family of  HD Büchi automata $\Hh_n$ with $65 n$ states whose alphabet size is $16 n$, such that every deterministic Büchi automaton $\Dd_n$ that is language-equivalent to $\Hh_n$ requires $66n$ states. We are doubtful if we would obtain any additional insights by proving this using our example, and thus, we have decided to not look into this for the current paper.

Beyond theory, the techniques we developed to establish lower bounds on deterministic state complexity have practical potential. An interesting future direction would be to integrate these heuristics into SAT-based minimization tools like \texttt{SPOT}~\cite{spot2.1}.
More precisely, given a deterministic B\"uchi automaton we can find a non-trivial lower bound on the size of every language-equivalent deterministic (or HD) automaton by minimising the coB\"uchi automaton for its complement language, following~\cite{AK22} and as described in \cref{subsec:complementation}. Can we further use these minimal HD coBüchi automata to obtain better minimisation algorithms for deterministic coBüchi automata?





\newpage
\bibliography{gfg}
\appendix

\section{Appendix: Encoding of 3-DFAs for DFAMiner}\label{appendix:code}

\subsection{3-DFA for when $q$ is $\circstate{p_2}$}
\begin{lstlisting}

16 5 -- number of states, 
-- alphabet: 0 a, 1 b, 
-- 2 c, 3 1, 4 4 

i 0 -- initial state

t 0 0 1 -- transistion from 0 on 0 to 1
t 0 1 0
t 0 2 2
t 0 3 0
t 0 4 0

t 1 0 1
t 1 1 4
t 1 2 3
t 1 3 0
t 1 4 0

t 2 0 1
t 2 1 0
t 2 2 2
t 2 4 0

t 3 0 1
t 3 1 4
t 3 2 3
t 3 4 0


t 4 0 1
t 4 1 4
t 4 2 3
t 4 3 0



t 5 0 6 
t 5 1 5
t 5 2 7
t 5 3 5
t 5 4 5

t 6 0 6
t 6 1 9
t 6 2 8
t 6 3 5
t 6 4 5

t 7 0 6
t 7 1 5
t 7 2 7
t 7 4 5

t 8 0 6
t 8 1 9
t 8 2 8
t 8 4 5

t 9 0 6
t 9 1 9
t 9 2 8
t 9 3 5

t 10 0 11 
t 10 1 10
t 10 2 12
t 10 3 10
t 10 4 10

t 11 0 11
t 11 1 14
t 11 2 13
t 11 3 10
t 11 4 10

t 12 0 11
t 12 1 10
t 12 2 12
t 12 4 10

t 13 0 11
t 13 1 14
t 13 2 13
t 13 4 10


t 14 0 11
t 14 1 14
t 14 2 13
t 14 3 10



t 2 3 5
t 3 3 5 
t 4 4 10 

t 7 3 15
t 8 3 15
t 9 4 10

t 12 3 5
t 13 3 5
t 14 4 15

--accepting states
a 0 
a 1
a 2
a 3
a 4
--rejecting states
r 15
\end{lstlisting}

\subsection{3DFA for when $q$ is $\circstate{p_5}$}
\begin{lstlisting}
16 5 -- number of states, 
-- alphabet: 0 is a, 1 is b, 
-- 2 is c, 3 is 1, 4 is 4 

i 0 -- initial state

t 0 0 1 -- transistion from 0 on 0 to 1
t 0 1 0
t 0 2 2
t 0 3 0
t 0 4 0

t 1 0 1
t 1 1 4
t 1 2 3
t 1 3 0
t 1 4 0

t 2 0 1
t 2 1 0
t 2 2 2
t 2 4 0

t 3 0 1
t 3 1 4
t 3 2 3
t 3 4 0


t 4 0 1
t 4 1 4
t 4 2 3
t 4 3 0



t 5 0 6 
t 5 1 5
t 5 2 7
t 5 3 5
t 5 4 5

t 6 0 6
t 6 1 9
t 6 2 8
t 6 3 5
t 6 4 5

t 7 0 6
t 7 1 5
t 7 2 7
t 7 4 5

t 8 0 6
t 8 1 9
t 8 2 8
t 8 4 5

t 9 0 6
t 9 1 9
t 9 2 8
t 9 3 5

t 10 0 11 
t 10 1 10
t 10 2 12
t 10 3 10
t 10 4 10

t 11 0 11
t 11 1 14
t 11 2 13
t 11 3 10
t 11 4 10

t 12 0 11
t 12 1 10
t 12 2 12
t 12 4 10

t 13 0 11
t 13 1 14
t 13 2 13
t 13 4 10


t 14 0 11
t 14 1 14
t 14 2 13
t 14 3 10



t 2 3 5
t 3 3 5 
t 4 4 10

t 7 3 15
t 8 3 15
t 9 4 10

t 12 3 5
t 13 3 5
t 14 4 15

--accepting states
a 0 
a 1
a 2
a 3
a 4
--rejecting states
r 15

\end{lstlisting}

\subsection{3-DFA for when $q$ is $\sqrstate{p_1}$}
\begin{lstlisting}
    16 5 -- number of states, 
-- alphabet: 0 is a, 1 is b, 
-- 2 is c, 3 is 1, 4 is 4 

i 0 -- initial state

t 0 0 1 -- transistion from 0 on 0 to 1
t 0 1 0
t 0 2 2
t 0 3 0
t 0 4 0

t 1 0 1
t 1 1 4
t 1 2 3
t 1 3 0
t 1 4 0

t 2 0 1
t 2 1 0
t 2 2 2
t 2 4 0

t 3 0 1
t 3 1 4
t 3 2 3
t 3 4 0


t 4 0 1
t 4 1 4
t 4 2 3
t 4 3 0



t 5 0 6 
t 5 1 5
t 5 2 7
t 5 3 5
t 5 4 5

t 6 0 6
t 6 1 9
t 6 2 8
t 6 3 5
t 6 4 5

t 7 0 6
t 7 1 5
t 7 2 7
t 7 4 5

t 8 0 6
t 8 1 9
t 8 2 8
t 8 4 5

t 9 0 6
t 9 1 9
t 9 2 8
t 9 3 5

t 10 0 11 
t 10 1 10
t 10 2 12
t 10 3 10
t 10 4 10

t 11 0 11
t 11 1 14
t 11 2 13
t 11 3 10
t 11 4 10

t 12 0 11
t 12 1 10
t 12 2 12
t 12 4 10

t 13 0 11
t 13 1 14
t 13 2 13
t 13 4 10


t 14 0 11
t 14 1 14
t 14 2 13
t 14 3 10



t 2 3 15
t 3 3 15
t 4 4 10 

t 7 3 15
t 8 3 15
t 9 4 10

t 12 3 5
t 13 3 5
t 14 4 15

--accepting states
a 0 
a 1
a 2
a 3
a 4
--rejecting states
r 15

\end{lstlisting}
\subsection{3-DFA for when $q$ is $\sqrstate{p_4}$}
\begin{lstlisting}
    16 5 -- number of states, 
-- alphabet: 0 is a, 1 is b, 
-- 2 is c, 3 is 1, 4 is 4 

i 0 -- initial state

t 0 0 1 -- transistion from 0 on 0 to 1
t 0 1 0
t 0 2 2
t 0 3 0
t 0 4 0

t 1 0 1
t 1 1 4
t 1 2 3
t 1 3 0
t 1 4 0

t 2 0 1
t 2 1 0
t 2 2 2
t 2 4 0

t 3 0 1
t 3 1 4
t 3 2 3
t 3 4 0


t 4 0 1
t 4 1 4
t 4 2 3
t 4 3 0



t 5 0 6 
t 5 1 5
t 5 2 7
t 5 3 5
t 5 4 5

t 6 0 6
t 6 1 9
t 6 2 8
t 6 3 5
t 6 4 5

t 7 0 6
t 7 1 5
t 7 2 7
t 7 4 5

t 8 0 6
t 8 1 9
t 8 2 8
t 8 4 5

t 9 0 6
t 9 1 9
t 9 2 8
t 9 3 5

t 10 0 11 
t 10 1 10
t 10 2 12
t 10 3 10
t 10 4 10

t 11 0 11
t 11 1 14
t 11 2 13
t 11 3 10
t 11 4 10

t 12 0 11
t 12 1 10
t 12 2 12
t 12 4 10

t 13 0 11
t 13 1 14
t 13 2 13
t 13 4 10


t 14 0 11
t 14 1 14
t 14 2 13
t 14 3 10



t 2 3 5
t 3 3 5 
t 4 4 15 

t 7 3 15
t 8 3 15
t 9 4 10

t 12 3 5
t 13 3 5
t 14 4 15

--accepting states
a 0 
a 1
a 2
a 3
a 4
--rejecting states
r 15

\end{lstlisting}

\end{document}